\documentclass[11pt,a4paper]{article}
\usepackage{jheppub}
\usepackage{bbm}
\usepackage{youngtab}
\usepackage{rotfloat}
\usepackage{stmaryrd}
\usepackage{amsfonts,amssymb,amsmath,amsthm,multirow}
\usepackage{mathtools}
\usepackage{arydshln}
\usepackage{blkarray, bigstrut}
\usepackage{tikz-cd}
\usepackage{caption}
\usepackage[all,cmtip]{xy}
\usepackage{fancybox}
\usepackage{xcolor}

\newtheorem{lem}{Lemma}
\newtheorem{conj}{Conjecture}

\usepackage{bm} 

\usepackage[T1]{fontenc}

\DeclareSymbolFont{Letters} {U}{zeur}{m}{n}
\DeclareMathSymbol{\talpha}  {\mathalpha}{Letters}{"0B}
\DeclareMathSymbol{\tbeta}   {\mathalpha}{Letters}{"0C}
\DeclareMathSymbol{\tgamma}  {\mathalpha}{Letters}{"0D}
\DeclareMathOperator{\Coker}{Coker}

\newcommand{\bea}{\begin{eqnarray}}
\newcommand{\eea}{\end{eqnarray}}
\newcommand{\be}{\begin{equation}}
\newcommand{\ee}{\end{equation}}

\newcommand{\Z}{{\mathbb Z}}
\newcommand{\R}{{\mathbb R}}
\newcommand{\C}{{\mathbb C}}

\def\Tr{{\rm Tr \,}}

\def\frak{\mathfrak}

\def\tilde{\widetilde}
\def\hat{\widehat}
\def\bar{\overline}

\def\CB{{\mathcal B}}
\def\CC{{\mathcal C}}

\def\CF{{\mathcal F}}

\def\CH{{\mathcal H}}
\def\CI{{\mathcal I}}

\def\CL{{\mathcal L}}
\def\CM{{\mathcal M}}
\def\CN{{\mathcal N}}

\def\CS{{\mathcal S}}
\def\CT{{\mathcal T}}

\def\CW{{\mathcal W}}

\renewcommand{\bar}{\overline}
\renewcommand{\hat}{\widehat}

\DeclareMathAlphabet      {\mathbi}{OML}{cmm}{b}{it}

\definecolor{ao(english)}{rgb}{0.0, 0.5, 0.0}

\preprint{CALT-2019-048}

\title{3d-3d correspondence for mapping tori}

\author[1]{Sungbong Chun}
\author[2,3]{Sergei Gukov}
\author[2]{Sunghyuk Park}
\author[2]{Nikita Sopenko}
\date{\today}
\affiliation[1]{New High Energy Theory Center, Rutgers University, Piscataway, NJ 08854, USA}
\affiliation[2]{California Institute of Technology, Pasadena, CA 91125, USA}
\affiliation[3]{Max-Planck-Institut f\"ur Mathematik, Vivatsgasse 7, D-53111 Bonn, Germany}

\emailAdd{winnety99@gmail.com}
\emailAdd{gukov@theory.caltech.edu}
\emailAdd{spark3@caltech.edu}
\emailAdd{niksopenko@gmail.com}

\abstract{One of the main challenges in 3d-3d correspondence is that no existent approach offers a complete description of 3d $\CN=2$ SCFT $T[M_3]$ --- or, rather, a ``collection of SCFTs'' as we refer to it in the paper --- for all types of 3-manifolds that include, for example, a 3-torus, Brieskorn spheres, and hyperbolic surgeries on knots. The goal of this paper is to overcome this challenge by a more systematic study of 3d-3d correspondence that, first of all, does not rely heavily on any geometric structure on $M_3$ and, secondly, is not limited to a particular supersymmetric partition function of $T[M_3]$. In particular, we propose to describe such ``collection of SCFTs'' in terms of 3d $\CN=2$ gauge theories with ``non-linear matter'' fields valued in complex group manifolds. As a result, we are able to recover familiar 3-manifold invariants, such as Turaev torsion and WRT invariants, from twisted indices and half-indices of $T[M_3]$, and propose new tools to compute more recent $q$-series invariants $\hat Z (M_3)$ in the case of manifolds with $b_1 > 0$. Although we use genus-1 mapping tori as our ``case study,'' many results and techniques readily apply to more general 3-manifolds, as we illustrate throughout the paper.}

\begin{document}
\setcounter{tocdepth}{3}
\maketitle

\section{Introduction and motivation}

3d-3d correspondence, originally proposed in \cite{Dimofte:2010tz}, relates (quantum) topology of 3-manifolds to physics of 3-dimensional supersymmetric gauge theories in various backgrounds.

In particular, to an arbitrary 3-manifold $M_3$ and a choice of ADE type group $G$ it assigns a topological invariant $T[M_3,G]$ valued in 3d $\CN=2$ superconformal field theories. Over the years, it has been shown that many numerical and homological 3-manifold invariants that admit an indepenent mathematical definition factor through $T[M_3,G]$, in a sense that they can be computed as partition functions and spaces of BPS states of the theory $T[M_3,G]$. Among the invariants discussed in \cite{Dimofte:2010tz} are the moduli space of a theory on a circle (with KK modes included) and the so-called ``K-theory version'' of the vortex partition function. Both will be the key players here, albeit the latter will be replaced by its close cousin $\hat Z_a (q) : = \hat Z_{S^1 \times_q D^2} (\CB_a)$. More recently, $T[M_3]$ has been studied in many other backgrounds, with and without additional defects, boundaries, {\it etc.}:
$$
\xymatrixcolsep{4pc}\xymatrix{
	& & ~~~ \text{SW} \\
	\boxed{~{\text{3-manifold} \atop M_3}~} \ar[r] &
	\boxed{~{\text{3d $\CN=2$ theory} \atop T[M_3]}~}
	\ar[ur] \ar[r] \ar[dr] \ar[ddr] & ~~ \CM_{\text{flat}} (M_3, G_{\C}) \\
	& & ~~\vdots \\
	& & ~~~ \hat Z_a (q)}
$$
Here, and mostly throughout this paper, we assume $M_3$ is closed. If the boundary $\partial M_3 \ne \emptyset$, then $T[M_3]$ is, in fact, a boundary condition in 4d $\CN=2$ theory \cite{Dimofte:2011ju,Cecotti:2011iy}. And, similarly, if $M_3$ is the boundary of a 4-manifold, then $T[M_3]$ comes equipped with a 2d $\CN=(0,2)$ boundary condition \cite{Gadde:2013sca}. This leads to far-reaching implications and ensures that the QFT$_3$-valued invariant $T[M_3]$ is functorial.

Among other things, the functoriality requires the vacua of 3d $\CN=2$ theory $T[M_3]$ on a small circle to contain {\it all} flat connections on $M_3$, reducible and irreducible, abelian and non-abelian\footnote{The precise role of $\CM_{\text{vacua}} ( T[M_3,G] )$ will be explained in section \ref{sec:Mflat}; for now the reader can think of it as a target space of the effective two-dimensional theory.}
\be
\CM_{\text{vacua}} \left( T[M_3,G] \right)
\; \supset \;
\CM_{\text{flat}} \left( M_3 , G_{\C} \right)
\label{Mvacua}
\ee
The same conclusion follows from \cite{Cordova:2013cea,Chung:2014qpa} and can also be seen via exchanging the order of compactification on $S^1$ and $M_3$. Of course, on each branch of vacua \eqref{Mvacua} the theory can be simpler than the full theory $T[M_3]$. For example, the analysis in section~\ref{sec:Mflat} suggests that, for $G=SU(2)$, the Coulomb branch theory of $T[M_3,G]$ is determined by its abelian version $T[M_3,U(1)]$, and
\be
\CM_{\text{Coulomb}} \left( T[M_3,SU(2)] \right)
\; = \;
\frac{\C \times \CM_{\text{flat}} \left( M_3 , U(1)_{\C} \right)}{\Z_2}
\label{MCoulomb}
\ee
Unfortunately, no similar proposal for the ``Higgs branch'' of $T[M_3]$ is known at present. In the special case of 3-manifolds with non-empty toral boundaries, the construction \cite{Dimofte:2011ju} offers the best candidate for the Higgs branch of the 3d $\CN=2$ theory $T[M_3]$. This construction does not include abelian flat connections, which are crucial for computing WRT invariants and Floer homology of $M_3$ from $T[M_3]$. Mistaking this Higgs branch theory for $T[M_3]$ would be analogous to mistaking the Higgs branch theory \cite{Moore:2011ee} of ``class S theory'' for the class S theory itself \cite{Gaiotto:2008cd,Gaiotto:2009we}. Yet, in the literature on 3d-3d correspondence, this distinction is sometimes overlooked and $T[M_3]$ is (incorrectly) replaced by its Higgs branch theory.

One goal of the present paper is to correct this fallacy. In fact, abelian flat connections \eqref{MCoulomb} will play a crucial role in all aspects of the story. For example, they will enter our study of new $q$-series invariants $\hat Z_a (M_3;q)$ that, at roots of unity, are supposed to be related to more traditional Witten-Reshetikhin-Turaev (WRT) invariants \cite{witten1989quantum,reshetikhin1991invariants}. How these invariants, old and new, are encoded in the data of 3d $\CN=2$ theory $T[M_3]$ is by now fairly well understood for 3-manifolds with $b_1 = 0$ ({\it i.e.} for rational homology spheres), and one of our main goals is to explore the novelties of 3d-3d correspondence for 3-manifolds with $b_1 > 0$. A new feature of 3-manifolds with $b_1 > 0$ is that their moduli spaces of flat connections have positive dimension.
The simplest class of such manifolds consists of mapping tori, obtained by identifying two boundaries of the mapping cylinder $\Sigma \times I$ via an orientation-preserving homeomorphism $\varphi \in \text{MCG} (\Sigma)$. Specifically,
\be
M_3 \; = \; \Sigma \times [0,1] / \sim \,,
\qquad (x,0) \sim (\varphi (x),1)
\label{mappingtori}
\ee
where, in general, $\Sigma$ can be any orientable surface, possibly with punctures. It is well known \cite{MR3109869} that even among mapping tori of genus 1 --- with $\Sigma = T^2$ and $M_3$ labeled by a choice of $\varphi \in SL(2,\Z)$ --- there are examples of pairs, $M_3$ and $M_3'$, that can not be distinguished by WRT invariants (see also \cite{MR1204405}). It therefore raises a natural questions whether such pairs can be distinguished by the $q$-series invariants $\hat Z_a$. Our results indicate that the answer to this question may, surprisingly, be ``yes'' in a rather interesting way that involves abelian flat connections and labels ``$a$'' of $\hat Z_a$ as key ingredients.\footnote{In particular, we expect theories $T[M_3^+]$ and $T[M_3^-]$ to be different, where $M_3^{\pm}$ are mapping tori with monodromies $\pm \varphi$.}

Another motivation and another aspect of 3d-3d correspondence where abelian flat connections are unavoidable has to do with twisted indices and the Heegaard Floer homology $HF^+ (M_3)$. It makes a subtle and surprising appearance in the study of $q$-series invariants $\hat Z_a (q)$, for reasons which are not completely clear at present. One possible explanation could involve a relation between $HF^+ (M_3)$ and homology groups categorifying $\hat Z_a (q)$, similar to the relation between the corresponding knot homologies, namely the knot Floer homology and the Khovanov homology. We hope that our study, focused on 3-manifolds with $b_1 > 0$, will serve as a useful step in future developments and better understanding of such questions.

An added benefit of discussing in parallel rather different looking invariants $\hat Z_a (M_3)$ and $HF^+ (M_3)$ in the context of 3d-3d correspondence is that it sheds light on the following important question: What does a 4d TQFT categorifying $\hat Z_a (M_3)$ associate to $\Sigma$? By general rules of extended TQFT it should be a category, in which 3-manifolds bounded by $\Sigma$ define objects, and Hom spaces correspond to gluing along $\Sigma$. So, the question really is: Concretely, what is this category associated to $\Sigma$? This is precisely where mapping tori \eqref{mappingtori} are helpful. The category in question has $\text{MCG} (\Sigma)$ as a group of autoequivalences, {\it cf.} \cite{Gukov:2007ck,MR3319701}, and the trace functor associated with $\varphi$ should return $Q$-cohomology of $T[M_3]$ on $\R \times D^2$. When $\varphi = {\bf 1}$, this trace (decategorification) map is simply the Grothendieck group.

The paper is organized as follows. We start in section~\ref{sec:Mflat} by analyzing more carefully what role the moduli space of $G_{\C}$ flat connections on $M_3$ plays in 3d $\CN=2$ theory $T[M_3,G]$. In particular, we ask whether quantum corrections can modify the semi-classical picture \eqref{MCoulomb} by lifting some of the vacua, and how the vacua of 3d theory (not compactified on a circle) differ from \eqref{MCoulomb}. The first approximation to the answer can be gleaned from the fact that, when $T[M_3,G]$ is compactified on a circle of a finite but small radius, the effective theory is two-dimensional and explores all of its vacua due to long-range fluctuation. Therefore, $\CM_{\text{vacua}}$ is better viewed as a ``target space'' than the ``space of vacua.'' More interestingly, for $G=U(1)$ and also in many cases for $G=SU(2)$, we find that $\CM_{\text{vacua}}$ has the interpretation directly in three dimensions, as a target space of 3d $\CN=2$ theory $T[M_3,G]$ itself, with non-trivial interacting SCFTs residing at singularities of $\CM_{\text{vacua}}$. A better and more accurate version of this 3d interpretation is that $T[M_3,G]$ should be viewed not as a single SCFT, but rather as a collection of SCFTs parametrized by points of the moduli space of $G_{\C}$ flat connections on $M_3$. We refer to this collection of SCFTs as a {\it ``sheaf'' of SCFTs}.

Then in section~\ref{sec:Zhat} we explore $\hat{Z}_b(q)$ invariants for some class of 3-manifolds with positive first betti number, namely plumbings with loops and 0-surgery on knots. One important lesson we learn is that in order to reproduce WRT invariants for such three-manifolds it is not enough to use abelian $\hat{Z}_b$ that were introduced in \cite{Gukov:2018iiq}, but rather, we need more blocks that we call \emph{almost abelian}. Moreover for 0-surgery on knots we observe a nice interplay between these $\hat{Z}_b$ and $F_K$, a two-variables series recently introduced in \cite{Gukov:2019mnk}. 

In section~\ref{sec:HF} we discuss twisted indices of $T[M_3]$ on different manifolds and their relation to known topological invariants such as Turaev torsion and Heegaard Floer homology. 
We conclude in secton~\ref{sec:discuss} with a brief summary of possible generalizations and future directions.

\begin{figure}[ht]
	\centering
	\includegraphics[width=2.7in]{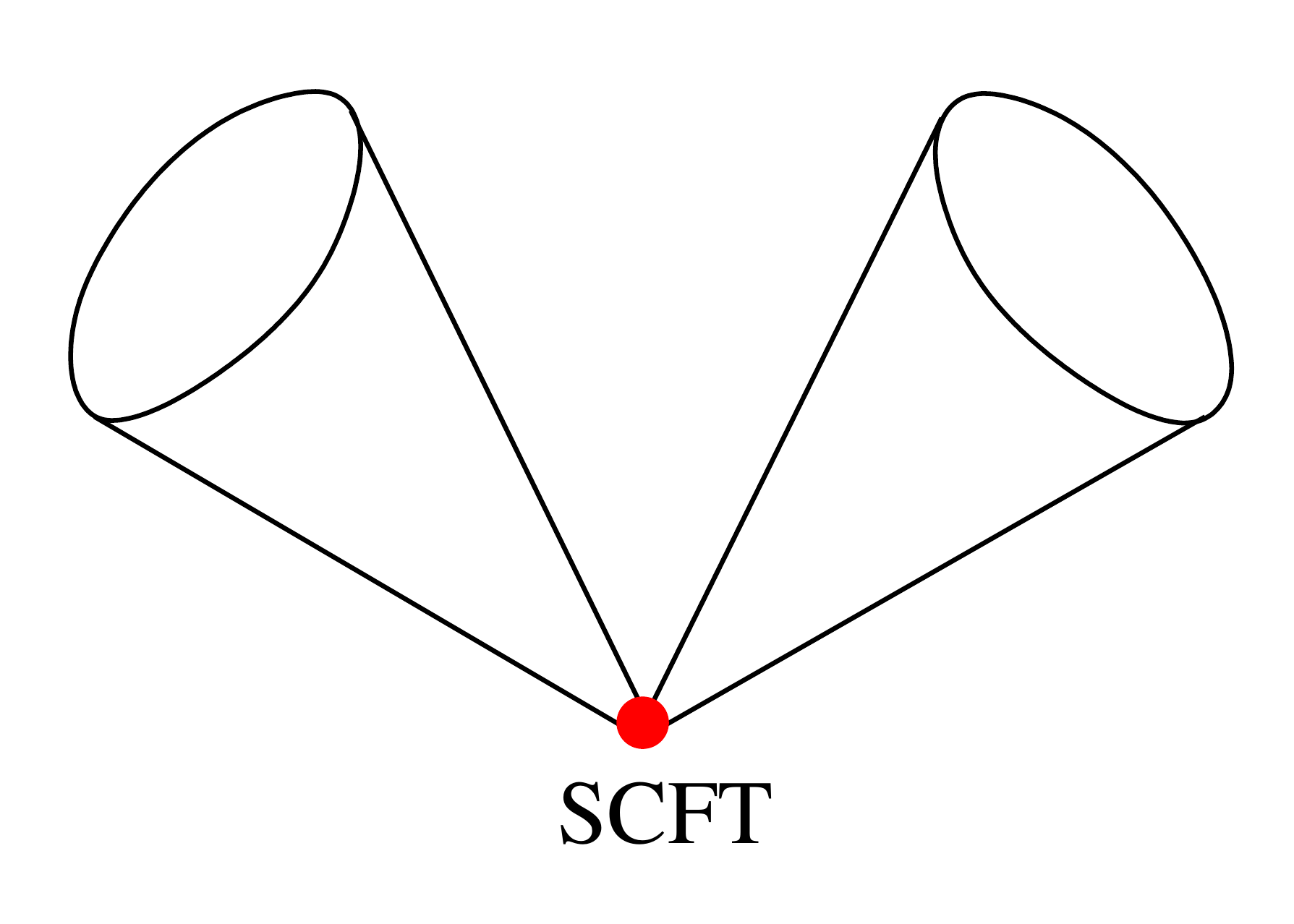}
	\caption{A cartoon illustrating different branches of vacua in a SCFT. Aside from the original SCFT, one can also speak of a theory on each given branch.}
	\label{fig:branches}
\end{figure}

\section{``Coulomb'' and ``Higgs'' branches}
\label{sec:Mflat}

In standard terminology \cite{Intriligator:1995au}, a Higgs branch usually refers to vacua parametrized by scalar fields in matter multiplets where gauge group is completely broken. In contrast, a Coulomb branch is parametrized by scalars in vector multiplets which leave a non-trivial abelian part of the gauge symmetry unbroken. A natural adaptation of this terminology to 3d-3d correspondence, at least for $G$ of rank 1, would mean that branches of the moduli space \eqref{Mvacua} where flat connections have a stabilizer of positive dimension ({\it i.e.} reducible flat connections) play the role of Coulomb branches, while irreducible flat connections should be viewed as analogues of Higgs branches. This terminology is consistent with the uses of ``Coulomb'' and ``Higgs'' in theories with small amounts of supersymmetry, see {\it e.g.} \cite{Aharony:1997bx} for the context of 3d $\CN=2$ physics relevant to us here.

In general, the space of flat connections on any manifold can be described by homomorphisms from its fundamental group to the gauge group, modulo conjugations. In particular,
\be
\CM_{\text{flat}} \left( M_3 , G_{\C} \right) \; = \; \text{Hom} \big( \pi_1 (M_3), G_{\C} \big) \, / \, \text{conj.}
\label{Mflat}
\ee

\subsection{$T[M_3]$ as a ``sheaf'' of SCFTs}

In string theory and in quantum field theory, one of the standard tools is a compactification on a circle $S^1$ or, more generally, on a manifold $M_n$ of dimension $n$. This tool is very useful because it relates the physics of theories in dimension $d$ and $d-n$.
In particular, for $d-n>2$, field configurations of a $d$-dimensional theory that satisfy the required equations of motion along $M_n$ become classical vacua of the effective $(d-n)$-dimensional theory.

With sufficiently large amount of supersymmetry, one can be sure that all such vacua are present in the quantum theory as well, and each choice of vacuum, {\it i.e.} each choice of the background on $M_n$, then flows to a SCFT$_{d-n} (v)$ parametrized by $v \in \CM_{\text{vacua}}$. Depending on the context, $\CM_{\text{vacua}}$ can be either a discrete set of points ({\it e.g.} different flux sectors) or a continuum (parametrized by moduli), or a union of components of both types.

Typically, generic points on the continuum part of $\CM_{\text{vacua}}$ correspond to relatively simple SCFTs, which become more interesting interacting SCFTs at special loci of this moduli space. {\it A priori}, all such SCFTs can be quite different and each SCFT only captures the physics of its local neighborhood on $\CM_{\text{vacua}}$. In order to describe the global structure of the effective $(d-n)$-dimensional theory one therefore needs a collection of SCFTs glued together in a way reminiscent of a sheaf. In the physical problem at hand, the ``stalk'' at a given point $v \in \CM_{\text{vacua}}$ provide a local description that is captured by the SCFT at this point, but its global structure describes the way these local patches are glued together.\footnote{Note, that the data of a SCFT at a given point $v \in \CM_{\text{vacua}}$ is encoded in an algebra, namely its OPE algebra. It also knows about the SCFTs in the neighborhood of  $v$, but cannot reach points far in the space of vacua $\CM_{\text{vacua}}$; for this one needs to use several charts which agree on overlaps. Although this structure is reminiscent of the mathematical notion of a sheaf, the analogy is not quite precise since in the present physical setup it is not clear what algebraic structure is assigned to a general open set on $\CM_{\text{vacua}}$. For this reason, we use the word ``sheaf'' in quotes when applied to the collection of SCFTs parametrized by points of $\CM_{\text{vacua}}$ and glued together.}

This way of looking at the physics of the partially compactified theory is, in fact, rather common, although perhaps not phrased in exactly this language. Consider for example, successive compactifications of the six-dimensional $(2,0)$ theory on circles. At first step, one finds a maximally supersymmetric Yang-Mills theory, with a dimensionful gauge coupling constant determined by the size of the circle. Then, a further compactification on another $S^1$ gives a maximally supersymmetric ($\CN=4$) theory in four dimensions. Unless the size of the circle(s) is taken to zero, one of the scalar fields in the effective 4d theory is $G$-valued, parametrized by holonomy of the 5d gauge field $A$. (The other five scalars are $\text{Lie} (G)$-valued, as in the original 5d theory.) As a result, the moduli space of the effective 4d theory is $\C^* \times \C^2$ when $G=U(1)$ and
\be
\frac{\C^* \times \C^2}{\Z_2}
\label{Mvac4d}
\ee
when $G=SU(2)$. The neighborhood of each orbifold point here is well described by 4d $\CN=4$ super-Yang-Mills (SYM) theory which has the moduli space of vacua $\C^3 / \Z_2$. The latter, clearly, does not capture the low-energy physics at arbitrary points of \eqref{Mvac4d} and one needs two copies of 4d $\CN=4$ SYM to describe the effective physics of 5d theory compactified on a circle of a finite size.

Similarly, the next step, {\it i.e.} a further compactification of the 6d $(2,0)$ theory on a third circle of small but finite size gives a collection of SCFTs with very rich BPS spectra. By definition, this data is $T[M_3,G]$ with $M_3 = T^3$. When $G=U(1)$, the moduli space of this three-dimensional theory (in flat space) is $\C \times (\C^*)^3$, where two copies of $\C^*$ can be understood as holonomies of 5d gauge field along the two cycles of $T^2 = S^1 \times S^1$ paired up with 5d scalars, whereas the third copy of $\C^*$ is parametrized by the dual 3d photon (which is also a periodic scalar in three dimensions) paired with a non-compact scalar \cite{Aharony:1997bx}. This field content can be equivalently summarized in a form of four 3d $\CN=2$ chiral multiplets, three of which are $\C^*$-valued and one is $\C$-valued. This is the same as the moduli space of vacua --- or, rather, the target space --- of the 3d theory on a circle, {\it cf.} \eqref{Mabelian}.

As we explain later in this section, in the non-abelian case too, namely for $G=SU(2)$, the moduli space of 3d theory in flat space is the same as, {\it cf.} \eqref{MCoulomb},
\be
\CM_{\text{vacua}} \; = \; \frac{(\C^*)^3 \times \C}{\Z_2}
\label{T3vacua}
\ee
It has 8 singular point, each described by a copy of 3d $\CN=8$ ABJM theory~\cite{Aharony:2008ug}. Again, each ABJM theory only provides a local description of its local neighborhood, modeled on a space of vacua~$\C^4/\Z_2$. If we restrict our attention to only one copy of ABJM theory, we loose the fact that there are additional singular points. In order to recover them all, we need to patch eight copies of the ABJM theory into a ``sheaf'' of SCFTs~$T[M_3]$.

This example of $T[M_3]$ with $M_3 = T^3$ will be a useful prototype for what to expect in the case of general 3-manifolds. In particular, our next goal is to show that $\CM_{\text{vacua}} \left( T[M_3,U(1)] \right)$ can be always understood as a space of vacua of 3d theory in flat space ({\it i.e.} on $\R^3$), for any $M_3$.

\begin{figure}[ht]
	\centering
	\includegraphics[width=3.0in]{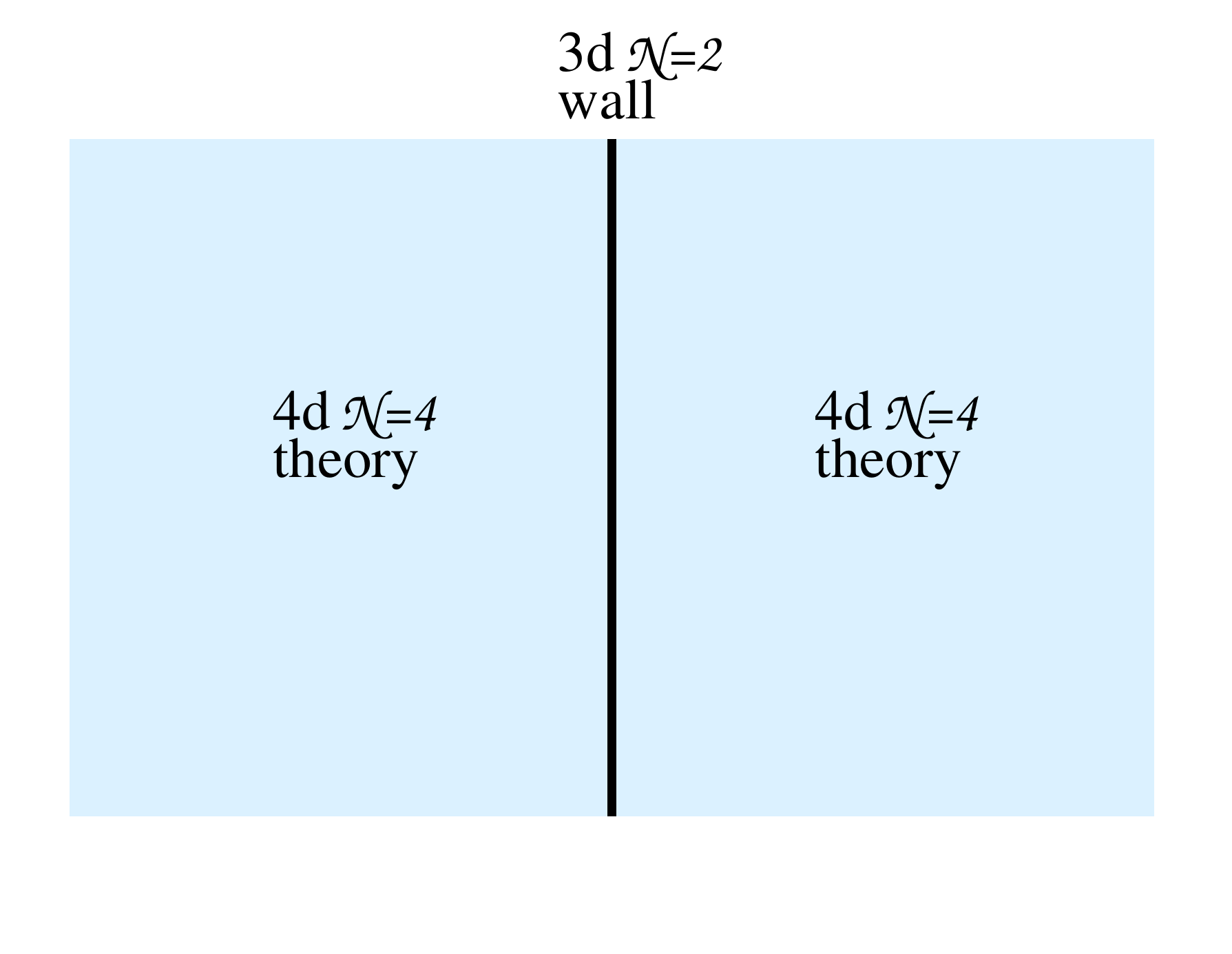}
	\caption{For genus-1 mapping tori, $T[M_3]$ can be constructed from $\frac{1}{4}$-BPS walls in 4d $\CN=4$ theory. Note, the 4d $\CN=4$ theory in question is 5d super-Yang-Mills on a circle of {\it finite} radius, and 3d walls preserve only 4 out of 16 supercharges, {\it i.e.} $\CN=2$ supersymmetry in three dimensions.}
	\label{fig:BPSwall}
\end{figure}

\subsection{Moduli spaces for $G=U(1)$}
\label{sec:abelian}

Abelian version of 3d-3d correspondence is a useful guide to the general case. In particular, it can often shed light on what algebraic structures in 3d $\CN=2$ theory $T[M_3,G]$ one should expect for different constructions of 3-manifolds, how much supersymmetry the basic ingredients preserve, {\it etc.}

For $G=U(1)$, the theory $T[M_3]$ is defined as the effective 3d theory on a single five-brane, partially twisted along the 3-manifold $M_3$. Since in this case the 6d theory on the five-brane world-volume is Lagrangian --- namely, it is a theory of a free 6d $(0,2)$ tensor multilet --- the resulting 3d theory $T[M_3]$ can be simply obtained by the usual rules of the Kaluza-Klein (KK) reduction.

In particular, due to the partial topological twist along $M_3$, three out of five real scalars in 6d theory turn into a 1-form on $M_3$; its KK modes contain $b_1$ scalars and become supersymmetric partners of $b_1$ vector fields which are KK modes of the tensor field in six dimensions. The other two scalars of the 6d theory can be combined into one complex scalar. It has only one obvious KK mode which is a complex scalar field $\phi_0$ in a 3d chiral multiplet that we denote $\Phi_0$. The field $\phi_0$ has a simple geometric interpretation; it describes the displacement of the five-brane along $\R^2 \subset \R^5$ transverse to the five-brane world-volume $\R^3 \times M_3 \subset \R^5 \times T^* M_3$.

To summarize, we conclude that the spectrum of light dynamical fields in the 3d $\CN=2$ theory $T[M_3]$ includes $b_1$ abelian vector multiplets (without Chern-Simons couplings) and a free chiral multiplet $\Phi_0$. When 3d theory is considered on a circle, {\it i.e.} on $S^1 \times \R^2$, these fields contribute factors $(\C^*)^{b_1}$ and $\C$ to the space of vacua \eqref{Mvacua}. Or, directly in 3d, one can dualize $U(1)$ gauge fields into compact ({\it i.e.} periodic, circle-valued) scalars via $dA_i = * d \phi_i$. One then finds $b_1$ copies of the dual photon multiplet, each of which is a 3d $\CN=2$ chiral multiplet with $\C^*$-valued complex scalar \cite{Aharony:1997bx}.

Another, less obvious factor in $\CM_{\text{vacua}}$ is the set of discrete vacua that comes from the torsion part of $H_1 (M_3,\Z)$. One way to see it is to note that, for $G = U(1)$, all representations of $\pi_1 (M_3)$ into $G_{\C}$ factor through its abelianization, $H_1 (M_3, \Z)$. Therefore, in this case \eqref{Mflat} is the complexification of the Pontryagin dual of $H_1 (M_3, \Z)$,
\be
\CM_{\text{flat}} \left( M_3 , \C^* \right) \; = \;
\left( \C^* \right)^{b_1} \, \times \, \text{Tor} H_1 (M_3,\Z)^{\vee}
\label{Mflatabelian}
\ee
Note, that it properly accounts for all vacua of $T[M_3]$, except the ones parametrized by vevs of $\Phi_0$ (which, of course, is not too surprising since $\Phi_0$ is not affected by the topological twist and represents a different sector of the theory, independent of complex $G_{\C}$ connections). Taking into account $\Phi_0$, we get the complete moduli space
\be
\CM_{\text{vacua}} \; = \; \C \, \times \, \left( \C^* \right)^{b_1} \, \times \, \text{Tor} H_1 (M_3,\Z)^{\vee}
\label{Mabelian}
\ee
Since the fields parametrizing various components are independent, the resulting set of vacua is simply a product. The derivation of \eqref{Mabelian} and the final result are completely general; they apply to an arbitrary closed oriented 3-manifold $M_3$. In particular, they illustrate the assertion made earlier, namely that $\CM_{\text{vacua}}$ can be viewed as the space of vacua of 3d theory on a circle as well as in flat space-time. To be more precise, these two moduli spaces are related by T-duality or mirror symmetry, which relates dual photons $\phi_i$ and holonomies of the abelian gauge fields $A_i$ on a circle.\footnote{See {\it e.g.} \cite{Karch:2018mer}.} It also helps to anticipate the subtleties of such relation in the non-abelian case.

Now, let us specialize to the case of genus-1 mapping tori and see how \eqref{Mabelian} can be reproduced from the structure of a torus bundle. First, we can reduce 5d maximally supersymmetric gauge theory with gauge group $G$ on a torus $\Sigma = T^2$, the fiber of $M_3$. We obtain a maximally supersymmetric 3d gauge theory with the same gauge group, where two scalar fields are compact, parametrized by $G$-valued holonomies along the generators of $H_1 (\Sigma) \cong \Z^2$. They naturally combine with two non-compact, $\text{Lie} (G)$-valued scalars to form a bosonic content of 3d $\CN=4$ hypermultiplet with values in
\be
\C^* \times \C^* \; \cong \; \CM_{\text{flat}} \left( \Sigma , G_{\C} \right)
\label{Tabelian}
\ee
on which $\text{MCG} (\Sigma) = SL(2,\Z)$ acts in an obvious way. The rest of the fields comprise 3d $\CN=4$ vector multiplet.

In order to make contact with the space of vacua \eqref{Mabelian}, we need to compactify this 3d $\CN=4$ theory further on a circle $S^1$ with a monodromy (``duality wall'') by the element $\varphi \in \text{MCG} (\Sigma)$. The reduction of the 3d $\CN=4$ vector multiplet contributes to the moduli space of vacua a factor $\C \times \C^*$, independent of the choice of $\varphi$. On the other hand, the target space of the hypermultiplet is projected onto the fixed point set of the map $\varphi$ acting on \eqref{Tabelian}.

This is in perfect agreement with \eqref{Mabelian}. Indeed, the homology $H_1 (M_3)$ has three obvious generators, namely the generators of $H_1 (\Sigma) \cong \Z^2$ and $H_1 (S^1) \cong \Z$, with a matrix of relations given by $\varphi - {\bf 1}$. They fit into the following long exact sequence (see {\it e.g.} \cite{MR1867354}):
\be
\ldots \longrightarrow H_n(\Sigma) \xrightarrow{~\varphi_{*}-{\bf 1}~} H_n(\Sigma) \longrightarrow H_n(M_3) \longrightarrow H_{n-1}(\Sigma) \xrightarrow{~\varphi_{*}-{\bf 1}~} \ldots
\label{Hexsequence}
\ee
from which it follows that, for genus-1 mapping tori,
\be
H_1(M_3) = \mathbb{Z} \oplus \text{coker}(\varphi_*-{\bf 1}) \quad \Rightarrow \quad b_1(M_3) = \begin{cases} 1 &\text{if} \ \ \text{tr}(\varphi) \neq 2 \\ 2 &\text{if} \ \  \text{tr}(\varphi) = 2, \ \varphi \neq {\bf 1} \\ 3 &\text{if} \ \ \varphi = {\bf 1}. \end{cases}
\ee
In other words, there are three possible cases:
\begin{itemize}

\item $b_1 = 1$. This is the generic case.

\item $b_1 = 2$. In this case, $\varphi \in SL(2,\Z)$ is conjugate to $T^p = \bigl(\begin{smallmatrix} 1 & p \\ 0 & 1 \end{smallmatrix} \bigr)$, with $p \ne 0$, and $M_3$ can be also viewed as a degree-$p$ circle bundle over $T^2$. In latter presentation, its homology $H_1 (M_3) = \Z \oplus \Z \oplus \Z_p$ can be also computed via the Leray-Serre spectral sequence.

\item $b_1 = 3$. In this very special case $M_3 = T^3$.

\end{itemize}

\noindent
Here we used the standard notations for generators of $SL(2,\Z)$, which satisfy $S^2 = - {\bf 1} = (ST)^3$ and can be represented by matrices\footnote{Hopefully, this will not cause a confusion between elements $T^p = \bigl(\begin{smallmatrix} 1 & p \\ 0 & 1 \end{smallmatrix} \bigr)$ and $p$-dimensional tori, denoted in a similar way.}
\be
S \; = \; \begin{pmatrix}0 & -1 \\ 1 & 0\end{pmatrix}
\quad,\quad
T \; = \; \begin{pmatrix}1 & 1 \\ 0 & 1\end{pmatrix}
\label{SandT}
\ee
In what follows, we also denote by $U$ a $2 \times 2$ matrix representing $\varphi \in SL(2,\Z)$. It can always be expressed as a word in $S$ and $T$:
\be
U \; = \; S T^{a_1} S T^{a_2} \ldots S T^{a_n}
\label{UviaST}
\ee
for some $a_1, \ldots, a_n \in \Z$.
Using the following simple rule
\be\label{mappingtoridict} \begin{array}{cl@{\qquad}l}
	T^a & \quad = \quad \overset{\displaystyle{a}}{\bullet}
	& (\text{vertex with framing}~a) \\
	S & \quad = \quad \frac{\phantom{xxx}}{\phantom{xxx}}
	& (\text{edge})
\end{array}
\ee
we can graphically represent a genus-1 mapping torus associated with \eqref{UviaST} by a {\it plumbing graph}:
\be
	\includegraphics[scale=1.2]{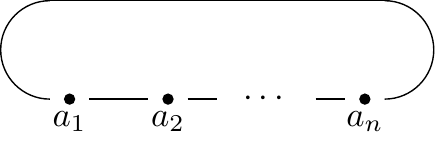}
	\label{aaaloop}
\ee
The meaning of such graphs will be explained in section~\ref{sec:plumbings}, where we also consider generalizations. For now, we only note that, in terms of the linking matrix of this graph,
\be
Q \; = \;
\begin{pmatrix}
	a_1 & -1   & 0 & \cdots  & -1 \\
	-1   & a_2 & -1        &   & \vdots \\
	0 & -1 &  & \ddots  &   0 \\
	\vdots  & & \ddots & ~\ddots~  &  ~-1~  \\
	-1 &  \cdots  & 0 & -1 & a_n
\end{pmatrix}
\label{Qmatrix}
\ee
we have
\be
H_1 (M_3,\Z) \; = \; \Z \times \Z^n / Q \Z^n
\label{HmaptoriviaQ}
\ee
In particular, when $Q$ is nondegenerate, $\text{Tor} H_1(M_3,\Z)=\Z^n/Q\Z^n$.

Comparing \eqref{HmaptoriviaQ} with our discussion around \eqref{Mabelian}, we obtain a concrete description of the 3d $\CN=2$ theory $T[M_3, U(1)]$:
\be
T[M_3, U(1)]
\quad = \quad
\boxed{\begin{array}{c}
		~U(1)^{n+1} \text{ gauge theory with a matrix of Chern-Simons} \\
		~\text{levels } Q \text{ and a free chiral multiplet } \Phi_0
\end{array}}
\label{quiverCS}
\ee
Another useful description, that follows from $\Z^n/Q\Z^n \cong \Z^2 / (U - {\bf 1}) \Z^2$, is
\be
T[M_3, U(1)]
\quad = \quad
\boxed{\begin{array}{c}
		~U(1)^{3} \text{ gauge theory with a matrix of Chern-Simons} \\
		~\text{levels } U-{\bf 1} \text{ and a free chiral multiplet } \Phi_0
\end{array}}
\label{UquiverCS}
\ee
According to \cite{Gukov:2015sna,Pei:2015jsa}, the R-charge of the chiral multiplet $\Phi_0$ is
\be
R (\Phi_0) \; = \; 2
\label{RPhi0}
\ee

\subsection{Moduli spaces for $G=SU(2)$}
\label{sec:sutwomoduli}

In this case, the representations of $\pi_1 (M_3)$ into $G_{\C}$ no longer factor through $H_1 (M_3, \Z)$, and the suitable analogue of \eqref{Hexsequence} is the following exact sequence:
\be
1 \longrightarrow \pi_1 (\Sigma) \longrightarrow \pi_1 (M_3) \longrightarrow \Z \longrightarrow 1
\ee

\begin{figure}[ht]
	\centering
	\includegraphics[width=3.7in]{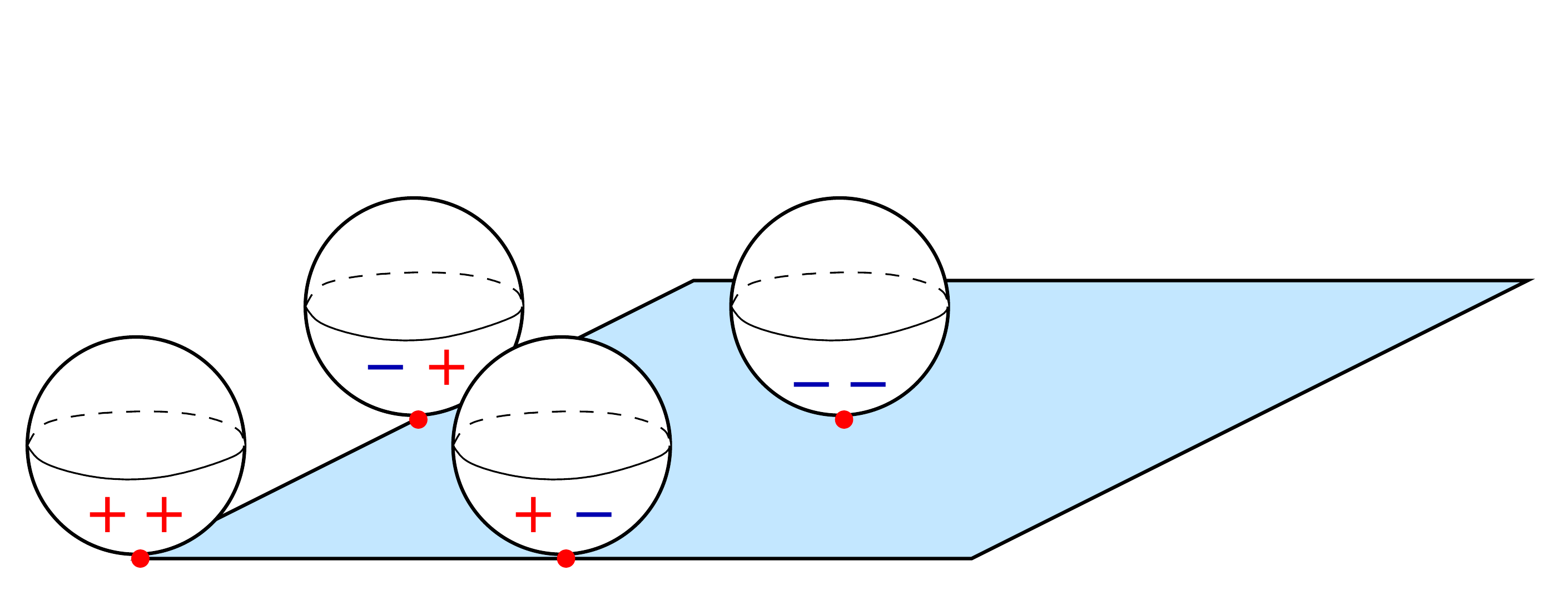}
	\caption{The structure of the moduli space of flat $G_{\C} = SL(2,\C)$ flat connections on $\Sigma = T^2$. Over each of the four orbifold points of the ``semisimple'' branch \eqref{MTsutwo}, with $x = \pm 1$ and $y = \pm 1$, one finds the corresponding ``unipotent'' branch.}
	\label{fig:T2moduli}
\end{figure}

The moduli space of $G_{\C} = SL(2,\C)$ flat connections on $\Sigma = T^2$ consists of several components: the main ``semisimple'' component and four ``unipotent'' components $D_i$, $i=1,\ldots,4$. The semisimple component is simply a quotient of \eqref{Tabelian}:
\be
\frac{\C^* \times \C^*}{\Z_2}
\label{MTsutwo}
\ee
by the Weyl group $W = \Z_2$, which acts on $(x,y) \in \C^* \times \C^*$ via $(x,y) \mapsto (x^{-1},y^{-1})$. Here, $x$ and $y$ are to be understood as holonomy eigenvalues of the $SL(2,\C)$ flat connection along $A$ and $B$ cycles of $\Sigma = T^2$. When $x = \pm 1$ and $y = \pm 1$ (with independent signs), the quotient space \eqref{MTsutwo} is singular. It has four orbifold points of the type $\C^2 / Z_2$, {\it i.e.} $A_1$ Kleinian singularities. Over each of these orbifold points, a new branch of complex flat connections opens up, as illustrated in Figure~\ref{fig:T2moduli}.

Namely, when $x = \pm 1$ and $y = \pm 1$, one can deform the $SL(2,\C)$ holonomies $A = \pm {\bf 1}$ and $B = \pm {\bf 1}$ into an upper-triangular form
\be
A \; = \;
\begin{pmatrix}
	\pm 1 & u \\
	0 & \pm 1
\end{pmatrix}
\qquad,\qquad
B \; = \;
\begin{pmatrix}
	\pm 1 & v \\
	0 & \pm 1
\end{pmatrix}
\label{ABunipotent}
\ee
which we shall call ``unipotent'' whenever $|u|^2 + |v|^2 \ne 0$. Since the two commuting holonomies $A$ and $B$ are defined only up to conjugation by $SL(2,\C)$ matrices, it is easy to see that the two complex off-diagonal values $u$ and $v$ are defined only up to simultaneous multiplication by a non-zero complex number, $(u,v) \sim (\lambda u, \lambda v)$ with $\lambda \in \C^*$. Therefore, for each choice of (independent) signs in \eqref{ABunipotent}, the corresponding unipotent branch is
\be
D_i \; = \; \frac{\C^2 - \{ 0 \} }{\C^*}
\; \cong \; {\mathbb{C}}{\mathbf{P}}^1
\qquad\qquad i = 1, \ldots, 4
\ee
To summarize, the moduli space of flat $G_{\C} = SL(2,\C)$ flat connections on $\Sigma = T^2$ has the following structure:
\be
\CM_{\text{flat}} \left( \Sigma , G_{\C} \right) \; \cong \; \frac{\C^* \times \C^*}{\Z_2}
\, \cup \, \bigcup_{i=1}^{4} D_i
\label{MTsutwototal}
\ee
Note, the unipotent components $D_i$ are \emph{not} disjoint from the semisimple component. This is clear before modding out by the $G_{\C}$ symmetry, when the two branches simply touch each other as Coulomb and Higgs branches often do. And, even after the conjugation by $G_{\C}$, the four components $D_i$ can also be interpreted as exceptional divisors associated with resolution of four $\C^2 / \Z_2$ singularities.

The moduli space of $G_{\C} = SL(2,\C)$ flat connections on the mapping cylinder $\Sigma \times I$ of the element $U = \bigl(\begin{smallmatrix} a & b \\ c & d \end{smallmatrix} \bigr) \in SL(2,\Z)$ is the correspondence
\be
\rho(U) \; : \qquad
\CM_{\text{flat}} \left( \Sigma , G_{\C} \right)
\; \to \;
\CM_{\text{flat}} \left( \Sigma , G_{\C} \right)
\ee
which acts on the main (semisimple) component \eqref{MTsutwo} in the obvious way
\be
\rho(U) \; : \qquad
\begin{array}{rcl}
	\left( \C^* \times \C^* \right) / \Z_2
	& \to & \left( \C^* \times \C^* \right) / \Z_2 \\
	 (x,y) & \mapsto & (x^a y^b , x^c y^d)
\end{array}
\label{UonCC}
\ee
Therefore, apart from the holonomy around the base circle, the moduli space of flat $G_{\C}$ connections on the mapping torus \eqref{mappingtori} is given by the intersection of the graph of function $\rho(U)$ with the diagonal $\Delta : \CM_{\text{flat}} \left( \Sigma , G_{\C} \right) \to \CM_{\text{flat}} \left( \Sigma , G_{\C} \right)$. These are the fixed points of the map $\rho(U)$. In particular, according to \eqref{UonCC}, the moduli space of flat $G_{\C}$ connections on the mapping torus coming from $\left( \C^* \times \C^* \right) / \Z_2 \subset \CM_{\text{flat}} \left( \Sigma , G_{\C} \right)$ has the form
\be
\frac{\CM_+ \times \C^*}{\Z_2} \, \bigcup \, \frac{\CM_-}{\Z_2}
\quad \subset \quad
\CM_{\text{flat}} \left( M_3 , G_{\C} \right)
\label{sspartofM}
\ee
where $\CM_{\pm}$ is defined to be the fixed point set of the simpler version of \eqref{UonCC}:
\be
\rho(U)_{\pm} \; : \qquad
\begin{array}{rcl}
	\C^* \times \C^*
	& \to & \C^* \times \C^* \\
	(x,y) & \mapsto & (x^{\pm a} y^{\pm b} , x^{\pm c} y^{\pm d})
\end{array}
\ee
Note, that $\CM_+ \times \C^*$ is simply the moduli space of complex abelian flat connections \eqref{Mflatabelian} discussed earlier, and $\CM_- $ is its close cousin comprised of flat connections twisted by the generator of the Weyl group
\be
\CM_{\pm} \; = \;
\left \{ (x,y) \in \C^* \times \C^* \left|
\begin{array}{rcl}
x^{\pm a - 1} y^{\pm b} & = & 1 \\
x^{\pm c} y^{\pm d - 1} & = & 1
\end{array}
\right.
\right \}
\label{eigeneqs}
\ee
Therefore, the contribution of the first component $\CM_+ \times \C^*$ to the moduli space of vacua \eqref{Mvacua} correspond to the ``Coulomb'' branch of the 3d $\CN=2$ theory $T[M_3]$ on a small circle \eqref{MCoulomb}, while the points of $\CM_- $ correspond to Higgs branch vacua. Note, the latter are ``almost abelian''; if not for the twist by the non-trivial Weyl group element, they would represent reducible complex flat connections on $M_3$.

The action of the mapping class group $\text{MCG} (\Sigma) \cong SL(2,\Z)$ on the other components $D_i$ in \eqref{MTsutwototal} parametrized by ``unipotent'' holonomies is slightly more involved and depends on the choice of signs in \eqref{ABunipotent}. We summarize the result of this action in Table~\ref{tab:unipotent}, where the fixed points of the map $\rho(U)$ are listed. It is easy to see that, for generic $U$ with $\Tr U \ne 2$, the solutions for $u/v$ and $z$ are indeed isolated points. These are non-abelian flat connections on $M_3$.

\begin{table}[htb]
	\centering
	\renewcommand{\arraystretch}{1.3}
	\begin{tabular}{|@{~~}c@{~~}|@{~~}c@{~~}|@{~~}c@{~~}|@{~~}c@{~~}|}
		\hline {\bf component} & ~~{\bf signs in \eqref{ABunipotent}} & ~~{\bf parity constraints} & {\bf fixed points of $\rho (U)$}
		\\
		\hline
		\hline
		\multirow{2}{*}{$D_1$} & \multirow{2}{*}{$(+,+)$} & \multirow{2}{*}{none} & \multirow{2}{*}{\Bigg(
			$\begin{matrix}
			a & b \\
			c & d
			\end{matrix}
			\Bigg)
			\Bigg(
			\begin{matrix}
			u \\
			v
			\end{matrix}
			\Bigg)
			\; = \;
			z^2
			\Bigg(
			\begin{matrix}
			u \\
			v
			\end{matrix}$
			\Bigg)} \\		
		{} & {} & {} & {} \\
		\hline
\multirow{2}{*}{$D_2$} & \multirow{2}{*}{$(+,-)$} & $b = \text{even}$ & \multirow{2}{*}{\Bigg(
	$\begin{matrix}
	a & -b \\
	-c & d
	\end{matrix}
	\Bigg)
	\Bigg(
	\begin{matrix}
	u \\
	v
	\end{matrix}
	\Bigg)
	\; = \;
	z^2
	\Bigg(
	\begin{matrix}
	u \\
	v
	\end{matrix}$
	\Bigg)} \\		
{} & {} & $d = \text{odd}$ & {} \\
		\hline
\multirow{2}{*}{$D_3$} & \multirow{2}{*}{$(-,+)$} & $a=\text{odd}$ & \multirow{2}{*}{\Bigg(
	$\begin{matrix}
	a & -b \\
	-c & d
	\end{matrix}
	\Bigg)
	\Bigg(
	\begin{matrix}
	u \\
	v
	\end{matrix}
	\Bigg)
	\; = \;
	z^2
	\Bigg(
	\begin{matrix}
	u \\
	v
	\end{matrix}$
	\Bigg)} \\		
{} & {} & $c=\text{even}$ & {} \\
		\hline
\multirow{2}{*}{$D_4$} & \multirow{2}{*}{$(-,-)$} & $a+b = \text{odd}$ & \multirow{2}{*}{\Bigg(
	$\begin{matrix}
	a & b \\
	c & d
	\end{matrix}
	\Bigg)
	\Bigg(
	\begin{matrix}
	u \\
	v
	\end{matrix}
	\Bigg)
	\; = \;
	z^2
	\Bigg(
	\begin{matrix}
	u \\
	v
	\end{matrix}$
	\Bigg)} \\		
{} & {} & $c+d = \text{odd}$ & {} \\
		\hline
	\end{tabular}
	\caption{Fixed points of the map $\rho (U)$ on the ``unipotent'' components $D_i$. The fixed points exist only when suitable parity constraints are satisfied. The system of equations in the last column is the analogue of \eqref{eigeneqs}. Here, $z$ denotes the holonomy eigenvalue of $SL(2,\C)$ flat connection along the base circle in \eqref{mappingtori}.}
	\label{tab:unipotent}
\end{table}

Also note that, prior to conjugation by $SL(2,\C)$, the values of $(u,v)$ parametrize points on the (``nilpotent'') cone $\C^2$, which is enhanced to $\C^3$ once we include the nilpotent part of the center of mass chiral multiplet $\Phi_0$. These vacua are connected to the corresponding points on the semisimple part of the moduli space \eqref{sspartofM}, where parity constraints of Table~\ref{tab:unipotent} also apply and have a clear meaning. Conjugation by $SL(2,\C)$, however, acts very differently on semisimple and unipotent vacua. In particular, it turns each $\C^3$ cone of unipotent/nilpotent values of $A$, $B$, and $\Phi_0$ into $\C^3 / \C^* \cong {\mathbb{C}}{\mathbf{P}}^2$. Even though such factors can appear for manifolds with $\Tr{U}=2$ as we will see later, for generic $U$ there is a constraint on $u/v$ given in Table~\ref{tab:unipotent}. And, when the holonomy along the base of the mapping torus is diagonal with eigenvalues $z,z^{-1}$, $\Phi_0$ vanishes. Thus, we are still left with isolated points. Finally, when $A$ and $B$ are $\pm1$ with unipotent holonomy along the base, together with nilpotent $\Phi_0$ it leads to $\C^2 / \C^* \cong {\mathbb{C}}{\mathbf{P}}^1$.

For now, let us summarize the structure of the moduli space of vacua of $T[M_3]$ with $G=SU(2)$ that we deduced so far for generic $\Tr{U} \neq 2$. Following the general rule explained earlier, we refer to the vacua represented by abelian (resp. non-abelian) flat connections on $M_3$ as Coulomb (resp. Higgs) branches. With this nomenclature, the Coulomb branch vacua correspond to brane configurations in which individual fivebranes can be separated in transverse directions, whereas turning on vevs of the off-diagonal field components corresponds to Higgs branch vacua of the fivebrane system, in which individual fivebranes are bound to each other:
\be
\CM_{\text{vacua}} \left( T[M_3,SU(2)] \right) \; = \; 
\underbrace{\frac{\CM_+ \times \C^* \times \C}{\Z_2}}_{\text{``Coulomb''}} 
 \; \bigcup \; \underbrace{\frac{\CM_-}{\Z_2} \; \cup \; \{ \bullet , \cdots , \bullet \} \; \cup \; {\mathbb{C}}{\mathbf{P}}^1\;\cup \cdots \cup \;{\mathbb{C}}{\mathbf{P}}^1 }_{\text{``Higgs''}}
\label{MSU2}
\ee
Similarly to the case of $T^2$, the unipotent components can be thought of as resolutions of the orbifold singularities. Since this space is K\"ahler, it can be used as a target space of a sigma-model with $\CN=2$ supersymmetry. In a variety of questions, this provides a good approximation to the 3d $\CN=2$ theory $T[M_3]$, at least on a circle. Could it be the full story?

The answer is ``yes'' in the abelian case, where $T[M_3]$ is indeed equivalent to a sigma-model on \eqref{Mabelian}, either via reduction on a circle or via dualizing 3d photons into compact scalars. These two ways of completing $\R$-valued scalars in 3d vector multiplets into $\C^*$-valued scalars of twisted chiral multiplets are equivalent (T-dual) in the abelian case. In the non-abelian case, only the first option is readily available, and moreover the sigma-model on \eqref{MSU2} can not be the full story since this space has singularities.

Luckily, the singularities in \eqref{MSU2} have a clear origin and a familiar form: they arise from conjugation action by $G_{\C} = SL(2,\C)$ and appear in ABJM theory of type $G$ (= IR limit of 3d maximally supersymmetric gauge theory with group $G$) \cite{Aharony:2008ug}. This suggests a non-abelian model that we describe next.

\subsubsection{3d $\CN=2$ Skyrme-like models}

We wish to realize a 2d $\CN=(2,2)$ sigma-model with target \eqref{MSU2} as a reduction of some 3d theory associated to $\Sigma$, namely $T[\Sigma \times S^1]$, on a circle, with 2d interfaces ($S$ and $T$ duality walls) inserted at points on a circle in a periodic arrangement \eqref{UviaST}:
\begin{center}
\includegraphics[scale=0.4]{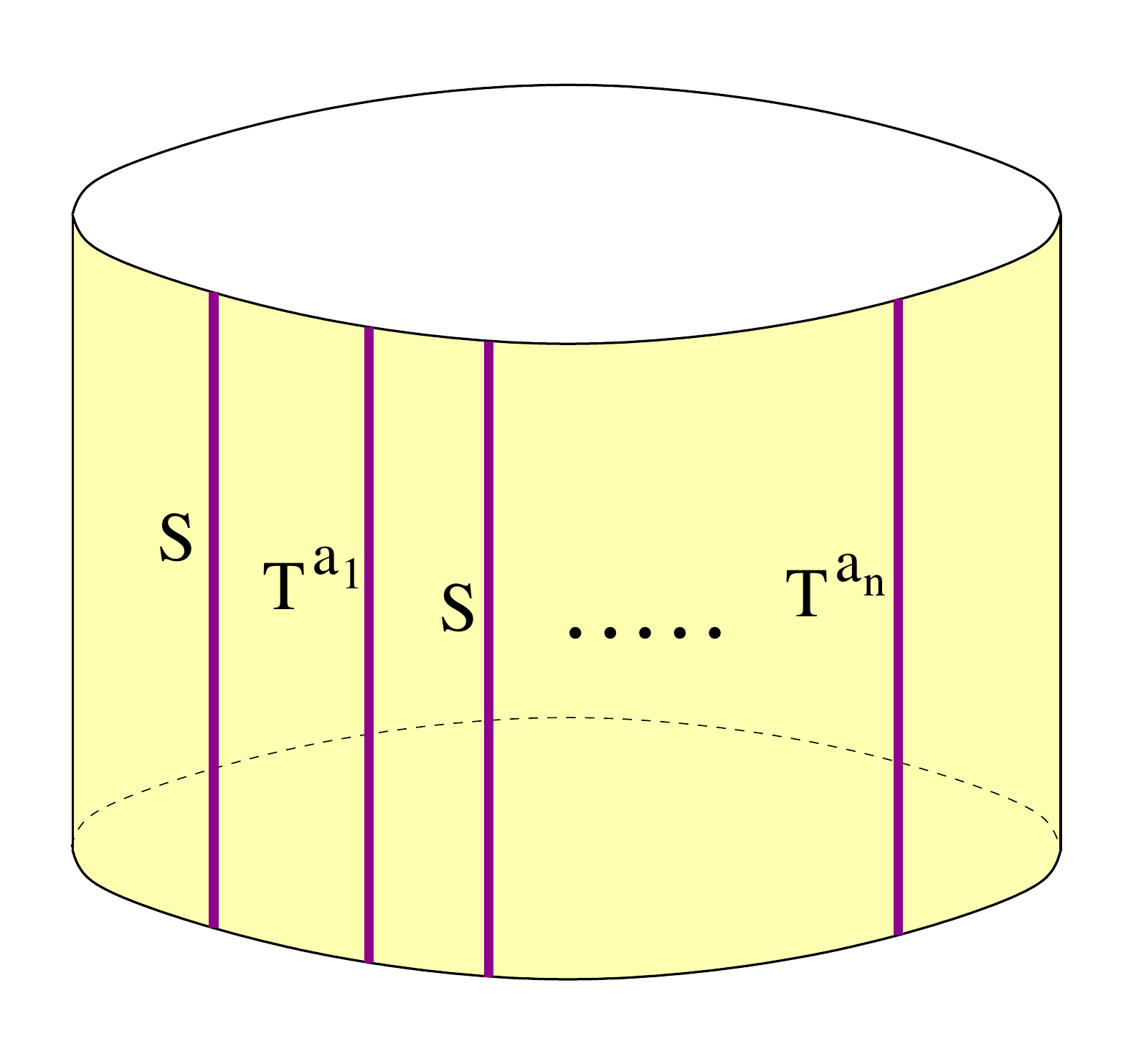}
\end{center}
The 3d theory $T[\Sigma \times S^1]$ in question is supposed to be a reduction of 5d maximally supersymmetric Yang-Mills theory on $\Sigma$ of finite size, {\it i.e.} with Kaluza-Klein modes included. In particular, it should enjoy a duality symmetry $\text{MCG} (\Sigma)$ realized by BPS duality walls.

For the moment, let us assume that $\Sigma$ has arbitrary genus $g$. Then a compactification with a partial topological twist along $\Sigma$ turns two adjoint 5d scalars into a 1-form on $\Sigma$, which complexifies $G$-valued holonomies of the gauge field. As a result, we get $2g$ three-dimensional chiral multiplets with values in $G_{\C}$, on which the gauge symmetry $G$ acts in a natural way. The 3d theory $T[\Sigma \times S^1]$ also has a gauge-invariant superpotential that imposes a constraint 
\be
\prod_{i=1}^g A_i B_i A_i^{-1} B_i^{-1} \; = \; {\bf 1}
\label{highergenusW}
\ee
In addition, there are three adjoint scalars of the original 5d theory, which in 3d become superpartners in a $\CN=4$ vector multiplet. This 3d $\CN=4$ vector multiplet is equivalent to a 3d $\CN=2$ vector multiplet and an adjoint chiral multiplet that we call $\Phi_0$.

In principle, this provides a potential candidate for the 3d $\CN=2$ theory $T[\Sigma \times S^1,G]$, with $\Sigma$ of arbitrary genus $g$ (and arbitrary non-abelian $G$ of ADE type).
As far as we know, the IR physics of such 3d $\CN=2$ gauged sigma-models, with target spaces that involve products of several copies of $G_{\C}$, have not been studied so far. Such models are supersymmetric gauged analogues of 3d Skyrme, WZW-like, and principal chiral models~\cite{Zahed:1986qz}.
Since the complex group manifold $G_{\C}$ admits a K\"ahler metric \cite{MR1083440}, one should expect that a sigma-model with this target admits a supersymmetric extension with $\CN=2$ supersymmetry. This is indeed the case, as can be verified by a direct calculation \cite{Gudnason:2015ryh}.

Relegating a more systematic study of such 3d $\CN=2$ theories with multiple $G_{\C}$-valued chiral multiplets to future work, in what follows we mainly focus on the case of genus $g=0$ (for which the candidate $T[\Sigma \times S^1]$ has no $G_{\C}$-valued multiplets) and on the case of genus $g=1$, for which the candidate $T[\Sigma \times S^1]$ has two $G_{\C}$-valued matter multiplets, $A$ and $B$:
\be
T[T^2 \times S^1]
\quad = \quad
\boxed{\begin{array}{c}
		~3d~\CN=2 \text{ gauged sigma-model with gauge group } G \\
		\text{ and target space } \\
		\left \{  (A,B,\Phi_0) \in G_{\C} \times G_{\C} \times {\frak g}_{\C} \left|
		\begin{array}{rcl}
			\left[A,B\right] & = & 0 \\
			\left[A,\Phi_0\right] & = & 0 \\
			\left[B,\Phi_0\right] & = & 0
		\end{array}
		\right.
		\right \}
\end{array}}
\label{TofSigma}
\ee
The mapping class group $\text{MCG} (\Sigma) = SL(2,\Z)$ acts only on $A$ and $B$ by sending $A \mapsto A^a B^b$ and $B \mapsto A^c B^d$.
Note, that $G_{\C} \times G_{\C}$ here is a non-abelian analogue of $\C^* \times \C^*$ from \eqref{Tabelian}.
Since $T^2 \times S^1$ is simply a 3-torus $T^3$, the manifest $SL(2,\Z)$ symmetry of \eqref{TofSigma} should be enhanced to $SL(3,\Z)$ symmetry in the infra-red.

Indeed, independently of the proposal \eqref{TofSigma}, we know that $T[M_3 = T^3]$ should have maximal ($\CN=8$) supersymmetry and enjoy an $SL(3,\Z)$ self-duality symmetry. We can also easily describe its moduli space. For simplicity and concreteness, let us return to the case of $G=SU(2)$. Then, asymptotically, when the vevs of non-compact scalars are large, the moduli space looks like
\be
\frac{(\C^*)^3 \times \C}{\Z_2}
\label{TTTvacua}
\ee
Here, as in the discussion around \eqref{T3vacua}, we use the fact that $T[M_3 = T^3]$ is a 5d SYM on a 2-torus $T^2$. Holonomies of the 5d gauge field along 1-cycles of the $T^2$ were precisely the motivation for \eqref{TofSigma} and account for two out of three $\C^*$ factors in \eqref{TTTvacua}. These $\C^*$ factors can be thought of as Cartan components of the fields $A$ and $B$ enforced by commutators of the 5d theory in a background of large scalar vevs. Large vevs of non-compact scalars not only mean that all matter fields are simultaneously diagonalizable, but also break the $SU(2)$ gauge symmetry to a $U(1)$ subgroup. Therefore, in asymptotic region of the moduli space we can dualize the $U(1)$ photon into a compact (periodic) scalar, which then accounts for the third $\C^*$ factor in \eqref{TTTvacua}.

So far, this part of the analysis follows closely \cite{Seiberg:1996nz}, where moduli spaces of various rank-one 3d $\CN=4$ gauge theories were studied. The next step, however, is where $T[M_3 = T^3]$ differs from a typical 3d $\CN=4$ theory. Indeed, while the metric on the moduli spaces of the latter in general is quantum-corrected, the $\CN=8$ supersymmetry of $T[M_3 = T^3]$ prevents such corrections. Therefore, we conclude that \eqref{TTTvacua} is the actual quantum moduli space of 3d theory $T[M_3,G]$ with $M_3 = T^3$ and $G=SU(2)$. Note, unlike $\CM_{\text{vacua}}$ that features prominently throughout the paper and refers to the moduli space of a 3d theory on a circle, here we consider the moduli space of 3d theory in flat space-time, {\it i.e.} on $\R^3$.

\begin{figure}[ht]
	\centering
	\includegraphics[width=3.1in]{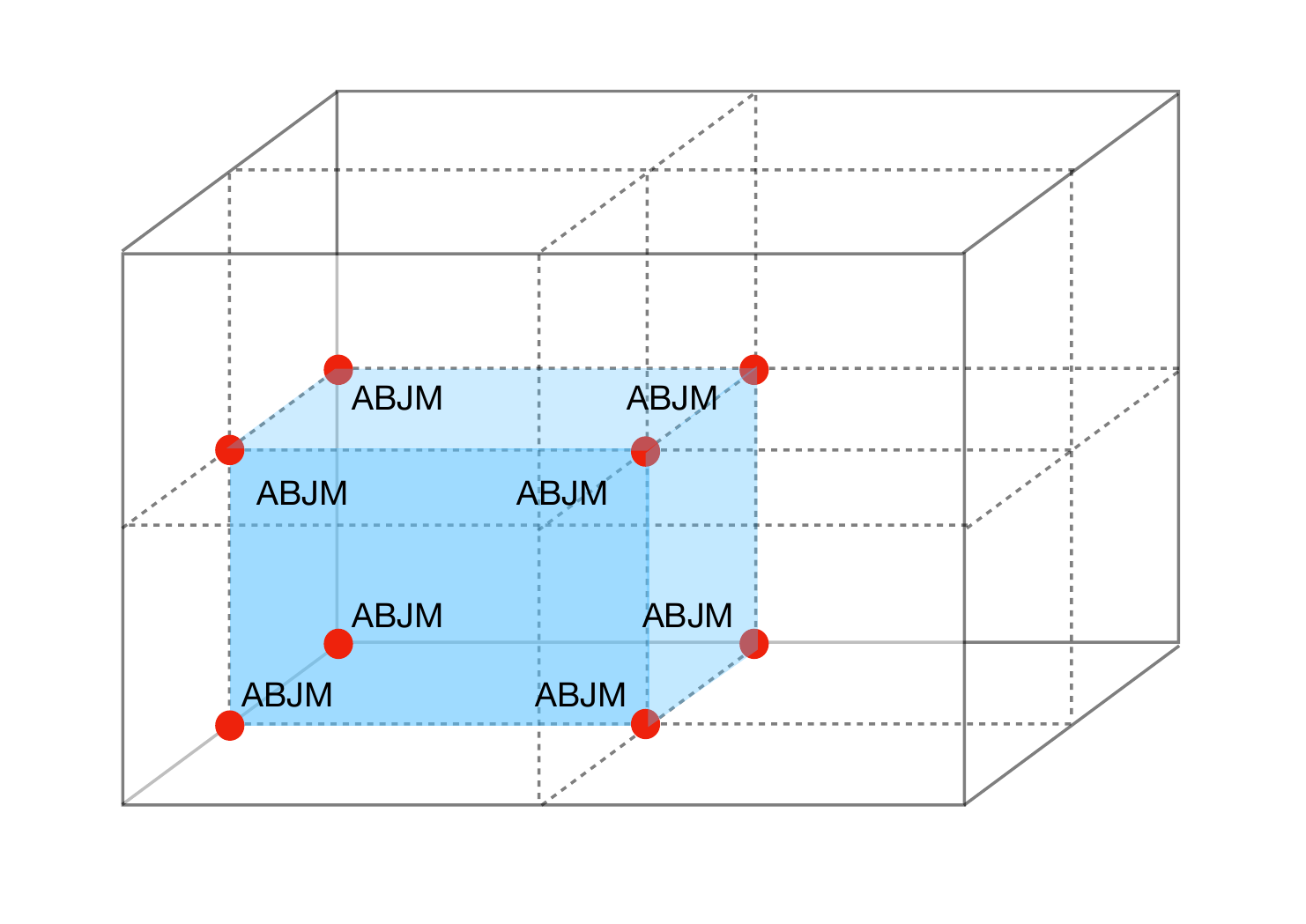}
	\caption{The moduli space of $T[M_3,G]$ with $M_3 = T^3$ and $G=SU(2)$. Shown here is the real slice of the space $\frac{\C^* \times \C^* \times \C^*}{\Z_2}$, a three-dimensional ``pillowcase.'' The eight ABJM theories are at the eight corners of the pillowcase.}
	\label{fig:ABJM}
\end{figure}

In fact, now is a good time to show that the moduli space of 3d theory $T[M_3 = T^3]$ on a circle {\it also} has the form \eqref{TTTvacua}, thus justifying an earlier claim \eqref{T3vacua}. Indeed, a bonus feature of our above analysis of the moduli space is the evidence for the effective theory \eqref{TofSigma} in the asymptotic region of the moduli space. Therefore, in order to describe $\CM_{\text{vacua}} \left( T[M_3] \right)$, we can analyze the moduli space of this theory on $S^1$ and make use of $\CN=8$ supersymmetry, as in the previous discussion.
Upon compactification on $S^1$, the 3d vector multiplet of the theory \eqref{TofSigma} gives rise to a standard 2d $\CN=(2,2)$ vector multiplet and a 2d chiral multiplet with values in $G_{\C}$, parametrized by the holonomy of $A + i \sigma$ along the $S^1$. If we denote this chiral multiplet by $C$, then
\begin{multline}
\CM_{\text{vacua}} \big( T[T^3,SU(2)] \big) \; = \; \\
\; = \; \left \{  (A,B,C,\Phi_0) \in G_{\C} \times G_{\C} \times G_{\C} \times {\frak g}_{\C} \left|
\begin{array}{rcccl}
	\left[A,B\right] & = & 0 = & \left[A,\Phi_0\right] \\
	\left[A,C\right] & = & 0 = & \left[B,\Phi_0\right] \\
	\left[B,C\right] & = & 0 = & \left[C,\Phi_0\right]
\end{array}
\right.
\right \} \Big/ G_{\C} \; \cong \; \\
\; \cong \; \frac{\C^* \times \C^* \times \C^* \times \C}{\Z_2}
\label{MforT3}
\end{multline}
This is the same as the moduli space \eqref{TTTvacua} of 3d theory $T[M_3 = T^3]$ in flat space. (To be more precise, the two moduli spaces are related by T-duality along one of the $S^1$'s.) This provides partial evidence for the proposed UV description \eqref{TofSigma} of $T[M_3,G]$ with $M_3 = T^3$ and $G = SU(2)$.
The IR description, according to the above analysis, consists of eight ABJM theories that share the same moduli space \eqref{TTTvacua}. This is also consistent with the proposed UV description \eqref{TofSigma}, where 8 copies of the ABJM theory reside at the fixed points of the $\Z_2$ Weyl symmetry of $G=SU(2)$.
Without loss of generality, we can consider the singular point at $A = B  = {\bf 1}$. Near this point we can replace $G_{\C}$-valued fields $A=e^{\alpha}$ and $B=e^{\beta}$ by the Lie algebra valued fields $\alpha$ and $\beta$. Therefore, near each singular point, the physics of \eqref{TofSigma} is well approximated by 3d $\CN=2$ gauge theory with three adjoint chiral multiplets $\alpha$, $\beta$, and $\Phi_0$, which in the IR indeed flows to the ABJM fixed point \cite{Aharony:2008ug}.

It would be interesting to test the proposal \eqref{TofSigma} further by computing various supersymmetric indices and partition functions in the UV and in the IR. We expect the special functions of \cite{Felder:2000mq,MR1800253} to play a role.

Note, the UV Lagrangian that we proposed here for the theory $T[M_3]$ with $M_3 = T^3$ is based on the flat metric on $T^3$. However, since 6d fivebrane theory is patrially twisted along $M_3$, we could use any other metric on $M_3 = T^3$. Since such perturbations are $Q$-exact, with respect to the topological supercharges on $M_3$, they do not affect the IR fixed point $T[M_3]$, which still should be a maximally supersymmetric SCFT in three dimensions. However, such perturbations modify the UV description of the theory $T[M_3]$, breaking SUSY down to $\CN=2$ for a generic choice of metric on $M_3$. Some of such 3d $\CN=2$ theories may provide a more convenient description of $T[M_3]$ for $M_3 = T^3$ and would be well worth investigating. They could also lead to new examples of RG flows with supersymmetry enhancement.

The moduli spaces \eqref{TTTvacua} and \eqref{MforT3} describe, respectively, the vacua on $\R^3$ and on $S^1 \times \R^2$ of the ``sheaf'' of SCFTs $T[M_3]$, with $M_3 = T^3$, {\it i.e.} a genus-1 mapping torus with $U = {\bf 1}$. In order to describe the analogue of \eqref{MforT3} for a more general mapping torus with $U = \bigl(\begin{smallmatrix} a & b \\ c & d \end{smallmatrix} \bigr) \in SL(2,\Z)$, one needs to introduce into 3d theory on $S^1 \times \R^2$ a periodic arrangement of $S$ and $T$ duality walls located at points on a circle. As a result, the fields $A$ and $B$ undergo non-trivial monodromies when we go around the base circle, and the relations $[A,C] = 0 = [B,C]$ in \eqref{MforT3} are replaced by
\be
C A C^{-1} \; = \; A^a B^b
\,, \qquad
C B C^{-1} \; = \; A^c B^d
\label{generaltriple}
\ee
with all other details unchanged.
These equations for $(A,B,C)$ are, then, simply the relations in the fundamental group of a genus-1 mapping torus $M_3$:
\be
\pi_1 (M_3) \; = \;
\left \langle
x,y,z \; \left| \;
xy = yx, \;
zxz^{-1} = x^a y^b, \;
zyz^{-1} = x^c y^d
\right.
\right \rangle
\label{mappingpione}
\ee
represented by $G_{\C}$-valued matrices:
\be
\begin{array}{cccc}
\rho: \quad & \pi_1 (M_3) & ~\longrightarrow~ & G_{\C} \\
& x & \mapsto & A \\
& y & \mapsto & B \\
& z & \mapsto & C
\end{array}
\label{ABChols}
\ee

Note, the case $U = \bigl(\begin{smallmatrix} 1 & p \\ 0 & 1 \end{smallmatrix} \bigr) = T^p$ is somewhat special. As pointed out earlier, this is precisely when $b_1 (M_3) \ge 2$. (In fact, $b_1 = 2$ for $p \ne 0$.) Moreover, in this case, according to \eqref{MforT3}--\eqref{generaltriple}, $B$ commutes with both $A$ and $C$.
On one hand, when matrices can be simultaneously diagonalized ({\it cf.} ``semisimple'' component) the eigenvalues of $A$, $B$, and $C$ take values in $\C^*$, $\Z_p$, and $\C^*$, respectively. On the other hand, for unipotent components $B = \pm 1$. Therefore, we conclude that, for a mapping torus of $U = T^p$,
\be
\CM_{\text{flat}} \left( M_3 , SL(2,\C) \right) \; = \; \frac{\Z_p \times \C^* \times \C^*}{\Z_2}\; \bigcup \;  {\mathbb{C}}{\mathbf{P}}^1\;\cup \cdots \cup \;{\mathbb{C}}{\mathbf{P}}^1
\ee
Moreover, when we take into account the adjoint scalar with diagonal values for abelian connections and nilpotent values for non-abelian connection we get
\be
\CM_{\text{vacua}} \big( T[M_3,SU(2)] \big) \; = \; \frac{\Z_p \times \C^* \times \C^* \times \C}{\Z_2}\; \bigcup \; {\mathbb{C}}{\mathbf{P}}^2\;\cup \cdots \cup \;{\mathbb{C}}{\mathbf{P}}^2
\ee
Comparing this result with the derivation of \eqref{MforT3} suggests that the theory $T[M_3]$ for the mapping torus of $U = T^p$ is a deformation of \eqref{TofSigma} by 3d $\CN=2$ supersymmetric Chern-Simons term at level $p$.

Note, in the case of mapping tori with $U = T^p$, including $p=0$, we have $\CM_- \subset \CM_+$ in \eqref{MSU2}. This is no longer the case for more general choice of the monodromy $U$, {\it i.e.} for mapping tori with $b_1 (M_3) = 1$. Indeed, writing
\be
A \; = \;
\begin{pmatrix}
	x & * \\
	0 & x^{-1}
\end{pmatrix}
\quad,\qquad
B \; = \;
\begin{pmatrix}
	y & * \\
	0 & y^{-1}
\end{pmatrix}
\ee
we see that equations \eqref{generaltriple} become \eqref{eigeneqs} for the holonomy eigenvalues $x$ and $y$.\footnote{Hopefully, $x$ and $y$ that denote holonomy eigenvalues here will not be mistaken for the generators of $\pi_1 (M_3)$ denoted by the same letters in \eqref{mappingpione}. They are related, of course, so that there is little danger in the confusion. And, in this paper, the primary use of $x$ and $y$ is to denote holonomy eigenvalues.}

Another useful remark is that the role of $\CM_+$ and $\CM_-$ is exchanged under $U \mapsto - U$. Since $\CM_+$ is part of the Coulomb branch and $\CM_-$ is part of the Higgs branch, {\it cf.} \eqref{MSU2}, it is natural to call this map a ``mirror symmetry,'' by analogy with \cite{Intriligator:1996ex}. Using \eqref{mappingpione} we can prove the following

\begin{lem}\label{thm:mirrormap}
	For $G=SU(2)$, there is a canonical bijection\footnote{Here, we use notation $\pi_0$ to denote the space of \emph{connected} components, not the space of path-connected components. The distinction is due to the fact that $\mathcal{M}_{flat}(M_3,G_{\mathbb{C}})$ is in general not Hausdorff, at least naively, due to the unipotent components.}
	\be
	\sigma: \quad \pi_0 \left( \CM_{\text{flat}} ( M_3 (U) , G_{\C}) \right)
	\; \longrightarrow \;
	\pi_0 \left( \CM_{\text{flat}} ( M_3 (-U) , G_{\C}) \right)
	\label{MMmap}
	\ee
such that
	\be
	\text{CS} (\rho) \; = \; \text{CS} (\sigma (\rho))
	\quad \mod~1
	\label{MMmapCSvalues}
	\ee
\end{lem}

\begin{proof}
First, following the same analysis that led us to unipotent $SL(2,\C)$ flat connections in Table~\ref{tab:unipotent}, one can quickly see that all such complex flat connections on $M_3 (U)$ are, in fact, in the same connected components of $\CM_{\text{flat}} ( M_3 (U) , G_{\C})$ as $SU(2)$ flat connections with holonomies $A = \pm {\bf 1}$ and $B = \pm {\bf 1}$. In the notations \eqref{ABChols}, these $SL(2,\C)$ flat connections have holonomy eigenvalue $z$ along the base circle such that
\begin{equation}\label{alexander}
z^4 - \mathrm{tr}(U) z^2 + 1 \; = \; 0
\end{equation}

Then, it suffices to prove the statement for $SU(2)$ flat connections. If the fundamental group of the mapping torus $M_3 (U)$ is given by \eqref{mappingpione}, then in the same conventions that of $M_3 (-U)$ is
\be
\pi_1 (M_3(-U)) \; = \;
\left \langle
x,y,z \; \left| \;
xy = yx, \;
zxz^{-1} = x^{-a} y^{-b}, \;
zyz^{-1} = x^{-c} y^{-d}
\right.
\right \rangle
\ee
For any $\rho \in \text{Hom} \left( \pi_1(M_3 (U)),SU(2) \right)$, we can think of $\rho' \in \text{Hom} \left( \pi_1 (M_3(-U),SU(2) \right) / \text{conj.}$
determined by
\[\rho'(x) \; = \; \rho(x) \quad,\quad
\rho'(y) \; = \; \rho(y) \quad,\quad
\rho'(z) \; = \; j \rho (z)
\]
where $j \in SU(2)$ is chosen such that $j\rho(x) j^{-1} = \rho(x)^{-1}$ and $j\rho(y)j^{-1} = \rho(y)^{-1}$.
Generically, there will be a $U(1)$-family of such matrices, and they all are conjugate to each other. For instance, if we take $\rho(x)$ and $\rho(y)$ to be diagonal, then $j$ can be $\Bigl( \begin{smallmatrix} 0 & e^{i\theta} \\ -e^{-i\theta} & 0\end{smallmatrix} \Bigr)$, and the corresponding homomorphisms are conjugate to each other. The non-generic case occurs when $\rho(x)$ and $\rho(y)$ are both $\pm {\bf 1}$, and in that case $j$ can be an arbitrary matrix. For example, we can pick $j={\bf 1}$ in that case. 
We have just defined a map
\[\text{Hom} \left( \pi_1(M_3(U)),SU(2) \right) \quad \rightarrow \quad \text{Hom} \left( \pi_1(M_3 (-U)),SU(2) \right) / \text{conj.}
\]
This map does {\it not} descend to a map
\[\text{Hom} \left( \pi_1(M_3(U)),SU(2) \right) / \text{conj.} \quad \rightarrow \quad \text{Hom} \left( \pi_1(M_3 (-U)),SU(2) \right) / \text{conj.}
\]
because, unlike $\Bigl( \begin{smallmatrix} 0 & e^{i\theta} \\ -e^{-i\theta} & 0\end{smallmatrix} \Bigr)$, the matrices $\Bigl( \begin{smallmatrix} e^{i\theta} & 0 \\ 0 & e^{-i\theta} \end{smallmatrix} \Bigr)$ are all non-conjugate to each other. Nevertheless, once we pass to connected components, we arrive at the desired canonical one-to-one correspondence~\eqref{MMmap}:
\be
\pi_0 \left( \CM_{\text{flat}} ( M_3 (U) , G_{\C}) \right)
\; \simeq \;
\pi_0 \left( \CM_{\text{flat}} ( M_3 (-U) , G_{\C}) \right)
\ee

What is even better, the correspondence holds at the level of the Chern-Simons functional. That is, for a pair of flat connections $A\in \CM_{\text{flat}} ( M_3 (U) , G_{\C})$ and $A'\in \CM_{\text{flat}} ( M_3 (-U) , G_{\C})$ related by the above correspondence, the values of the Chern-Simons functionals coincide:
\[CS(A) \; = \; CS(A')\]
This is because of the fact that $\tilde{M_3} := M_3 (U^2)$ is a double cover of both $M_3 (U)$ and $M_3 (-U)$, and $A$ and $A'$ lift to the same connection, say, $\tilde{A}$ on $\tilde{M_3}$. Hence,
$2 CS(A) = CS(\tilde{A}) = 2 CS(A')$.

The reason $A$ and $A'$ lift to the same connection on $M_3 (U^2)$ is that they induce, respectively, a representation of the fundamental groupoid $\pi_1(M_3 (U))$ and $\pi_1(M_3 (-U))$ to $SU(2)$.
The corresponding representations of $\pi_1(M_3 (U^2))$ are just their pull-backs. It is easy to see that these two representations are related by a gauge transformation that is null-homotopic. Therefore, $A$ and $A'$ lift to the same connection on $M_3 (U^2)$.
\end{proof}

\medskip

It is a good problem to generalize the analysis in this section to groups of higher rank, {\it cf.} \cite{Chung:2018rea,Park}. One novel feature that will play an interesting role in such generalizations is the existence of non-trivial commuting triples \cite{Witten:1997bs,deBoer:2001wca}.

\subsubsection{0-surgeries on knots}

Perhaps the simplest 3-manifold with $b_1 (M_3) > 0$ is $M_3 = S^2 \times S^1$. As such, it should then be our central example, if not the starting point, in the study of 3d-3d correspondence and $q$-series invariants of 3-manifolds with $b_1 > 0$. However, $M_3 = S^2 \times S^1$ is not a genus-1 mapping torus. Why do we mention it here, then?

One reason is that $M_3 = S^2 \times S^1$ is indeed a more basic and fundamental example\footnote{As we will see shortly, the simplicity of this example is very deceptive.} since it is a genus-0 mapping torus, {\it i.e.} it is part of the family~\eqref{mappingtori} with $\Sigma = S^2$ and $\varphi = {\bf 1}$. Another, less obvious reason is that $M_3 = S^2 \times S^1$ has something in common with the following genus-1 mapping tori:
\be
M_3 \; = \;
\begin{cases}
	S^3_0 ({\bf 3_1}) &\text{if} \ \ U = - STST =
	\begin{pmatrix}
		1 & 1 \\ -1 & 0
	\end{pmatrix} \\
	S^3_0 ({\bf 4_1}) &\text{if} \ \  U = - STST^{-1} =
	\begin{pmatrix}
		1 & -1 \\ -1 & 2
	\end{pmatrix} 
\end{cases}
\label{UUsurgeries}
\ee
which also can be realized as 0-surgeries on knots (see {\it e.g.} \cite{MR1283727}). Indeed,
\be
S^3_0 (\text{unknot})
\; = \; S^2 \times S^1
\label{SoneStwo}
\ee
is a part of this family too. In fact, all 0-surgeries on knots, $M_3 = S^3_0 (K)$, have the property $H_1 (M_3) = \Z$ and, therefore, also provide many non-trivial examples of 3-manifolds with $b_1 > 0$.

\begin{figure}[ht]
	\centering
	\includegraphics[scale=0.5]{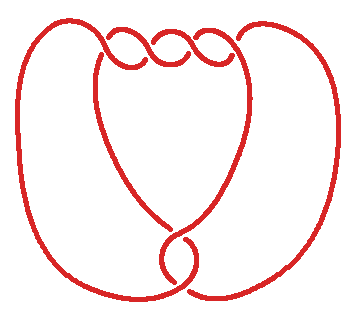}
	\caption{The $\mathbf{5}_2$ knot.}
	\label{fig:52knot}
\end{figure}

In Table~\ref{tab:knots} we summarize flat connections on $M_3 = S^3_0 (K)$ for various simple examples of $K$. Besides the unknot, the trefoil knot ${\bf 3_1}$ and the figure-eight knot ${\bf 4_1}$ mentioned earlier, included in this list is a hyperbolic knot ${\bf 5_2}$, whose 0-surgery is not a mapping torus (of any genus). All these knots, however, are examples of {\it twist knots} which form a natural generalization of the examples listed in Table~\ref{tab:knots}.

\begin{table}[htb]
	\centering
	\renewcommand{\arraystretch}{1.3}
	\begin{tabular}{|@{~~}c@{~~}|@{~~}c@{~~}|@{~~}r@{~~}@{~~}l@{~~}|}
		\hline {\bf knot} & {\bf ~~Alexander polynomial} & {\bf ~~flat connection} & {\bf CS value}
		\\
		\hline
		\hline
		\multirow{2}{*}{unknot} & \multirow{2}{*}{$\Delta_{\text{unknot}} (x) = 1$} & \multirow{2}{*}{abelian:} & \multirow{2}{*}{$0$} \\		
		{} & {} & {} & {} \\
		\hline
		\multirow{2}{*}{${\bf 3_1}$} & \multirow{2}{*}{$\Delta_{{\bf 3_1}} (x) = x^{-1} - 1 + x$} & abelian: & $0$ \\		
		{} & {} & almost abelian: & $0$, $-\frac{1}{3}$ \\
		\hline
		\multirow{2}{*}{${\bf 4_1}$} & \multirow{2}{*}{$\Delta_{{\bf 4_1}} (x) = - x^{-1} + 3 - x$} & abelian: & $0$ \\		
		{} & {} & almost abelian: & $0$, $\frac{1}{5}$, $\frac{4}{5}$ \\
		\hline
		\multirow{2}{*}{${\bf 5_2}$} & \multirow{2}{*}{$\Delta_{{\bf 5_2}} (x) = 2 x^{-1} - 3 + 2x$} & abelian: & $0$ \\		
		{} & {} & almost abelian: & $0$, $-\frac{1}{7}$, $-\frac{2}{7}$, $-\frac{4}{7}$ \\
		\hline
	\end{tabular}
	\caption{Alexander polynomials and flat connections for 0-surgeries on some simple knots. The 0-surgeries on knots ${\bf 3_1}$, ${\bf 4_1}$, and ${\bf 5_2}$ have 2, 3, and 4 non-abelian flat connections, represented here by their Chern-Simons values, that are ``almost abelian'' in the sense briefly mentioned below \eqref{eigeneqs} and explained more fully in section \ref{sec:Zhat}. In addition, in all three cases, there are 2 complex flat connections at the singularities of the abelian branch.}
	\label{tab:knots}
\end{table}

A convenient tool in studying complex flat connections on knot surgeries comes from polynomial knot invariants, such as the Alexander polynomial and the A-polynomial (see {\it e.g.} \cite{MR1054574,Gukov:2003na,Gukov:2016njj}). In particular, it allows to deduce the results of Table~\ref{tab:knots} in our simple examples. (However, in general, a more direct analysis of \eqref{Mflat} is required.) For any knot $K$, the A-polynomial has the form
\be
A(x,y) \; = \; (y-1) \, A_{\text{nab}} (x,y)
\ee
where $y-1=0$ is the locus of abelian flat connections on the knot complement, $S^3 \setminus K$, whereas $A_{\text{nab}} (x,y)$ accounts for other, non-abelian flat connections. Performing a 0-surgery on the knot $K$ restricts the value of the holonomy along longitude to the identity, {\it i.e.} $y=1$. This preserves the entire abelian branch, $y-1=0$, and also the non-abelian flat connections with the meridian holonomy eigenvalue $x$ that solves
\be
A_{\text{nab}} (x,1) \; = \; 0
\label{Anabzero}
\ee
For example, for the trefoil knot $K = {\bf 3_1}$ we have $A_{\text{nab}} (x,y) = y + x^6$, so that on $(\C^* \times \C^*) / \Z_2$ parametrized by $(x,y)$ the equation \eqref{Anabzero} has three solutions. One of these solutions, namely $x = \pm i$, is precisely the solution to \eqref{eigeneqs} for $\CM_- (S^3_0 ({\bf 3_1}))$, with $U$ given in \eqref{UUsurgeries}. In other words, this is what we call an ``almost abelian'' flat connection. The other two solutions to \eqref{Anabzero} are complex non-abelian flat connections on $M_3 = S^3_0 ({\bf 3_1})$, which are connected to the abelian branch. In fact, such non-abelain flat connections (or, rather, the corresponding values of $x$) are precisely the roots of the Alexander polynomial ({\it cf.} \eqref{alexander}):
\be
\Delta_K (x^2) \; = \; 0
\ee
These properties hold true for 0-surgeries on arbitrary twist knots, including ${\bf 4_1}$ and ${\bf 5_2}$ listed in Table~\ref{tab:knots}. Namely, the roots of the Alexander polynomial are in 1-to-1 correspondence with the solutions to \eqref{Anabzero} that represent complex non-abelian flat connections on $S^3_0 (K)$ connected to the abelian branch, whereas all almost abelian flat connections have $x = \pm i$. The union of these two disjoint sets accounts for all solutions to \eqref{Anabzero} when $K$ is a twist knot. In fact, for twist knots, the multiplicity of the root $A_{\text{nab}} (\pm i,1) = 0$ determines the total number of (branches of) almost abelian flat connections on $M_3 = S^3_0 (K)$, which equals
\be
1 \, + \, (\text{multiplicity of } i) \; = \; \# \text{ of almost abelian flat connections}
\ee

In the case of the 0-surgery on the unknot \eqref{SoneStwo}, the moduli space of $SL(2,\C)$ flat connection can be obtained directly from \eqref{Mflat}. Indeed, since for $M_3 = S^2 \times S^1$ the fundamental group $\pi_1 (M_3) = \Z$ is abelian, one might expect the space of vacua to consist entirely of the Coulomb branch, {\it cf.} \eqref{MSU2},
\be
\CM_{\text{vacua}} \left( T[S^2 \times S^1,SU(2)] \right) \; \stackrel{?}{=} \; \frac{\C^* \times \C}{\Z_2}
\label{SSvacua}
\ee
While in this simple example it is clear that $T[S^2 \times S^1]$ should be a 3d theory with $\CN=4$ supersymmetry --- because the holonomy of $M_3 = S^2 \times S^1$ is reduced --- it is less clear what this theory is. A simple 3d $\CN=4$ theory that has Coulomb branch \eqref{SSvacua} is 3d $\CN=4$ SQED with gauge group $G=SU(2)$ and $N_f = 2$ flavors \cite{Seiberg:1996nz}. This theory, however, also has a Higgs branch for which there is no room in $T[S^2 \times S^1]$.

A more promising candidate for $T[S^2 \times S^1]$ --- that can be justified either via string dualities or by extrapolating the family of Lens space theories $T[L(k,1)]$ all the way to $k=0$ --- is a 3d $\CN=4$ pure super-Yang-Mills (SYM) with gauge group $G=SU(2)$. At low energies, it reduces to a 3d $\CN=4$ sigma-model on the Atiyah-Hitchin manifold $\CM_{\text{AH}}$ \cite{MR934202}, which is the Coulomb branch of this gauge theory \cite{Seiberg:1996nz}. Indeed, by viewing $S^2$ as a (M-theory) circle fibration over an interval, one can reduce $N$ M5-branes wrapped on $S^2 \times S^1$ to a 5d SYM on D4-branes on an interval \cite{Assel:2016lad} or, upon further reduction and T-duality on $S^1$, to a D3-D5 brane system that describes $N$ monopoles in $SU(2)$ gauge theory \cite{Diaconescu:1996rk,Hanany:1996ie}. Either way, for $N=2$, one finds a 3d $\CN=4$ sigma-model on $\CM_{\text{AH}}$.

Note, however, that $\CM_{\text{AH}}$ is very different from \eqref{SSvacua}, even asymptotically. This, by itself, is not necessarily a problem since \eqref{Mvacua} involves the moduli space of vacua of 3d theory on a circle of small but finite size. In fact, since the resulting theory is basically a 2d theory (with a tower of KK modes), the proper interpretation of $\CM_{\text{vacua}} \left( T[M_3] \right)$ should be as a target space of 2d $\CN=(2,2)$ sigma-model. And, it is entirely conceivable for two 2d $\CN=(2,2)$ sigma-models with different target manifolds to be equivalent as quantum theories. Therefore, one might expect that, after putting 3d $\CN=4$ SYM on a small circle,\footnote{For comparison, the semi-classical moduli space of vacua in 2d $\CN=(4,4)$ super-Yang-Mills with gauge group $G=SU(2)$ is $\C^2 / \Z_2$.} either quantum corrections change it into a theory equivalent to the sigma-model on \eqref{SSvacua}, or this is not the right candidate for $T[S^2 \times S^1]$ altogether.

While the answer to this question is not clear at present, we will try to shed light on it by considering various partition functions of $M_3 = S^2 \times S^1$ in sections \ref{sec:Zhat} and \ref{sec:HF}.


\section{WRT invariants and $q$-series $\hat Z_a (M_3)$}
\label{sec:Zhat}

Given a choice of the gauge group $G$ the invariants of Witten \cite{witten1989quantum} and Reshetikhin-Turaev~\cite{reshetikhin1991invariants} assign a complex number $\text{WRT} (M_3,k) \in \C$ to a closed 3-manifold and an integer $k \in \Z$, called the ``level'' (or, equivalently, a root of unity $q = e^{2\pi i /k}$). These invariants form a TQFT, {\it i.e.} they can be constructed via cutting and gluing. Since the early days of WRT invariants, one of the long-standing problems --- especially important for categorification --- was to find an extension of this TQFT to generic values of $q$.

Recent developments suggest an answer to this question in the form of a TQFT denoted $\hat Z$, which is essentially a concrete realization of complex, analytically continued Chern-Simons theory with gauge group $G_{\C}$ to generic values of $|q|<1$. The resulting theory has all the desired properties of a TQFT; it allows surgery operations as well as gluing 3-manifolds with boundary into closed ones, and it fits into the framework of 3d-3d correspondence via the so-called ``half-index'' partition functions:
\be
\Tr_{\CH_{D^2}} (-1)^F q^{R/2 + J_3} \; = \; \text{ partition function on } S^1 \times_q D^2
\label{halfindex}
\ee
The half-index of a 3d $\CN=2$ theory, introduced in \cite{Gadde:2013wq}, is a combined index of a 3d theory together with a choice of 2d $\CN=(0,2)$ boundary condition, {\it cf.} Figure~\ref{fig:halfplane}. It basically is a 3d analogue of the elliptic genus of 2d theories. In particular, it enjoys hidden modularity properties of the type that first appeared in Ramanujan's ``lost'' notebook and are commonplace in logarithmic CFTs.

Our main goal in this section is to explore the properties of $\hat Z$ for 3-manifolds with $b_1 > 0$. As we shall see, there are many new features, which include a surprising role played by ``almost abelian'' flat connections.

\begin{figure}[ht]
	\centering
	\includegraphics[width=3.7in]{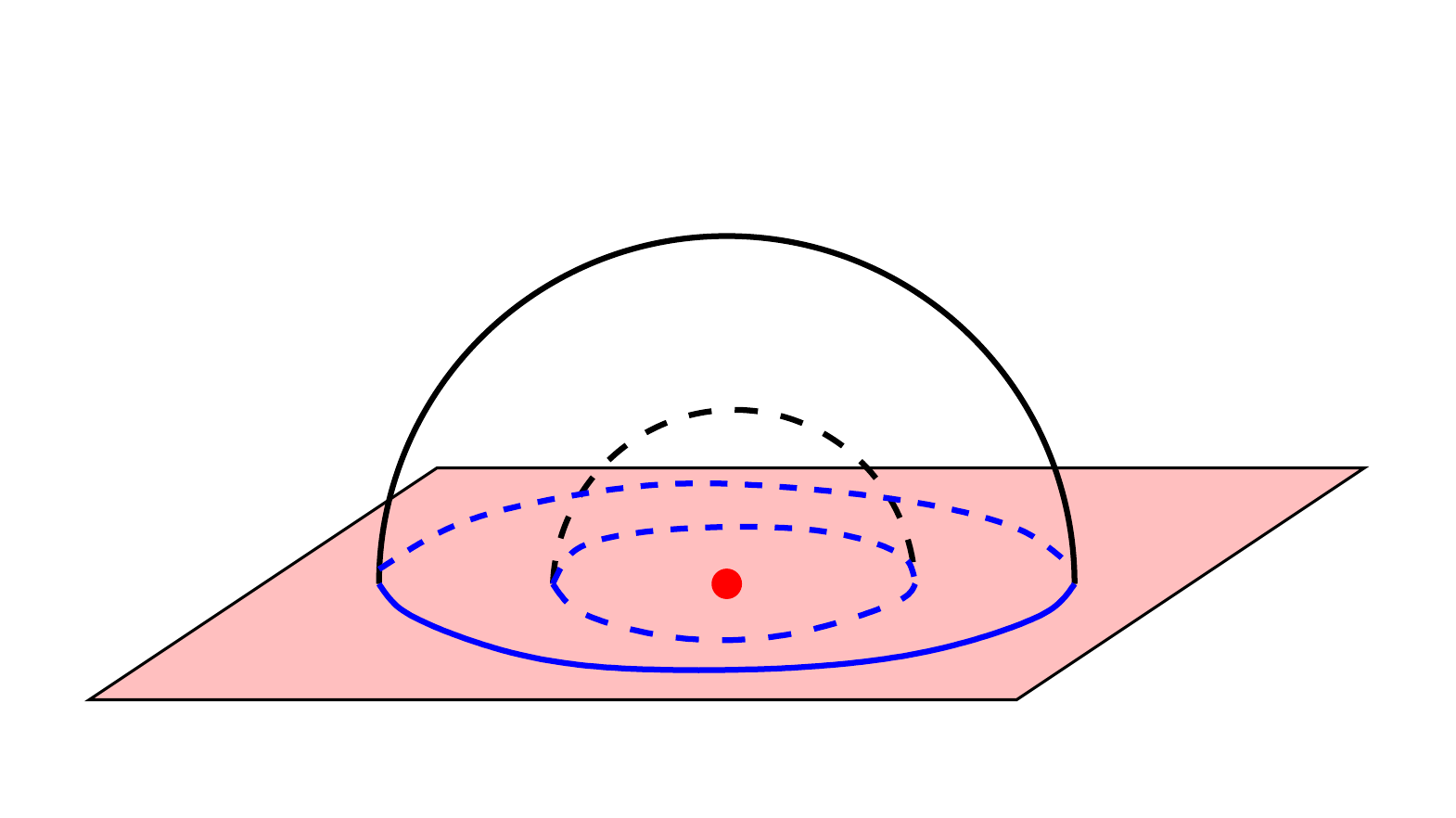}
	\caption{Counting operators on the boundary of half-space $\R^2 \times \R_+$ leads to combined 2d-3d half-index on $S^1 \times D^2$. In radial quantization, the boundary is foliated by concentric circles, which form $T^2$ boundary of $S^1 \times D^2$.}
	\label{fig:halfplane}
\end{figure}

What does the half-index \eqref{halfindex} count?
In general, the answer to this question depends in a highly non-trivial way on both 3d $\CN=2$ bulk theory as well as the choice of 2d $(0,2)$ boundary condition.
However, as argued in \cite{Gukov:2017kmk} (based on an earlier work \cite{Gukov:2016njj}),
several nice things happen for a special (finite!) set of 2d boundary conditions that
correspond to maximally degenerate vacua of 3d $\CN=2$ theory.
When 3d $\CN=2$ theory has a gauge theory description, these are the vacua where at least
the Cartan part of gauge symmetry is unbroken, {\it i.e.} the ``Coulomb branches,''
and the corresponding boundary conditions impose Neumann boundary conditions on
3d $\CN=2$ vector multiplets~\cite{Gadde:2013sca,Gukov:2017kmk} or, equivalently, Dirichlet boundary conditions on dual photon multiplets.

{}From our discussion in section \ref{sec:Mflat}, we know that in 3d-3d correspondence
the torsion part of $H_1 (M_3)$ labels different components of the Coulomb branch of $T[M_3]$,
whereas the free part of $H_1 (M_3)$ controls the dimension of the Coulomb branch, which for $G=SU(2)$ is indeed parametrized by $b_1$ dual photons.
These dual photons contribute to the half-index \eqref{halfindex} through the fermion modes (and their derivatives), whereas the discrete data that label the components of the Coulomb branch become the labels of the special 2d $(0,2)$ boundary conditions $\CB_a$.
Therefore, with these special boundary conditions, the half-index \eqref{halfindex}
roughly counts the Coulomb branch operators.
A simple illustration is the half-index of the abelian theory \eqref{quiverCS} with Neumann boundary conditions for 3d $\CN=2$ vector multiplets and for the chiral multiplet $\Phi_0$, a simple prototype for many calculations in this section.

Of course, compared to theories with larger
amount of supersymmetry, {\it e.g.} \cite{Gadde:2011uv}, where one can be very precise about Higgs or Coulomb branch operators contributing to different types of SUSY indices, in 3d theories with only $\CN=2$ supersymmetry this characterization can only be approximate at best.
And, as we explain in this section, not only Higgs vacua of 3d $\CN=2$ theory contribute to the half-index \eqref{halfindex}, but they in fact also give rise to special 2d $(0,2)$ boundary conditions akin to those arising from Coulomb branch vacua.

\subsection{``Almost abelian'' flat connections}\label{almostabelianflatconnections}

In section \ref{sec:Mflat}, we already encountered a particular class of complex flat connections which, for $G=SU(2)$, under the map $U \mapsto - U$ correspond to abelian flat connections on a ``mirror'' mapping torus, {\it cf.} \eqref{MMmap} and Proposition~\ref{thm:mirrormap}. Concretely, the components of such ``almost abelian'' flat connections on a mapping torus \eqref{mappingtori} form a set
\be
\pi_0 \, \CM_{\text{flat}}^{\text{almost ab.}} \left( M_3 , SL(2,\C) \right)
\; = \;
\pi_0 \, \CM_-
\; \cong \;
\text{coker} \, \left( - U - {\bf 1} \right)
\label{almostabviaU}
\ee
where, as usual, $U$ denotes a $2 \times 2$ matrix representing $\varphi \in SL(2,\Z)$. This has to be compared with the set
\be
\pi_0 \, \CM_{\text{flat}}^{\text{ab.}} \left( M_3 , SL(2,\C) \right)
\; = \;
\pi_0 \, \CM_+
\; \cong \;
\text{coker} \, \left( U - {\bf 1} \right)
\label{abviaU}
\ee
which, for 3-manifolds with $b_1=0$, is the labeling set of the invariants $\hat Z_b (M_3,q)$, that is the set where labels $a$ take their values.

The main result of this section (if not of the entire paper!) is to provide evidence that, for manifolds with $b_1 > 0$, the set of values of the label $a$ should be extended to include certain very special non-abelian flat connections that we call ``almost abelian'' --- namely, \eqref{almostabviaU} in the case of genus-1 mapping tori --- if we wish to construct $\hat Z_b (M_3,q)$ with the following properties:

\begin{itemize}

\item cutting/gluing operations that are necessary for defining a 3d TQFT ({\it i.e.} building closed 3-manifolds from open ones), 

\item topological invariance ({\it e.g.} invariance under 3d Kirby moves), 

\item relation to WRT invariants, 

\item integrality of the coefficients, 

\item convergence (as a $q$-series) inside the unit disk, $|q|<1$.

\end{itemize}

\noindent
In particular, we expect the relation between $\hat Z_b (M_3)$ and $\text{WRT} (M_3)$ to have the following structure:

\begin{conj}\label{conjWRT}
For any 3-manifold $M_3$
\be
\text{WRT} (M_3,k) \; = \; \Big( \frac{-i}{\sqrt{2k}}  \Big)^{1-b_1} \sum_{a} e^{2 \pi i k \text{CS} (a)} \sum_{b} \CS_{ab} \, \hat Z_b (q) \Big|_{q \to e^{\frac{2 \pi i}{k}}}
\label{WRTconj}
\ee
where $a$ and $b$ run over (almost) abelian flat connections on $M_3$.\footnote{More precisely, the labels $a$ and $b$ run over sets of the same cardinality, yet of possibly different nature, in view of \cite{Gukov:2019mnk}.}
\end{conj}
\noindent
Note that sometimes the naive radial limit of the right-hand side in this conjecture is ill-defined, while the left-hand side makes perfect sense. That is why by $q \to e^{2 \pi i/k}$ limit we actually mean a zeta-regularized value.

For example, for genus-1 mapping tori \eqref{UviaST} with $\Tr U \ne \pm 2$ (all such $M_3$ have $b_1 = 1$), the labels $a$ and $b$ run over the union of the two sets \eqref{abviaU} and \eqref{almostabviaU} which are, respectievly, the sets of abelian and almost abelian flat connections. Correspondingly, the ``S-matrix'' $\CS$ that appears in \eqref{WRTconj} is block-diagonal,
\be
\CS \; = \;
\begin{blockarray}{*{2}{c} l}
	\begin{block}{*{2}{>{$\footnotesize}c<{$}} l}
		abelian & almost abelian &  \\
	\end{block}
	\begin{block}{({c}{c})>{$\footnotesize}l<{$}}
		\CS^{(+)} & 0 \bigstrut[t]& abelian \\
		0 & \CS^{(-)} & almost abelian \\
	\end{block}
\end{blockarray}
\label{SpSm}
\ee
so that \eqref{WRTconj} in this case takes the form
\begin{equation}
\text{WRT} (M_3,k) = 
\sum_{\text{abelian}} e^{2 \pi i k \text{CS} (a)} \CS^{(+)}_{ab} \hat Z^{(+)}_b \Big|_{q \to e^{\frac{2 \pi i}{k}}} + \sum_{\text{almost abelian}} e^{2 \pi i k \text{CS} (a)} \CS^{(-)}_{ab} \hat Z^{(-)}_b \Big|_{q \to e^{\frac{2 \pi i}{k}}}
\label{WRTmappingtori}
\end{equation}
As we will see through many examples below, this structure applies to a more general class of 3-manifolds described by plumbing graphs.

The invariants $\hat Z^{(\pm)}_b$ are quite simple for genus-1 mapping tori in question. Before we write down the explicit expressions, however, let us note that the set of abelian flat connections \eqref{abviaU} can be equivalently described as $\Z^n / Q \Z^n$, where $Q$ is the matrix introduced \eqref{Qmatrix}. This description, better adapted to \eqref{aaaloop}, was already used in section~\ref{sec:abelian}. The set of almost abelian flat connections \eqref{almostabviaU} has a similar description. Namely, we have
\be
\text{coker} \, Q_{\pm} \; \cong \; \text{coker} \, \left( \pm U - {\bf 1} \right)
\ee
where, in the notations \eqref{UviaST}--\eqref{aaaloop}, 
\be
Q_{\pm} \; = \;
\begin{pmatrix}
	a_1 & -1   & 0 & \cdots  & \mp 1 \\
	-1   & a_2 & -1        &   & \vdots \\
	0 & -1 &  & \ddots  &   0 \\
	\vdots  & & \ddots & ~\ddots~  &  ~-1~  \\
	\mp 1 &  \cdots  & 0 & -1 & a_n
\end{pmatrix}
\ee
Therefore, we can describe the sets of abelian and almost abelian flat connections, as well as other ingredients in \eqref{WRTmappingtori}, either in terms of $n \times n$ matrices $Q_{\pm}$ or in terms of $2 \times 2$ matrix $U = \bigl(\begin{smallmatrix} a & b \\ c & d \end{smallmatrix} \bigr)$. These two descriptions, of course, are compatible with each other. For example, we have
\be
|\det Q_{\pm}| \; = \; |\det (\pm U - {\bf 1})| \; = \; |2 \mp (a+d)| \,.
\ee
which express the total number of abelian and almost abelian flat connections in these two descriptions.

Let's denote by $\sigma^{(s)} = b_{+}^{(s)}-b_{-}^{(s)}$ the signature of $Q_{s}$. Now, we are ready to write
\begin{equation}
    \hat{Z}_b^{(\pm)} = \begin{cases} \frac{(-1)^{ b_+^{(+)}\pm 1}}{2} e^{\frac{\pi i}{4} (\sigma^{(\pm)} - \sigma^{(+)})}   q^{\frac{3 \sigma^{(+)}-\sum_v a_v}{4}} &\text{ if }b = 0\in \mathbb{Z}^n/Q_+\mathbb{Z}^n \text{ or } b = 0\in \mathbb{Z}^n/Q_-\mathbb{Z}^n\\ 
    0 &\text{ otherwise}\end{cases}
\label{Zpmmaptori}
\end{equation}
as well as the remaining ingredients of \eqref{WRTmappingtori}.
The Chern-Simons invariant of the abelian flat connection $a \in \Z^n / Q_+ \Z^n \cong \Z^2 / (U-{\bf 1}) \Z^2$ is
\be
\text{CS} (a) \; = \; - (a, Q_+^{-1} a)
\; = \; - ((U-{\bf 1})^{-1} a, S a)
\ee
and, similarly, for the almost abelian $a \in \Z^n / Q_- \Z^n \cong \Z^2 / (-U-{\bf 1}) \Z^2$ we have
\be
\text{CS} (a) \; = \; - (a, Q_-^{-1} a)
\; = \; ((U+{\bf 1})^{-1} a, S a)
\ee
For convenience, both expressions are written in terms of $n \times n$ matrices $Q_{\pm}$ and in terms of $2 \times 2$ matrix $U$. In the same way we can write the S-matrix \eqref{SpSm}:
\be
\CS_{ab}^{(m)}
\;  = \;  \frac{e^{4\pi i (a, Q_{\pm}^{-1} b)} + e^{- 4\pi i (a, Q_{\pm}^{-1} b)} }{|\text{Stab}_{\mathbb{Z}_2} (a)| \sqrt{|\det Q_{\pm}|}} \label{Spmmaptori}
\;  = \; \frac{e^{4\pi i ((\pm U-{\bf 1})^{-1} a, S b)} + e^{-4\pi i ((\pm U-{\bf 1})^{-1} a, S b)} }{|\text{Stab}_{\mathbb{Z}_2} (a)| \sqrt{|\det (\pm U - {\bf 1})|}}
\ee
where $\text{Stab}_{\mathbb{Z}_2} (a)$ is the stabiliser of $a$ under the action of the Weyl group: $a \to -a$. The fact that \eqref{WRTmappingtori} is satisfied with all these ingredients \eqref{Zpmmaptori}--\eqref{Spmmaptori} is a classic result of~\cite{MR1175494};
see also \cite{Andersen:2011zh,Andersen:2011xe,Andersen:2016vsm,Andersen:2017awz,Andersen:2018lty} for recent work on WRT invariants of mapping tori.

Aiming to understand the precise definition/characterization of almost abelian flat connections on general 3-manifolds, in the rest of this section we extend the notion of the two sets \eqref{almostabviaU} and \eqref{abviaU} to plumbed 3-manifolds with $b_1 > 0$ and 0-surgeries on some knots.


\subsection{Plumbings with loops}
\label{sec:plumbings}

By a theorem of Lickorish and Wallace \cite{MR151948,MR125588}, any closed oriented 3-manifold can be obtained by performing an integral Dehn surgery on a link in $S^3$. Moreover, any two surgery descriptions of the same 3-manifold $M_3$ are related by a sequence of Kirby moves \cite{MR467753}. 
In this section, we focus on a class of 3-manifolds, plumbings with loops, each of which can be conveniently described by a decorated graph. 

For any graph $\Gamma$ whose vertices $v\in V(\Gamma)$ are labeled by some integers $a_v\in \mathbb{Z}$, we can associate a $3$-manifold $Y_\Gamma$ in the following way : 
\begin{enumerate}
	\item Embed the graph $\Gamma$ into $S^3$. 
	\item Add extra $0$-framed unknots along the meridians, one for each generator of $H_1(\Gamma)$.
	\item Replace each vertex $v\in V$ with an unknot $L_v$ with framing $a_v$. Link $L_v$ and $L_w$ when and only when $(v,w)\in E(\Gamma)$. 
	\item The surgery along the link gives us a 3-manifold $Y_\Gamma$. 
\end{enumerate}
\begin{figure}
    \centering
    \includegraphics[scale=0.5]{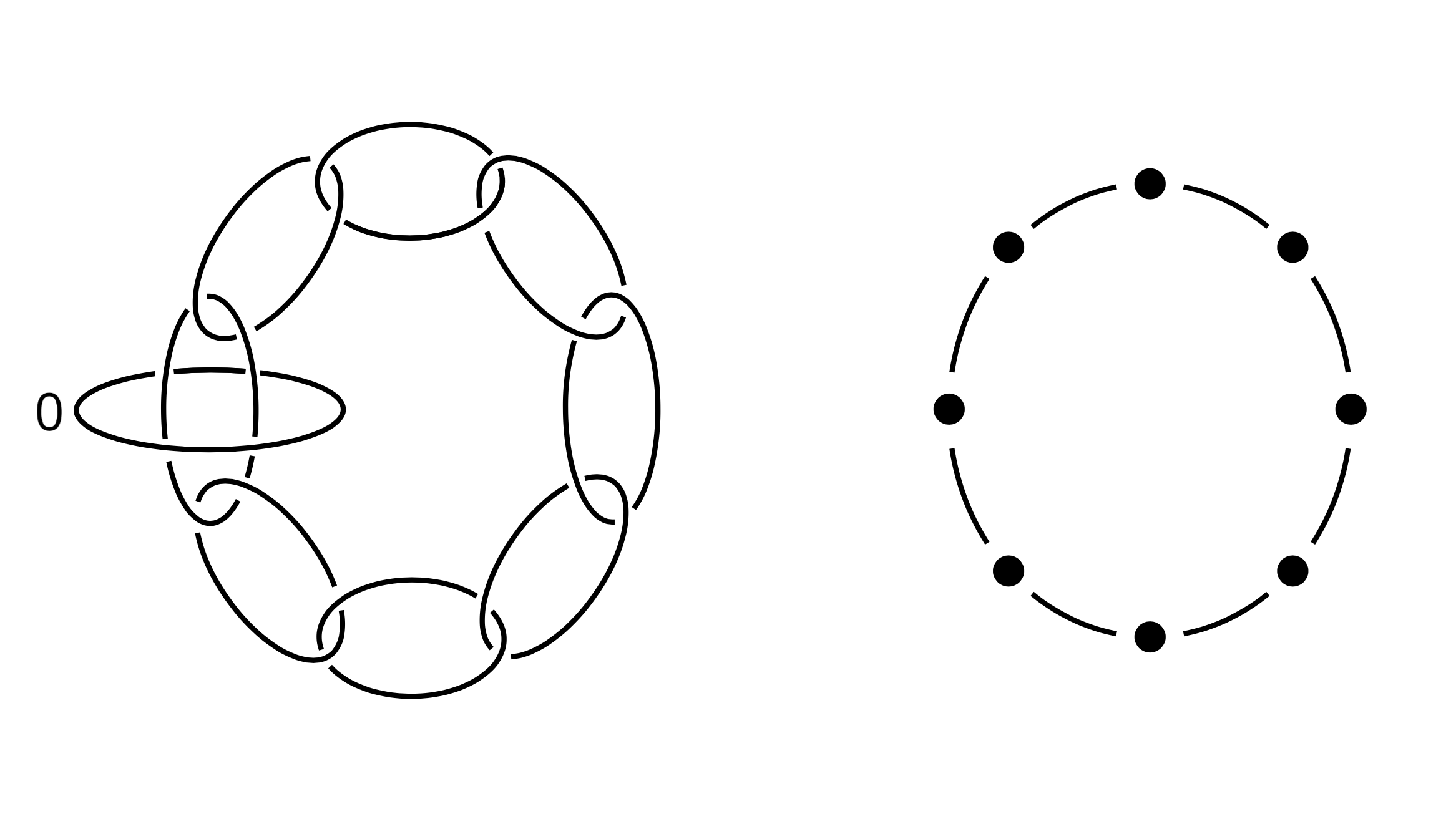}
    \caption{Additional 0-framed unknots for plumbing graph with loops.}
    \label{fig:add0surgery}
\end{figure}
Note that we add $b_1(\Gamma)$ number of $0$-framed unknots. The choice of generators of $H_1(\Gamma)$ doesn't matter, because they're all equivalent via handle slides. Another way to think about this is, rather than adding extra $b_1(\Gamma)$ $0$-framed unknots, we can think of putting the graph $\Gamma$ inside $\#^{b_1(\Gamma)}(S^2\times S^1)$ such that $H_1(\Gamma)\xrightarrow{\iota_*}H_1(\#^{b_1(\Gamma)}(S^2\times S^1))$ is an isomorphism. 

What is more important to note is that there are $2^{b_1(\Gamma)}$ different choices to make a surgery link out of $\Gamma$; they correspond to different choices of twisting each cycle of $\Gamma$. Hence, to make it clear, we should write $Y_{\Gamma,t}$ where $t$ is a twisting data for $\Gamma$. 
A twisting data $t$ of $\Gamma$ assigns to each generator $c\in H_1(\Gamma)$ a number in $t_c\in \frac{1}{4}\mathbb{Z}/\mathbb{Z}$ in such a way that if $c$ is a cycle of length $m$, then $t_c \in (\frac{1}{2}\mathbb{Z}+\frac{m}{4})/\mathbb{Z}$. It can be equivalently thought of as an equivalence class of assignment $t_e\in (\frac{1}{2}\mathbb{Z}+\frac{1}{4})/\mathbb{Z}$ to each edge such that for each cycle $c=[e_1,\cdots,e_m]$, $t_c=\sum_{i=1}^{m} t_{e_i}$.
Although the $3$-manifold $Y_{\Gamma,t}$ itself is dependent on the choice of $t$, its WRT invariant $\mathrm{WRT}(Y_{\Gamma})$ is independent of $t$. 

The first homology group $H_1(Y_\Gamma)$ can be easily computed from the plumbing data. 
\begin{equation}
	H_1(Y_{\Gamma,t})\cong\mathbb{Z}^{b_1(\Gamma)} \times \mathbb{Z}^{|V|}/Q \mathbb{Z}^{|V|}
\end{equation}
Here, the linking matrix $Q$ of $Y_{\Gamma,t}$ is a symmetric $|V|\times |V|$ matrix characterized by
\begin{itemize}
	\item For each $v\in V$, $(v,Qv) = a_v + \sum_{e\in E(v,v)}2(-1)^{2t_e-\frac{1}{2}}$
	\item For each $v\neq w\in V$, $(v,Qw) = \begin{cases} \sum_{e\in E(v,w)}(-1)^{2t_e-\frac{1}{2}} &\text{if }(v,w)\in E \\ 0 &\text{if }(v,w)\neq E \end{cases}$
\end{itemize}
where $E(v,w)$ denotes the set of edges between $v$ and $w$. 
For each $\mathbb{Z}_2$-flat connection on $\Gamma$, $s\in \mathrm{Hom}(\pi_1(\Gamma),\mathbb{Z}_2)\cong \mathbb{Z}_2^{b_1(\Gamma)}$, the corresponding \emph{twisted linking matrix} $Q_s$ is a symmetric $|V|\times |V|$ matrix characterized by the same bullet points, with $t_e$ replaced by $t_e + s_e$. 
In particular, $Q_0 = Q$. 

Assume that $Q_s$ is weakly negative definite for all $s\in \mathbb{Z}_2^{b_1(\Gamma)}$.  Then applying the Gauss sum reciprocity procedure similar to that of \cite{Gukov:2017kmk} we get\footnote{While in this paper we focus on $G=SU(2)$, the generalization to arbitrary gauge group $G$ is the subject of \cite{Chung:2018rea,Park}.}
\be
\text{WRT}(k) \sim \left( \frac{-i}{\sqrt{2k}}  \right)^{1-b_1} \sum_{\text{classes of }s} \sum_{\substack{a \in \Coker Q_{s}/\mathbb{Z}_2 \\ b \in (2 \Coker Q_{s} + \delta)/\mathbb{Z}_2} }    e^{2\pi i k\, CS(a)} \mathcal{S}_{a b}   \hat{Z}^{(s)}_b|_{q \to e^{\frac{2 \pi i}{k}}} 
\ee
where
\be
\delta_v = 2-\deg v, 
\ee
\be
CS^{(s)}(a) = -(a,Q_{s}^{-1}a), 
\ee
\be
\mathcal{S}^{(s)}_{ab} = \frac{e^{2 \pi i (a, Q_{s}^{-1} b)}+e^{- 2 \pi i (a, Q_{s}^{-1}b)}}{|\text{Stab}_{\mathbb{Z}_2}(a)|\sqrt{|\det(Q_{s})|}}.
\ee
Here $\hat{Z}_b$ is given by
\begin{multline}
\label{Zforplumbed}
\hat{Z}^{(s)}_b(q) = \frac{(-1)^{s + b_+^{(0)}}}{2^{b_1(\Gamma)}} e^{\frac{\pi i}{4} \left( \sigma^{(s)}-\sigma^{(0)} \right)}  q^{\frac{3 \sigma^{(0)}-\sum_v a_v}{4}}  \\ \times \text{v.p.} \oint_{|z_v|=1} \prod_v \frac{dz_v}{2 \pi i z_v} \left(z_v - \frac{1}{z_v} \right)^{2-\deg(v)} \Theta^{-Q_{s}}_b(\vec{z};q)
\end{multline}
where 
\be
\Theta^{-Q}_b(\vec{z};q) = \sum_{l \in 2 Q \mathbb{Z}^V + b} q^{-\frac{(l, Q^{-1} l)}{4}} \vec{z}^l.
\ee

We call the set of labels $a$ ``almost abelian'' in the following sense : 
Plumbed 3-manifolds $Y_\Gamma$ are naturally torus-fibered over the plumbing graph $\Gamma$. Each $s\in \mathbb{Z}_2^{b_1(\Gamma)}$ corresponds to a choice of a $W = \mathbb{Z}_2$-flat connection on the plumbing graph. Each label $a$ corresponds to an abelian flat connections on $Y_\Gamma$ twisted by the pull-back of the $\mathbb{Z}_2$-flat connection, hence the name ``almost abelian''. 

\subsubsection{Genus-1 mapping tori}
To illustrate how this works in practice, let's consider genus-1 mapping tori \eqref{UviaST} with $\Tr U \ne \pm 2$. They are special cases of plumbings with loops where the plumbing graph is necklace-shaped. For $U = ST^{a_1}\cdots ST^{a_n}$, the corresponding plumbing graph was described in (\ref{aaaloop}). 
Here let us explain how genus-1 mapping tori have such natural surgery presentations in $S^2\times S^1$. In the following we will think of the horizontal direction as $S^2$ and the vertical direction as $S^1$ (or an interval in $S^1$). Each element of $SL(2,\mathbb{Z})$ can then be expressed as a tangle (using Lickorish generators \cite{K}).
Such correspondence is summarized in the figure below:
\begin{center}
\includegraphics[scale=0.7]{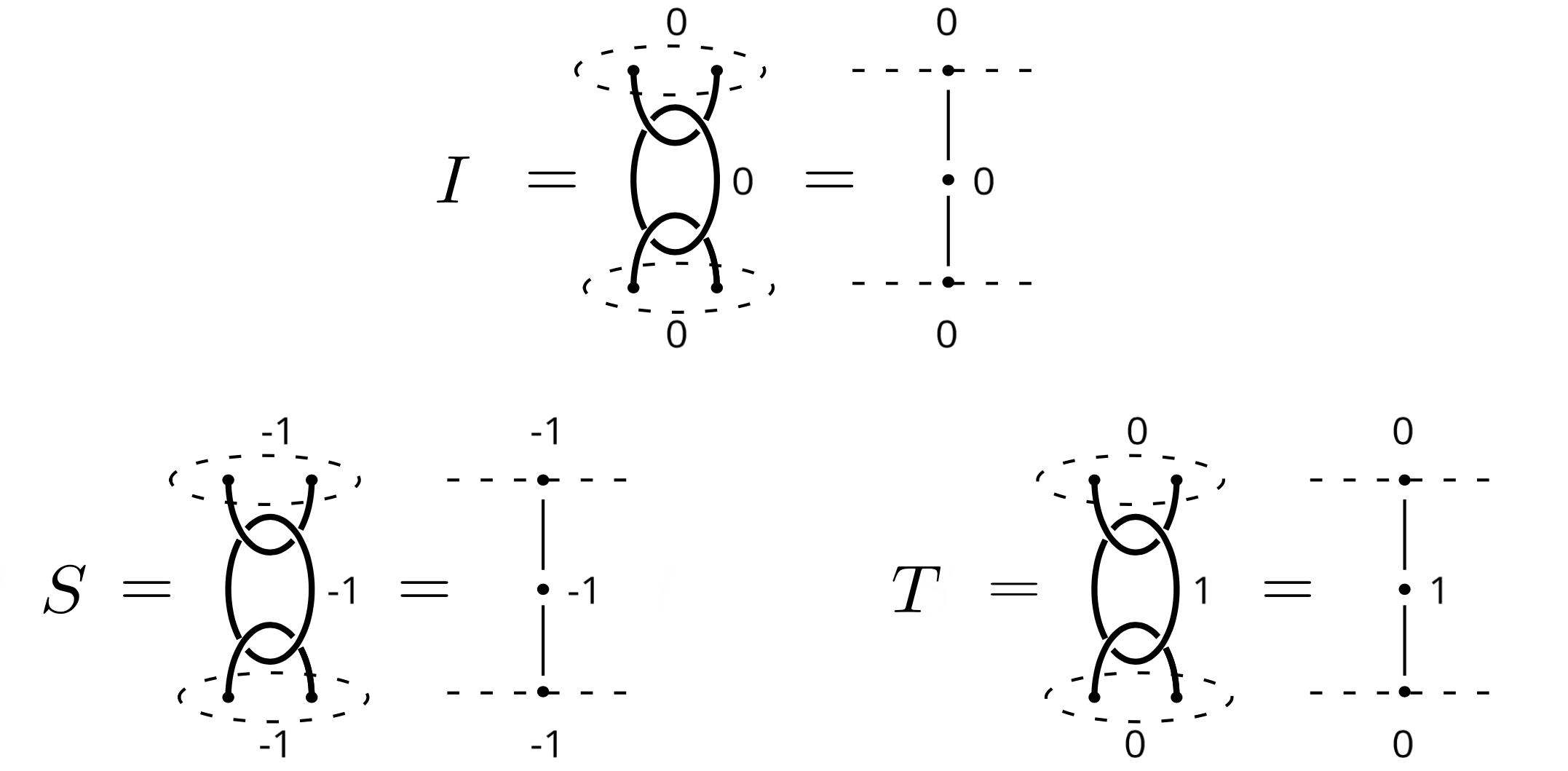}
\end{center}
It is a fun exercise to check that $S^4=I$ and $(S T)^3=S ^2$ under some sequence of 3d Kirby moves. It is good to keep track of twists. For each linking (i.e. an edge for plumbing graph), we assign a number in $\frac{1}{4}\mathbb{Z}/\mathbb{Z}$ as follows : 
\begin{center}
	\includegraphics[scale=0.4]{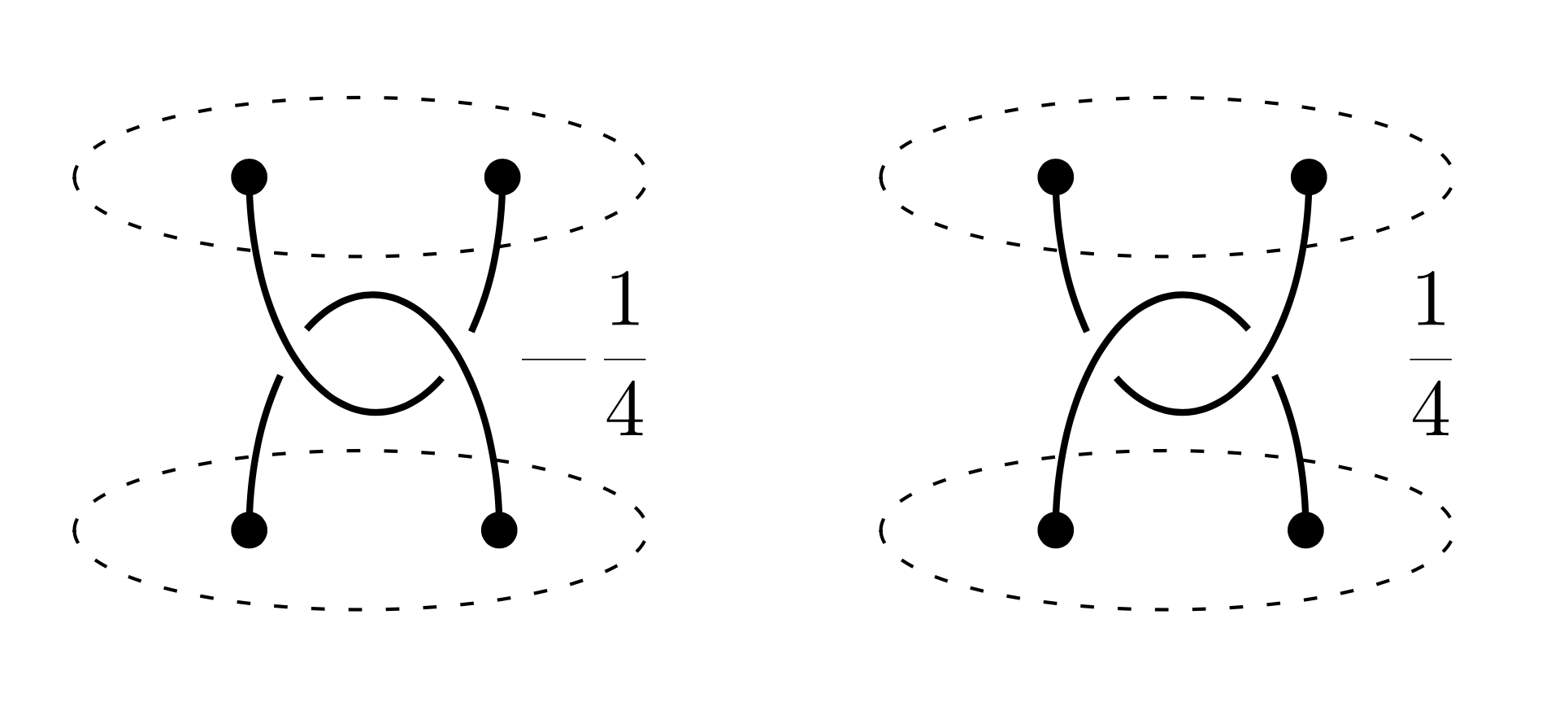}
\end{center}
Of course, for each tangle, only the overall twisting matters. It is easy to check that $S^2=(S T)^3$ is equivalent to the plumbing graph of $I$ but with a half-twist (meaning that it is $-I$). 

Using this dictionary, it is easy to see that the element $ST^{a_1}\cdots ST^{a_n}$ is equivalent to the following tangle : 
\begin{center}
	\includegraphics[scale=0.6]{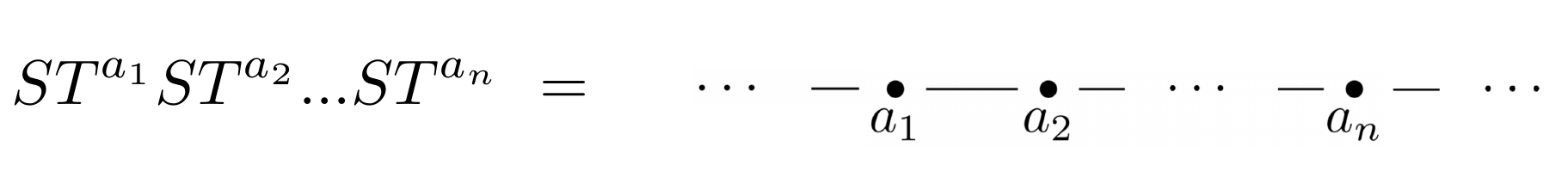}
\end{center}
with the total twist factor $-\frac{n}{4}$. Hence, the mapping torus $M_U$ for $U = ST^{a_1}\cdots ST^{a_n}$ is exactly the manifold obtained by a Dehn surgery on $S^2\times S^1$ along the closure of the above tangle (plumbing).\footnote{By closure, we mean identifying the top $S^2$ with the bottom one.} 
If one needs a surgery presentation inside $S^3$, of course it can be obtained by including an additional 0-surgery as shown on Figure \ref{fig:add0surgery}.
This simple dictionary between torus bundles and plumbing graphs was described in (\ref{mappingtoridict}).\footnote{This correspondence holds for $n\geq 2$, but for $n=0, 1$, we need to the previous (unsimplified) dictionary. For instance, when $U=t^p$, the surgery link in the above description is not a single vertex.}

From this plumbing graph description of $M_U$, it is straightforward to derive (\ref{Zpmmaptori}) and other formulas for genus 1 mapping tori given in subsection \ref{almostabelianflatconnections} from (\ref{Zforplumbed}). 


\subsubsection{Example: a tadpole diagram}

\be
\Gamma \qquad = \qquad
- \!\!\!\!\!\!\!\!\!\!\!\!\!\!\!\! {\raisebox{-0.4cm}{\includegraphics[width=3.0cm]{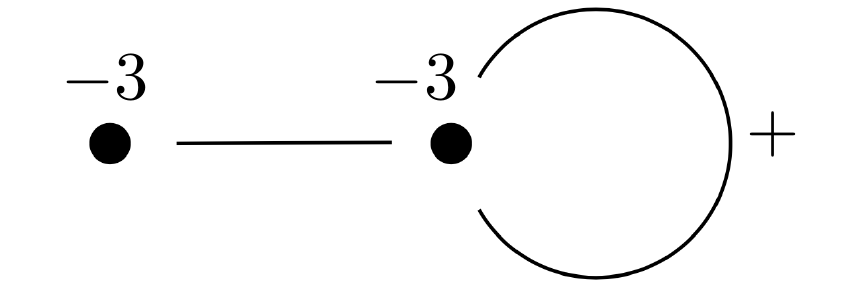}}\,}
\label{tadpoleexample}
\ee

This manifold is neither a genus 1 mapping torus nor a surgery on any knot. It has $H_1 (M_3) \cong \Z \times \Z_3$. In particular, for $G=SU(2)$ one finds two abelian flat connections and 14 almost abelian flat connections (after modding out by Weyl symmetry). The corresponding $q$-series invariants are 
\begin{align}
\hat Z^{(+)}_{0} & = - q^{1/12} \left[ \Psi_{24}^{(8)} - \Psi_{24}^{(16)} \right]
\\
\hat Z^{(+)}_{1} & = q^{1/12} \left[ \Psi_{24}^{(4)} - \Psi_{24}^{(20)} \right]
\nonumber
\end{align}
and
\begin{align}
\hat Z^{(-)}_{0} & = q^{1/12} \left[ \Psi_{168}^{(56)} - \Psi_{168}^{(112)} \right]
\nonumber \\
\hat Z^{(-)}_{1} & = - q^{1/12} \left[ \Psi_{168}^{(28)} - \Psi_{168}^{(140)} \right]
\nonumber \\
\hat{Z}^{(-)}_{2} = \hat{Z}^{(-)}_{3}
& = - \frac{1}{2} q^{1/12} \left[ \Psi_{168}^{(4)} - \Psi_{168}^{(164)} + \Psi_{168}^{(52)} - \Psi_{168}^{(116)} \right]
\nonumber \\
\hat{Z}^{(-)}_{4} = \hat{Z}^{(-)}_{5}
& = - \frac{1}{2} q^{1/12} \left[ \Psi_{168}^{(16)} - \Psi_{168}^{(152)} + \Psi_{168}^{(40)} - \Psi_{168}^{(128)} \right]
\\
\hat{Z}^{(-)}_{6} = \hat{Z}^{(-)}_{7}
& = \frac{1}{2} q^{1/12} \left[ \Psi_{168}^{(32)} - \Psi_{168}^{(136)} + \Psi_{168}^{(80)} - \Psi_{168}^{(88)} \right]
\nonumber \\
\hat{Z}^{(-)}_{8} = \hat{Z}^{(-)}_{9}
& = \frac{1}{2} q^{1/12} \left[ \Psi_{168}^{(44)} - \Psi_{168}^{(124)} + \Psi_{168}^{(68)} - \Psi_{168}^{(100)} \right]
\nonumber \\
\hat{Z}^{(-)}_{10} = \hat{Z}^{(-)}_{11}
& = \frac{1}{2} q^{1/12} \left[ \Psi_{168}^{(8)} - \Psi_{168}^{(160)} + \Psi_{168}^{(104)} - \Psi_{168}^{(64)} \right]
\nonumber \\
\hat{Z}^{(-)}_{12} = \hat{Z}^{(-)}_{13}
& = \frac{1}{2} q^{1/12} \left[ \Psi_{168}^{(20)} - \Psi_{168}^{(148)} + \Psi_{168}^{(92)} - \Psi_{168}^{(76)} \right]
\nonumber
\end{align}
written here in terms of the standard {\it false theta-functions},
\be
\Psi_p^{(a)} (q) \; := \;
\sum_{n \in 2 p \Z+a} \text{sign} (n) \, q^{\frac{n^2}{4 p}}
\label{falsetheta}
\ee


\subsubsection{Example: double loop}

\be
\Gamma \qquad = \qquad
- \!\!\!\!\!\!\!\!\!\!\!\!\!\!\!\! {\raisebox{-0.65cm}{\includegraphics[width=3.0cm]{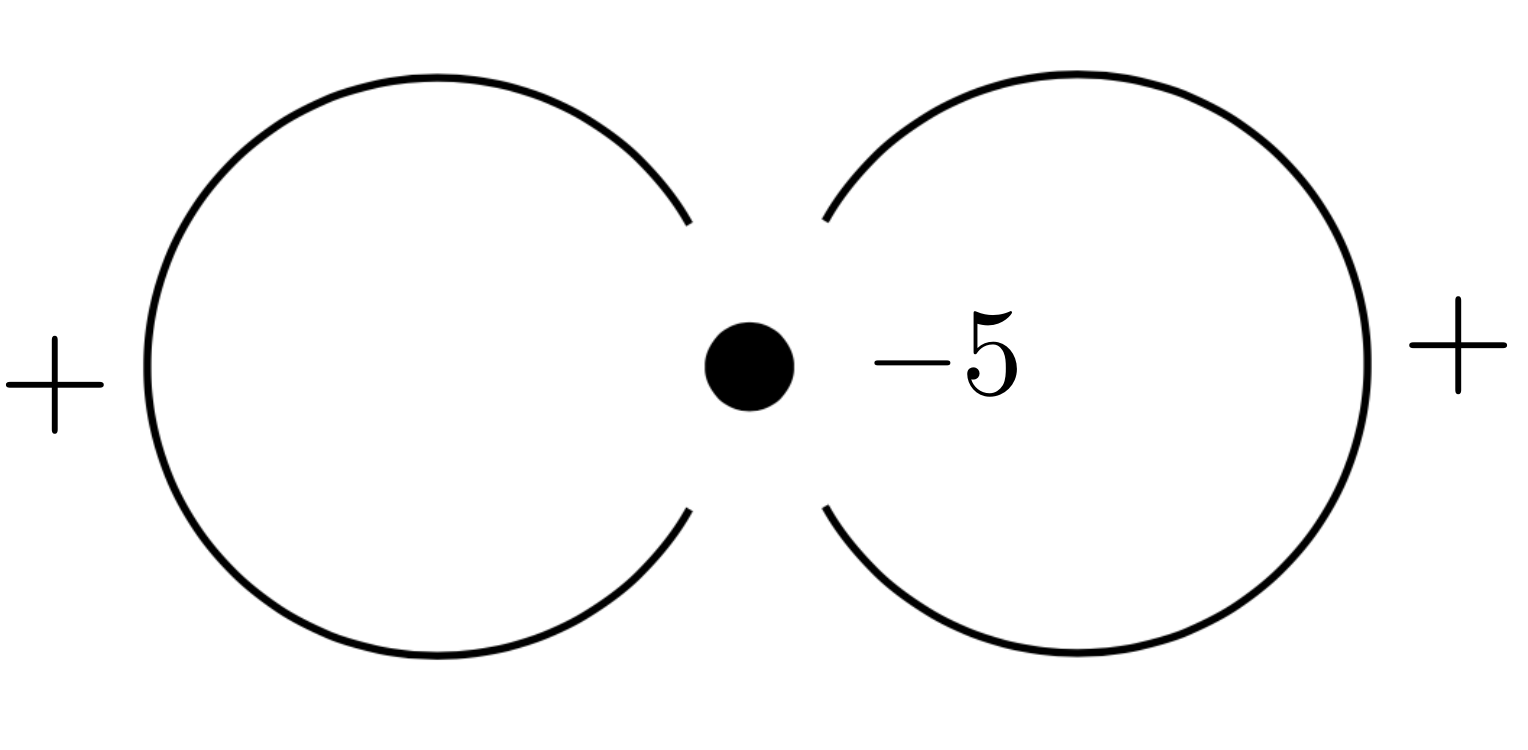}}\,}
\label{tadpoleexample}
\ee

As an example of a 3-manifold with $b_1=2$ we consider a graph with a single vertex and two loops. There are four different classes of abelian and almost abelian flat connections with the following $q$-series

\begin{align}
\hat Z^{(++)}_{(0)} & = \Phi^{(0)}_{1}
\\
\hat{Z}^{(+-)}_{(2r)} = \hat{Z}^{(-+)}_{(2r)}
& = \Phi^{(r)}_{5} \,\,\,\,\, \text{ for  } r=0,...,4 
\nonumber \\
\hat{Z}^{(--)}_{(2r)}
& = \Phi^{(r)}_{9} \,\,\,\,\, \text{ for  } r=0,...,8 
\nonumber \\
\end{align}
where
\be
\Phi^{(r)}_{p}(q) := \frac{q^{1/2}}{2} \sum_{n \in p \mathbb{Z} + r} |n| q^{n^2/p}
\ee
Using zeta-regularized values of $\Phi_{p}^{(r)}(q)$ at roots of unity, not just the radial limit of it, we were able to reproduce WRT invariant for this manifold from $\hat{Z}$.
As in \cite{Kucharski:2019fgh,Chung:2019jgw}, it would be interesting to study more carefully (and compare) the behavior near $q = e^{2 \pi i \tau}$, with different rational values of $\tau \in \mathbb{Q}$.


\subsection{0-surgery on knots}

In this subsection we study 0-surgeries on knots, which are another class of 3-manifolds with $b_1>0$. Although our main examples, 0-surgery on double twist knots, are plumbings with loops, viewing them as 0-surgeries turns out to be fruitful and provides a different perspective. 

\subsubsection{0-surgery on double twist knots $K_{n,m}$}

If we replace in \eqref{tadpoleexample} both of the framing coefficients $-3$ by $2$, we obtain a plumbing diagram for the 0-surgery on ${\bf 5_2}$ knot, that was included in our list of examples in Table~\ref{tab:knots}. More generally, the 0-surgery on a double twist knot $K_{n,m}$ admits a plumbing presentation, illustrated in Figure~\ref{fig:doubletwist}. Our example $K = {\bf 5_2}$ is a special case of this infinite family, with $(n,m)=(2,1)$. While we mostly focus on ${\bf 5_2}$ knot as an example, all the discussions in section can be easily generalized to double twist knots $K = K_{n,m}$. 

\begin{figure}[ht]
	\centering
	\includegraphics[scale=0.6]{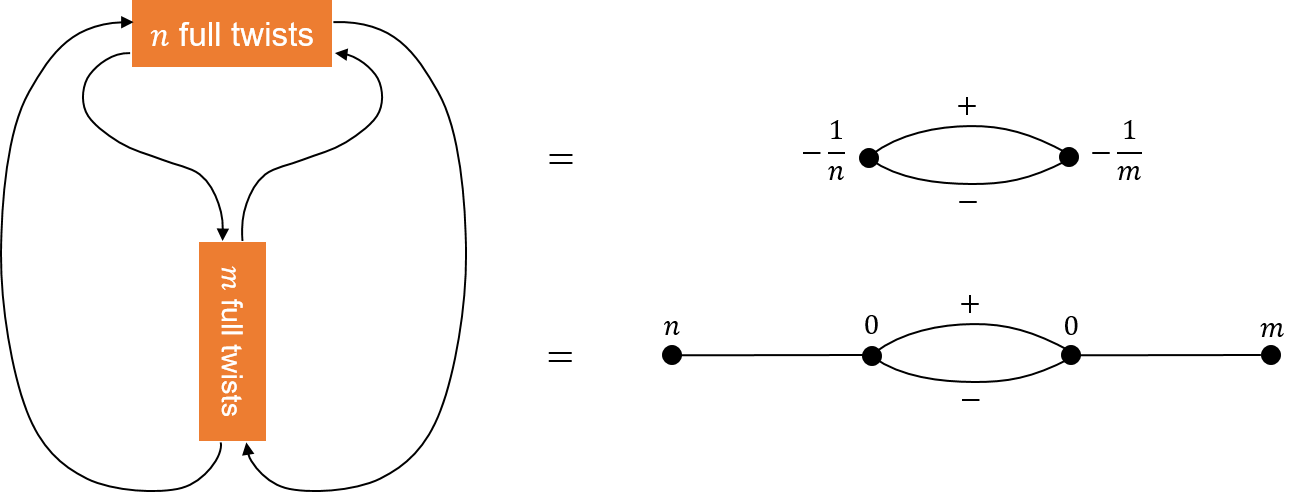}
	\caption{0-surgery on a double twist knot $K_{n,m}$ can be represented by a plumbing diagram with one loop. The ${\bf 5_2}$ knot is a double twist knot $K_{2,1}$.}
	\label{fig:doubletwist}
\end{figure}

As was also mentioned in Table~\ref{tab:knots}, for $G_{\C} = SL(2,\C)$ there is one abelian (trivial) flat connection on $M_3 = S^3_0 ({\bf 5_2})$ and four almost abelian flat connections.\footnote{In particular, since this plumbing has one loop there are only two values of $s$, $s=+$ and $s=-$, just like for genus-1 mapping tori with $b_1=1$.} Therefore, in this example, we expect one $q$-series invariant $\hat Z^{(+)}$ and four $q$-series invariants $\hat Z^{(-)}$. They can be easily computed from the plumbing graph presentation by using the general formula \eqref{Zforplumbed}. Let us start with the $q$-series invariant labeled by a single abelian flat connection. We find
\begin{equation}
\hat Z^{(+)}_{0} (q) \; = \; - q^{9/8} \, \Psi_{2}^{(1)}
\label{Zplus52}
\end{equation}
For the $q$-series invariants labeled by almost abelian flat connections, the direct application of \eqref{Zforplumbed} gives
\begin{align}
\hat Z^{(-)}_{0} & = - q^{9/8} \, \Psi_{14}^{(7)} \bigg\vert_{q \rightarrow q^{-1}}
\nonumber \\
\hat Z^{(-)}_{1} & = \frac{1}{2} q^{9/8} \left[ \Psi_{14}^{(1)} + \Psi_{14}^{(13)} \right] \bigg\vert_{q \rightarrow q^{-1}}
\\
\hat Z^{(-)}_{2} & = - \frac{1}{2} q^{9/8} \left[ \Psi_{14}^{(3)} + \Psi_{14}^{(11)} \right] \bigg\vert_{q \rightarrow q^{-1}}
\nonumber \\
\hat Z^{(-)}_{3} & = \frac{1}{2} q^{9/8} \left[ \Psi_{14}^{(5)} + \Psi_{14}^{(9)} \right] \bigg\vert_{q \rightarrow q^{-1}}
\label{Zminus52}
\end{align}
which, following the rules of \cite{Cheng:2018vpl}, should be further re-written as a $q$-series.

It is instructive to compare these results with the surgery formula based on a two-variable knot invariant $F_K (x,q)$ which is annihilated by the quantum A-polynomial and can be defined using Borel resummation (parametric resurgence) of the $n$-colored Jones polynomial \cite{Gukov:2019mnk}:
\be
J^K_n(e^{\hbar}) = \sum_{m=0}^{\infty} \sum_{j=0}^{m} c_{m,j} n^j \hbar^m
\; \;\;\stackrel{\text{Borel resum}}{=\joinrel=\joinrel=}\;\; \;
\frac{F_K(x,q)}{x^{1/2}-x^{-1/2}}
\label{JonesBorel}
\ee
This Borel-resummed invariant $F_K(x,q)$ looks quite different than the colored Jones polynomials themselves, but we conjecture that, analogously to Conjecture \ref{conjWRT}, it contains all the information of asymptotic expansions of colored Jones polynomials near each root of unity : 
\begin{conj}
    For any color $n$, the asymptotic expansion of $(q^{\frac{n}{2}} - q^{-\frac{n}{2}})J_n(q)$ and $F_K(q^n,q)$ near each root of unity agree.\footnote{This conjecture is not in contradiction with the volume conjecture, because $F_K(x,q)$ is not defined on $q=e^{2\pi i z}$ with $z\in \mathbb{R}\setminus\mathbb{Q}$. Hence the limit $\lim_{q\rightarrow 1} F_K(x,q) = \frac{x^{1/2}-x^{-1/2}}{\Delta_K(x)}$ and the limit taken in the volume conjecture are two limits of different nature. (The first one is radial, whereas the second one is along the unit circle.)} That is, for each root of unity $\zeta$, 
    \begin{equation}
        (\zeta^{\frac{n}{2}}e^{\frac{n\hbar}{2}}-\zeta^{-\frac{n}{2}}e^{-\frac{n\hbar}{2}})J_n(\zeta e^\hbar)\;\stackrel{\mathrm{perturbatively}}{=\joinrel=\joinrel=}F_K(\zeta^n e^{n\hbar},\zeta e^\hbar)\;
    \end{equation}
\end{conj}

Let's connect $F_K$ with the $q$-series we obtained above for a 0-surgery. 
Writing
\be
F_K (x,q) \; = \; \frac{1}{2} \sum_{m=1}^{\infty} f^K_m (q) \cdot \bigl(x^{\frac{m}{2}} - x^{-\frac{m}{2}} \bigr)
\ee
the surgery formula applied to $M_3 = S^3_{-1/r} (K)$ gives
\be
\hat Z^{(ab)}_{0} \big( S^3_{-1/r} (K) \big)
\; = \; q^{\frac{1}{4r}} \sum_{m=1}^{\infty}  q^{\frac{r m^2}{4} - \frac{m}{2}} (q^m - 1) f_m^K (q)
\ee
As the surgery coefficient $-1/r \to 0$, {\it i.e.} as $r \to \infty$, terms with different values of $m$ in this formula become separated by infinitely large powers of $q$. Therefore, it is natural to expect that for 0-surgery on the knot $K$, only one value of $m$ contributes to $\hat Z^{(+)}_{0} \big( S^3_{0} (K) \big)$, namely $m=1$ for which the accompanying $q$-power is the lowest. We formalize this in the form of the following conjecture:\footnote{For $q$-series, we use notation $\cong$ to denote equivalence up to sign and overall $q$-power.}

\begin{conj}[0-surgery formula]\label{0-surgery formula}
	For any knot $K$ in $S^3$,
	\begin{equation}
	\hat Z_{0}^{(ab)} \big( S^3_{0} (K) \big) \; \cong \; \, \frac{1}{2}f_1^K (q)
	\end{equation}
	In particular, we have
    \begin{equation}
    \hat Z^{(+)}_0 (S^3_0 (K_{m,n})) \; \cong \; \frac{1}{2} \Psi_{m}^{(m-1)} \Psi_{n}^{(n-1)}
    \end{equation}
\end{conj}

This prediction indeed agrees with the direct calculation of $f_1^{{\bf 5_2}}$ based on \eqref{JonesBorel} and the double expansion of the colored Jones polynomial:
\begin{align}
2 \cdot J_n({\bf 5_2}) &= \left( -\frac{1}{2} + \frac{5}{8}\hbar - \frac{13}{16}\frac{\hbar^2}{2!} + \frac{73}{64}\frac{\hbar^3}{3!} - \frac{121}{64}\frac{\hbar^4}{4!} + \frac{1135}{256}\frac{\hbar^5}{5!} - \frac{9161}{512}\frac{\hbar^6}{6!} - \cdots \right)\nonumber\\
&+ \frac{x^{3/2}-x^{-3/2}}{x^{1/2}-x^{-1/2}} \left(-\frac{1}{4} + \frac{7}{16}\hbar - \frac{25}{32}\frac{\hbar^2}{2!}+ \frac{191}{128}\frac{\hbar^3}{3!} - \frac{433}{128}\frac{\hbar^4}{4!} + \frac{5549}{512}\frac{\hbar^5}{5!} - \cdots \right) \nonumber\\
&+ \frac{x^{5/2}-x^{-5/2}}{x^{1/2}-x^{-1/2}} \left( \frac{1}{8} - \frac{1}{8}\hbar - \frac{7}{64}\frac{\hbar^2}{2!} + \frac{173}{128}\frac{\hbar^3}{3!} - \frac{1801}{256}\frac{\hbar^4}{4!} + \frac{1151}{32}\frac{\hbar^5}{5!} - \cdots \right) \nonumber\\
&+ \frac{x^{7/2}-x^{-7/2}}{x^{1/2}-x^{-1/2}} \left( \frac{7}{16} - \frac{35}{32}\hbar + \frac{315}{128}\frac{\hbar^2}{2!} - \frac{521}{128}\frac{\hbar^3}{3!} - \frac{1563}{512}\frac{\hbar^4}{4!} + \frac{96545}{1024}\frac{\hbar^5}{5!} - \cdots \right) \nonumber\\
&+ \cdots \nonumber
\end{align}
which leads to 
\begin{align}
    f_1^{\mathbf{5}_2} &= -\frac{1}{2} + \frac{5}{8}\hbar - \frac{13}{16}\frac{\hbar^2}{2!} + \frac{73}{64}\frac{\hbar^3}{3!} - \frac{121}{64}\frac{\hbar^4}{4!} + \frac{1135}{256}\frac{\hbar^5}{5!} - \frac{9161}{512}\frac{\hbar^6}{6!} - \cdots \nonumber\\
    &= -q^{-1} + 1 -q^2 +q^5 -q^9 +q^{14} -q^{20} +q^{27} -q^{35} +q^{44} -\cdots \nonumber\\
    &\cong \Psi_2^{(1)} (q)
    \label{f152}
\end{align}
It would be interesting to determine other $f_m^{{\bf 5_2}}$ from this double expansion.

Note, Conjecture~\ref{0-surgery formula} only produces $\hat Z_{0}^{(ab)} \big( S^3_{0} (K) \big)$ labeled by the abelian flat connection from the two-variable series $F_K(x,q)$. This conjecture can be equivalently formulated as a statement:
\begin{equation}
Z_0^{(ab)}(S_0^3(K)) \; \cong \; \mathrm{Res}_{x=0}\frac{x^{1/2}-x^{-1/2}}{x}F_K(x,q) \,.
\end{equation}
Motivated by this fact and the work of Ohtsuki \cite{Ohtsuki}, where perturbative invariants of 3-manifolds with $b_1 = 1$ are defined, we can try to relate our almost abelian $\hat{Z}$'s with other residues of $F_K(x,q)$. Indeed, somewhat surprisingly we find that, for twist knots, we can recover $Z_0^{(-)}(S_0^3(K))$ solely from $F_K$ by looking at residues at the roots of $\Delta_K(x)$: 
\begin{conj}\label{residue formula}
    \begin{equation}
    Z^{(-)}_0(S_0^3(K_n)) \; \cong \; \mathrm{Res}_{x = x_0} \frac{x^{1/2}-x^{-1/2}}{x}F_{K_n}(x,q)
    \label{resflaone}
    \end{equation}
    where $x_0$ is a solution\footnote{It doesn't matter which solution we choose, as their residues only differ by sign.} to $\Delta_{K_n}(x) = nx + nx^{-1} - (2n-1) = 0$.
\end{conj}
\noindent
This conjecture suggests a possible connection between almost abelian flat connections and the ``unipotent branch'' discussed in section \ref{sec:sutwomoduli}.
Moreover, at least for twist knots, we seem to be able to recover all $Z^{(-)}$'s by a similar procedure:
\begin{conj}
    \begin{equation}
    Z^{(-)}_a(S_0^3(K_n)) \; \cong \; \mathrm{Res}_{x=x_0}\frac{x^{1/2}-x^{-1/2}}{x}F_{K_n}(x,q) \bigg\vert_{q\rightarrow e^{2\pi i a}q}\footnote{The change of variables $q\rightarrow e^{2\pi i a}q$ is non-trivial because the residue is not a power series in $q$ but in $q^{\frac{1}{1-4n}}$.}
    \label{ZKnresidue}
    \end{equation}
    where $x_0$ is a solution to $\Delta_{K_n}(x) = 0$, as in Conjecture~\ref{residue formula}.
\end{conj}
These two conjectures imply that all the information about $\hat Z_a^{(\pm)} (S_0^3 (K_n))$, labeled by both abelian and almost abelian flat connections, is contained in the two-variable series $F_{K_n} (x,q)$. They also suggest that residues in \eqref{resflaone} and \eqref{ZKnresidue} arise from an $x$-integral {\it a la} \eqref{Zforplumbed}. Such an integral would naturally have an interpretation of the integral over the Coulomb branch of 3d $\CN=2$ theory $T[M_3]$. As we learned in section~\ref{sec:Mflat}, the roots of the Alexander polynomial are precisely the points on the Coulomb branch of $T[M_3]$ where it meets Higgs branch(es). Therefore, the poles of the $x$-integral with residues \eqref{resflaone} and \eqref{ZKnresidue} have a very natural physical interpretation as Higgs branch contributions to the integral over the Coulomb branch. This is similar (and perhaps even related!) to the Higgs branch contributions to $u$-plane integrals of 4d $\CN=2$ topologically twisted theories~\cite{Moore:1997pc,Losev:1997tp}.

Below we provide some evidence to Conjecture~\ref{residue formula} for the first few twist knots :

\noindent
$K_1 = \mathbf{3}_1$:
\[
\text{Res}_{x=(-1)^{1/3}}\frac{x^{1/2}-x^{-1/2}}{x}F_{\mathbf{3}_1}(x,q) \; = \; \frac{i}{\sqrt{3}}q \; \cong \; Z^{(-)}_0
\]
$K_{-1} = \mathbf{4}_1$:
\[
\text{Res}_{x=\frac{3+\sqrt{5}}{2}}\frac{x^{1/2}-x^{-1/2}}{x}F_{\mathbf{4}_1}(x,q) \; = \; -\frac{1}{\sqrt{5}} \; \cong \; Z^{(-)}_0
\]
$K_2 = \mathbf{5}_2$:
\begin{align*}
    &\text{Res}_{x=\frac{3-i\sqrt{7}}{4}}\frac{x^{1/2}-x^{-1/2}}{x}F_{\mathbf{5}_2}(x,q) \\ 
    \quad\quad &\; = \; \frac{i}{\sqrt{7}}\left(\frac{1}{2}-\frac{31 \hbar }{2^3 7}+\frac{481 \hbar ^2}{\left(2^4 7^2\right) 2!}-\frac{14939 \hbar ^3}{\left(2^6 7^3\right) 3!}+\frac{116077 \hbar ^4}{\left(2^6 7^4\right) 4!}-\cdots\right) \\ 
    \quad\quad &\; = \; \frac{i}{\sqrt{7}} q^{9/8} \left( \Psi_{14}^{(1)}-\Psi_{14}^{(3)}+\Psi_{14}^{(5)}-\Psi_{14}^{(7)} +\Psi_{14}^{(9)}-\Psi_{14}^{(11)}+\Psi_{14}^{(13)} \right) \bigg\vert_{q \rightarrow 1/q}  \\ 
    \quad\quad &\; \cong \; Z^{(-)}_0
\end{align*}
$K_3 = \mathbf{7}_2$:
\begin{align*}
    &\text{Res}_{x=\frac{5+i\sqrt{11}}{6}}\frac{x^{1/2}-x^{-1/2}}{x}F_{\mathbf{7}_2}(x,q) \\ 
    \quad &\; = \; \frac{i}{\sqrt{11}}\left(\frac{1}{3}-\frac{130 \hbar }{3^3 11}+\frac{5638 \hbar ^2}{\left(3^4 11^2\right) 2!}-\frac{81562 \hbar ^3}{\left(3^4 11^3\right) 3!}+\frac{31876978 \hbar ^4}{(3^7 11^4)4!}-\ldots \right) \\ 
    \quad &\; = \; \frac{i}{\sqrt{11}}q^{4/3} \left(\Psi_{33}^{(4)}+ \Psi_{33}^{(10)}+\Psi_{33}^{(16)}+\Psi_{33}^{(22)}+\Psi_{33}^{(28)}-\Psi_{33}^{(32)}-\Psi_{33}^{(26)}-\Psi_{33}^{(20)}-\Psi_{33}^{(14)}-\Psi_{33}^{(8)}-\Psi_{33}^{(2)} \right)\bigg\vert_{q \rightarrow 1/q}\\ 
    \quad &\; \cong \; Z^{(-)}_0
\end{align*}
Notice how in all these examples the residues of $F_K (x,q)$ come accompanied by the appropriate prefactor $S_{0b}$ of the S-matrix \eqref{Spmmaptori}!

It would be interesting to study whether there is a similar structure for other knots. We speculate that, for an arbitrary knot $K$, the roots of $\Delta(x)=0$, modulo $x \leftrightarrow x^{-1}$, determine different ``almost abelian sectors,'' and 
\[
\mathrm{Res}_{x=x_0}\frac{x^{1/2}-x^{-1/2}}{x}F_K(x,q)
\]
can be resummed into a power series in $q^{1/r}$, for some integer $r$ which divides the number of $b$ labels in that sector (before modding out by the Weyl symmetry). 
As an example, we study 0-surgery on torus knots in the next subsubsection.


\subsubsection{0-surgery on torus knots $T_{s,t}$}
As our first examples of higher genus mapping tori, we study 0-surgery on torus knots $T_{s,t}$; they are mapping tori of genus 
\[g=\frac{(s-1)(t-1)}{2}.\]
Following our observation in the last subsubsection, we expect that each root of $\Delta(x)=0$ determine an ``almost abelian sector''. Indeed, the residues of $\frac{x^{1/2}-x^{-1/2}}{x}F_{T_{s,t}}(x,q)$ all turn out to be monomials in $q$ of the same power, multiplied by some prefactor : 
\begin{equation}
\mathrm{Res}_{x=x_0}\frac{x^{1/2}-x^{-1/2}}{x}F_{T_{s,t}}(x,q) = C\; q^{\frac{(s^2-1)(t^2-1)}{24st}}
\end{equation}
where $C=\mathrm{Res}_{x=x_0}\frac{(x^{1/2}-x^{-1/2})^2}{x\Delta_{K}(x)}$.\footnote{Yamaguchi's theorem \cite{Y} says that this prefactor is basically the square root of the corresponding Reidemeister torsion.} 
Therefore we can try to decompose the WRT invariant into contributions of each pole. We find that the WRT invariants can be decomposed into the following form : 
\begin{equation}
\mathrm{WRT}(S_0^3(T_{s,t}),k) \;=\; q^{\frac{(s^2-1)(t^2-1)}{24st}}(\sum_{j=0}^{st-1}a_j e^{2\pi i k \frac{j}{st}})
\end{equation}
where $a_j$ are $\mathbb{Q}[\zeta_{st}]$-linear combinations of residues of $\frac{(x^{1/2}-x^{-1/2})^2}{x\Delta_{T_{s,t}}(x)}$. (There are $g=\frac{(s-1)(t-1)}{2}$ poles up to Weyl symmetry $x\leftrightarrow x^{-1}$.)
We provide some explicit decompositions below. In the following, we'll use notation $r_{\tau}$ to denote $\mathrm{Res}_{x=e^{2\pi i \tau}}\frac{(x^{1/2}-x^{-1/2})^2}{x\Delta_{T_{s,t}}(x)}$.
\begin{itemize}
    \item $K = T_{3,2}$ : 
\begin{align*}
\mathrm{WRT}(S_0^3(T_{3,2}),k) &= -\frac{q}{2}\left( 1 + r_{\frac{1}{6}}(1+2e^{2\pi i k \frac{1}{3}})\right)
\end{align*}
where $r_{\frac{1}{6}} = \frac{i}{\sqrt{3}}$.
    \item $K = T_{5,2}$ : 
\begin{align*}
&\mathrm{WRT}(S_0^3(T_{5,2}),k)\\ 
&= -\frac{q^{\frac{9}{5}}}{2} \left((r_{\frac{1}{10}}-r_{\frac{7}{10}}) + (r_{\frac{1}{10}}+r_{\frac{7}{10}})e^{2\pi i k \frac{1}{5}} + (-r_{\frac{1}{10}}-r_{\frac{7}{10}})e^{2\pi i k \frac{2}{5}} + (-r_{\frac{1}{10}}+r_{\frac{7}{10}})e^{2\pi i k \frac{3}{5}}\right)\\
&= -\frac{q^{\frac{9}{5}}}{2} \left( r_{\frac{1}{10}}(1 + e^{2\pi i k \frac{1}{5}} - e^{2\pi i k \frac{2}{5}} - e^{2\pi i k \frac{3}{5}}) + r_{\frac{7}{10}}(-1 + e^{2\pi i k \frac{1}{5}} - e^{2\pi i k \frac{2}{5}} + e^{2\pi i k \frac{3}{5}})\right)
\end{align*}
where $r_{\frac{1}{10}} = \frac{i}{5}\sqrt{5-2\sqrt{5}}$, $r_{\frac{7}{10}} = \frac{i}{5}\sqrt{5+2\sqrt{5}}$. 
    \item $K=T_{4,3}$ : 
    \begin{align*}
    &\mathrm{WRT}(S_0^3(T_{4,3}),k) \\
    &= -\frac{q^{\frac{5}{2}}}{2} \left((r_{\frac{1}{12}}-r_{\frac{10}{12}} +r_{\frac{5}{12}}) + (-r_{\frac{1}{12}}+r_{\frac{10}{12}} +r_{\frac{5}{12}})e^{2\pi i k \frac{2}{12}} +2(r_{\frac{1}{12}}+r_{\frac{5}{12}})e^{2\pi i k\frac{3}{12}}\right.\\
    &\quad\quad\quad\left.+ (r_{\frac{1}{12}}+r_{\frac{10}{12}} +r_{\frac{5}{12}})e^{2\pi i k \frac{6}{12}}+ (r_{\frac{1}{12}}-r_{\frac{10}{12}} -r_{\frac{5}{12}})e^{2\pi i k \frac{8}{12}}\right)\\
    &= -\frac{q^{\frac{5}{2}}}{2}\left(r_{\frac{1}{12}}(1-e^{2\pi i k \frac{2}{12}}+2e^{2\pi i k \frac{3}{12}} + e^{2\pi i k \frac{6}{12}} +e^{2\pi i k \frac{8}{12}}) +r_{\frac{10}{12}}(-1+e^{2\pi i k \frac{2}{12}}+ e^{2\pi i k \frac{6}{12}} -e^{2\pi i k \frac{8}{12}})\right.\\
    &\quad\quad\quad\left. +r_{\frac{5}{12}}(1+e^{2\pi i k \frac{2}{12}}+2e^{2\pi i k \frac{3}{12}} + e^{2\pi i k \frac{6}{12}} -e^{2\pi i k \frac{8}{12}}) \right)
    \end{align*}
    where $r_{\frac{1}{12}} = i|r_{\frac{1}{12}}|$, $r_{\frac{10}{12}} = \frac{i}{2\sqrt{3}}$, $r_{\frac{5}{12}} = i|r_{\frac{5}{12}}|$, and $r_{\frac{5}{12}} = r_{\frac{1}{12}} + r_{\frac{10}{12}}$. 
\end{itemize}
The geometric meaning of these decompositions is not so clear at this moment, but it is clear from these examples that the residues of $\frac{(x^{1/2}-x^{-1/2})^2}{\Delta_K(x)}$ should play an important role in analytic continuation of the WRT invariants for general 3-manifolds and its decomposition into $q$-series.


\subsubsection{Renormalon effects}\label{renormalon}

The $q$-series invariants $\hat Z_b$ --- or, rather, their linear combinations $Z_a = \sum_b \CS_{ab} \hat Z_b$, without a hat --- are closely related to Borel resummation of the perturbative series $Z_a^{\text{pert}} (\hbar)$ in complex Chern-Simons theory around a flat connection ``$a$'' of the form we already encountered {\it e.g.} in \eqref{JonesBorel} and \eqref{f152}.
In particular, the resurgent analysis in complex Chern-Simons theory \cite{Gukov:2016njj} suggests that the Borel transform
\be
BZ_{\alpha}^{\text{pert}} (\xi) \; \sim \; \sum_{\alpha} \frac{n_{\alpha \beta}}{\xi - \text{CS} (\beta)}
\label{BZpert}
\ee
of a perturbative expansion $Z_{\alpha}^{\text{pert}} (\hbar)$ around a flat connection $\alpha$ has poles at values of the complex Chern-Simons invariant $\text{CS}$ on $M_3$. On the universal cover, each pole is replicated infinitely many times, in agreement with the fact that $\text{CS} (\beta)$ is only defined modulo $\Z$. Even though individual residues of these poles are integer, this integrality is usually lost after taking a sum over the infinite set of poles required in Borel resummation, so that regularized values of $n_{\alpha \beta}$ are no longer integer.
Moreover, some flat connections $\beta$ have the property that $n_{\alpha \beta} = 0$ for any other $\alpha$, while $n_{\beta \alpha}$ can be non-zero.\footnote{In particular, it is typical to see $n_{\alpha \beta} \ne n_{\beta \alpha}$.} Such $\beta$ then become the labels (usually denoted $a$ or $b$) of the $q$-series invariants $\hat Z (M_3)$.

The structure outlined here has been so far tested and made explicit for many 3-manifolds with $b_1 (M_3) = 0$. Here, we make initial steps toward extending this analysis to 3-manifolds with $b_1 (M_3) > 0$. For example, as an obvious generalization of the above criterion, it is natural to expect that the set of labels $\{ a \}$ of $Z_a (M_3)$ consists of those flat connections on $M_3$ --- or, perhaps, more generally, poles on the Borel plane --- such that they do not appear as transseries in resummation of all other flat connections. In other words, the corresponding Stokes/transseries coefficients should vanish
\be
n_{\beta a} \; = \; 0
\ee
for any $\beta$. (Note, $n_{a \beta}$ may still be non-zero!) It would be interesting to study this more systematically and to see if this can provide a more general (geometric) characterization of ``almost abelian'' flat connections on arbitrary 3-manifolds.

While our initial analysis indicates that, with some obvious adaptations, much of this structure applies to 3-manifolds with $b_1>0$, we observe new peculiar features. For example, we find the expected pattern of the poles in the Borel plane, albeit at shifted values. Without a more detailed analysis, it is not clear whether these poles are ``spurious.'' (We will see a hint for this shortly.) Or, if the poles are all meaningful, of the form \eqref{BZpert}, and there is a systematic explanation for the shift in their position.

Consider, for example, the ``pertutbative'' expansion of the $q$-series invariants $\hat Z^{(+)}_{0} (q)$ and $\hat Z^{(-)}_{a} (q)$ for the 0-surgery on the ${\bf 5_2}$ knot, $M_3 = S^3_0 ({\bf 5_2})$, in the limit $q = e^{\hbar} = e^{2\pi i /k} \to 1$.
For the trivial flat connection, the perturbative expansion of $Z^{(+)}_{0} \big( S^3_0 ({\bf 5_2}) \big)$ looks like
\begin{align*}
Z^{(+)}_{0} \big( S^3_0 ({\bf 5_2}) \big) &= \frac{1}{2} f_1(\mathbf{5}_2) = -\frac{1}{2}q^{-1}\sum_{j=0}^{\infty}(-1)^j q^{\frac{j(j+1)}{2}} = -\frac{1}{2}q^{-\frac{9}{8}} \Psi_2^{(1)} (q)\\
&\stackrel{\hbar \to 0}{=}~ -\frac{1}{2}q^{-\frac{9}{8}}\left(\frac{1}{2} - \frac{1}{2^4}\hbar + \frac{5}{2^7}\frac{\hbar^2}{2!} - \frac{61}{2^{10}}\frac{\hbar^3}{3!} + \frac{1385}{2^{13}}\frac{\hbar^4}{4!} - \frac{50521}{2^{16}}\frac{\hbar^5}{5!} + \cdots \right)\\ 
&~=~ -\frac{1}{4} + \frac{5}{16}\hbar - \frac{13}{32}\frac{\hbar^2}{2!} + \frac{73}{128}\frac{\hbar^3}{3!} - \frac{121}{128}\frac{\hbar^4}{4!} + \frac{1135}{512}\frac{\hbar^5}{5!} - \frac{9161}{1024}\frac{\hbar^6}{6!} + \cdots
\end{align*}
Multiplying by the factor $\frac{1}{\sqrt{k}}$, which plays an important role and also simplifies the Borel transform, we get
\[B \left( \frac{1}{\sqrt{k}} \Psi_2^{(1)} \right)(\xi) \; = \; \frac{1}{\sqrt{\pi \xi}}\frac{1}{e^{\sqrt{\pi i \xi}} + e^{-\sqrt{\pi i \xi}}}\]
This expression has poles at $\sqrt{\pi i \xi} = \pi i(n + \frac{1}{2})$, or
\be
\frac{\xi}{2\pi i} \; = \; \frac{(2n+1)^2}{8} \; \equiv \; \frac{1}{8} \text{ mod }\mathbb{Z}.
\label{52polesa}
\ee
Similarly, for $a=0,1,2,3$, we have:
\begin{align*}
Z^{(-)}_{a} \big( S^3_0 ({\bf 5_2}) \big) &\cong -\frac{i}{2\sqrt{7}} \Bigl[ q^{3/8}
\sum_{b=1, 3, 5, \cdots, 13}
e^{-2\pi i \frac{a b}{7}}
(-1)^{\frac{b-1}{2}} \Psi_{14}^{(b)} (q)
\Bigr]\bigg|_{q\rightarrow q^{-1}}\\
&= -\frac{i}{2\sqrt{7}}\Bigl[q^{3/8}\sum_{j=1,3,5,\cdots}q^{\frac{j^2}{56}}e^{-2\pi i \frac{a j}{7}} \Bigr]\bigg|_{q\rightarrow q^{-1}}
\end{align*}
where $(\ldots)\vert_{q \to 1/q}$ means that we should replace $q$ by $q^{-1}$ in all terms. At the level of the perturbative $\hbar$-expansion, this means replacing $\hbar \to - \hbar$. The Borel transform of the resulting expression is
$$
B \left( \frac{-i}{\sqrt{7k}}
\sum_{b=1, 3, 5, \cdots, 13}
e^{-2\pi i  \frac{ab}{7}}
(-1)^{\frac{b-1}{2}}
\Psi_{14}^{(b)} (-\frac{1}{k}) \right) (\xi)
= \frac{- (-7\pi\xi)^{-1/2}}{e^{-2\pi i \frac{a}{7}+\sqrt{-\pi i \xi/7}} + e^{2\pi i \frac{a}{7}-\sqrt{-\pi i \xi/7}}}
$$
It has poles at $-2\pi i \frac{a}{7}+\sqrt{-\pi i \xi/7} = \pi i(n + \frac{1}{2})$, or 
\be
\frac{\xi}{2\pi i} \; = \; -7\frac{(2n+1 + \frac{4}{7}a)^2}{8} \; \equiv \; \frac{1}{8}-\frac{2}{7}a^2 \text{ mod }\mathbb{Z}.
\label{52polesb}
\ee
Up to a somewhat mysterious shift by $\frac{1}{8}$, the values in \eqref{52polesa} and \eqref{52polesb} agree with the values of $\text{CS} (\alpha)$ listed in Table~\ref{tab:knots}. For the moment we don't understand the meaning of these poles or why there is a shift by $\frac{1}{8}$. We leave this as a question for the future.\footnote{One possible explanation for such ``spurious poles'' is that, while expanding the range of summation over colors to include $0$ and $k$, they can contribute to abelian and almost abelian part even if the total contribution is 0. As a small evidence in favor of this scenario, we note that the poles in the Borel transform indeed cancel out in the combination $Z_0^{(+)} - \frac{i}{\sqrt{7}} Z_0^{(-)}$, when one adds abelian and almost abelian contributions together. This, however, does not explain the poles in other $Z^{(-)}_{a}$, with $a \ne 0$, as well as the shift by $\frac{1}{8}$.}


\subsubsection{0-surgery on the unknot : $M_3 = S^2 \times S^1$}\label{S2xS1}

For the 0-surgery on the unknot \eqref{SoneStwo} the plumbing diagram consists of only one vertex:
\be
\overset{\displaystyle{0}}{\bullet}
\quad = \quad
S^2 \times S^1
\ee
A naive application of the integral formula \eqref{Zforplumbed} --- or, rather, its refinement~\cite{Gukov:2017kmk} --- to this plumbing graph gives the {\it reduced} invariant
\begin{multline}
\hat Z (q,t) \; = \; \frac{1}{2} \int \frac{dz}{2\pi iz} \frac{(z^2;q)_{\infty} (z^{-2};q)_{\infty}}{(z^2 tq;q)_{\infty} (z^{-2} tq;q)_{\infty}}
\; = \; \frac{(q^2t;q)_{\infty} (qt;q)_{\infty}}{(q^2t^2;q)_{\infty} (q;q)_{\infty}} \; = \;
\\
\; = \; 1 + (1-t)q + (2-3t+t^2)q^2
+ \ldots
\label{ZhatS1S2}
\end{multline}
where, to avoid clutter, we omit the label $a \in \text{Tor} H_1 (M_3, \Z) =0$. This expression is precisely the half-index ({\it i.e.} a partition function on a solid torus~\cite{Gadde:2013wq}) of a simple 3d $\CN=2$ theory with $G=SU(2)$:
\be
T[S^2 \times S^1,G]
\quad \stackrel{?}{=} \quad
\boxed{\begin{array}{c}
		~\text{3d } \CN=2 \text{ gauge theory with gauge } \\
		~\text{group } G \text{ and an adjoint chiral }\Phi_0
\end{array}}
\label{TS1S2}
\ee
which is one of the natural candidates for $T[S^2 \times S^1]$ and was already mentioned in the end of section~\ref{sec:Mflat}. 
Indeed, the denominator of the integrand corresponds to the contribution of bosonic $(\partial_z^n \phi_0)$ operators coming from the adjoint chiral multiplet $\Phi_0$, while the numerator corresponds to the contribution of the derivatives of the gaugino field $(D_z^n \lambda_{-})$ coming from the vector multiplet. After the integration we are left with gauge invariant combinations of these fields describing meson fields. In computing \eqref{ZhatS1S2}, the R-charge assignment \eqref{RPhi0} of $\Phi_0$ was used, instead of a perhaps more natural (for a $\CN=4$ vector multiplet) value $R=1$. The variable $t$ in \eqref{ZhatS1S2} is the fugacity for $U(1)_{\beta}$ flavor symmetry\footnote{{\it cf.} the general discussion in \cite[sec.3.4]{Gukov:2016gkn}.}
\begin{equation}
\begin{array}{c|c|c}
& ~~U(1)_R~~ & ~~U(1)_{\beta} : \Phi_0 \to e^{i \theta} \Phi_0~~ \\
\hline
\Phi_0~~ & +2 & +1
\end{array}\,
\label{RFlavorPhi0}
\end{equation}

The {\it unreduced} invariant (with Cartan components of the adjoint chiral and gauge multiplet included) corresponding to~\eqref{ZhatS1S2} is
\begin{align}
\hat Z^{\text{unred}} (q,t)
& \; = \; \frac{(q^2t;q)_{\infty}}{(q^2t^2;q)_{\infty}} \notag \\
\hat Z^{\text{unred}} (1/q,1/t)
& \; = \; \frac{(q^{-1} t^{-2};q)_{\infty}}{(q^{-1} t^{-1}; q)_{\infty}} \label{S1S2unred} \\
\hat Z^{\text{unred}} (1/q,t)
& \; = \; \frac{(q^{-1} t^2;q)_{\infty}}{(q^{-1} t;q)_{\infty}} \notag
\end{align}
where the last two expressions can be obtained either with the help of the ``cyclotomic expansion'' of \cite{Gukov:2017kmk} or by using the familiar relations $(x;q^{-1})_{\infty} = \frac{1}{(xq;q)_{\infty}}$ and $(x;q)_n = \frac{1}{(x q^n ; q)_{-n}}$. Then, according to \cite{Gukov:2017kmk}, these building blocks determine the superconformal index of the theory \eqref{TS1S2}:
\be
\CI (q,t)
\; = \; \hat Z^{\text{unred}} (q,t) \, \hat Z^{\text{unred}} (1/q,1/t)
\; = \; \frac{(q^2t;q)_{\infty}}{(q^2t^2;q)_{\infty}}
\frac{(q^{-1} t^{-2};q)_{\infty}}{(q^{-1} t^{-1}; q)_{\infty}}
\ee
and its topologically twisted index on $S^1 \times S^2$:
\be
\CI_{\text{top}} (q,t)
\; = \; \hat Z^{\text{unred}} (q,t) \, \hat Z^{\text{unred}} (1/q,t)
\; = \; \frac{(q^{-1} t^2;q)_3}{(q^{-1} t;q)_3}
\label{ItopS1S2S1S2}
\ee
that will be discussed further in the next section.

The refinement by $t$, the fugacity for the flavor symmetry $U(1)_{\beta}$,
is closely related to the categorification of the $q$-series invariants $\hat Z$,
see {\it e.g.} \cite{Gukov:2016gkn,Gukov:2017kmk}.
In the ``unrefined'' limit $t \to 1$, the difference between reduced and unreduced
versions of $\hat Z$ is essentially in the extra factors of $(q;q)_{\infty} = q^{-\frac{1}{24}} \eta (q)$.
Although keeping such extra factors in the unreduced version of the invariants can make
the expressions look bulkier, physically it is more natural since this is what
the half-index of 3d $\CN=2$ theory $T[M_3]$ actually computes.
Moreover, these extra factors of $(q;q)_{\infty} = q^{-\frac{1}{24}} \eta (q)$
play an important role in {\it 3d Modularity Conjecture} \cite{Cheng:2018vpl}
that relates $\hat Z (M_3;q)$ to characters of chiral algebras. In relation to the WRT invariants of 3-manifolds, however, and their categorification,
we usually use the reduced version of $\hat Z (M_3;q)$, which in the present paper
means no additional label (unlike the unreduced version $\hat Z^{\text{unred}}$).


\subsection{Continuous versus discrete labels}\label{contvsdiscr}

Throughout this paper we considered invariants $\hat Z_b (M_3)$ labeled by discrete
set of labels $b$ when $M_3$ is closed, even if $b_1 (M_3) > 0$.
The continuous variables $x_i$ only appeared at intermediate stages of building
closed $M_3$ from 3-manifolds with toral boundaries,
{\it cf.} surgery formulae \eqref{Zforplumbed} or \eqref{ZKnresidue}.
As in \cite{Gukov:2003na}, these $\C^*$-valued variables $x_i$ can be interpreted
as holonomy eigenvalues of complex $SL(2,\C)$ connections along ``meridian'' cycles
of the toral boundaries, and the gluing along $T^2$ boundaries gives rise to $x_i$-integrals.

However, for closed 3-manifolds with $b_1 > 0$, which are of central importance in this paper,
there are some clues suggesting that we might need to add $b_1$ continuous labels, such that
\be
\hat{Z}_b \; = \; \oint \prod_{1\leq i\leq b_1}\frac{dx_i}{2\pi i x_i} \hat{Z}_b(x_1,\cdots,x_{b_1})
\label{Zxintvsa}
\ee
Just as for link complements and more general 3-manifolds with toral boundaries considered earlier,
these additional variables $x_i \in \C^*$ should be viewed as holonomies of $SL(2,\C)$ flat connections on $M_3$.

Indeed, recall from section~\ref{sec:Mflat}, that the generic component of the moduli space of $SL(2,\C)$ flat connections on $M_3$ is $b_1$-dimensional. Similarly, in 3d $\CN=2$ theory $T[M_3]$, the continuous variables $x_i$ play the role of local coordinates on $\CM_{\text{vacua}}$; in particular, they parametrize the ``Coulomb branch'' of $T[M_3]$ and, from the viewpoint of resurgent analysis \cite{Gukov:2016njj}, appear on the same footing as discrete labels $b$ in $\hat Z_b$ and their close cousins (without a hat),
\be
Z_a \; = \; \sum_b \CS_{ab} \hat Z_b
\label{ZvsZhat}
\ee
that we already encountered a few times earlier.

Another reason why one may find natural to label $Z_a$ and $\hat Z_b$ by continuous variables $x_i$
is that their perturbative expansion in the limit $q = e^{\hbar} \to 1$ has the ``constant term'' (of the order $\hbar^0$) related to the Reidemeister-Turaev torsion \cite{Gukov:2016njj}. This property of $Z_a$ and $\hat Z_b$ is inherited directly from the analogous perturbative expansion of the Chern-Simons path integral \cite{witten1989quantum} which they repackage as a trans-series.\footnote{If ``$a$'' denotes an abelian flat connection, then the perturbative expansion in Chern-Simons theory with compact gauge group $G=SU(2)$ is identical to the perturbative expansion around this flat connection in Chern-Simons theory with complex gauge group $G_{\C} =SL(2,\C)$.}
And, since for 3-manifolds with $b_1 > 0$,
the torsion $\tau_{M_3}$ can be regarded as a function
of continuous variables $x_i$, $i = 1, \ldots, b_1$,
it is natural to expect $Z_a$ and $\hat Z_b$ to have this property too.
More precisely \cite{Turaev}, the torsion $\tau_{M_3}$ is a function of a $\text{Spin}^c$ structure $b$
with values in $Q(H)$, the ring of fractions of $\mathbb{Z}[H]$, where $H = H_1(M_3)$.
Elements of $Q(H)$, in turn, can be regarded as functions
\be
\text{Hom}(\text{Tor}(H),\mathbb{C}^*) \; \rightarrow \; \mathbb{C}(x_1,\cdots,x_{b_1}) \,.
\ee
For each character $\sigma_i \in \text{Hom}(\text{Tor}(H),\mathbb{C}^*)$, $\tau_{M_3} (b)(\sigma_i)\in F_i = C_{\sigma_i}(x_1,\cdots,x_{b_1})$ where $C_{\sigma_i}$ is the cyclomotic field induced by $\sigma_i$. Moreover, this map is $H$-equivariant, {\it i.e.} 
\be
\tau_{M_3} (h\cdot b)(\sigma_i) = \sigma_i(h)\tau_{M_3} (b)(\sigma_i) \,.
\ee
In the next section, we propose a precise relation between $\tau_{M_3}$ and the invariants $Z_a$ and $\hat Z_b$.


\section{Twisted indices and Hilbert spaces}
\label{sec:HF}

In general, a $d$-dimensional quantum field theory (QFT) with sufficient amount of supersymmetry that allows a partial topological twist~\cite{Bershadsky:1995vm} on an arbitrary $(d-1)$-manifold $M_{d-1}$ (perhaps equipped with an extra structure, {\it e.g.} Spin structure) admits a {\it topologically twisted index}, or {\it twisted index} for short,
\be
\CI_{\text{top}} \; = \; \Tr_{\CH (M_{d-1})} (-1)^F \; = \; \chi \left( \CH (M_{d-1}) \right)
\label{topindex}
\ee
where the space of supersymmetric states $\CH (M_{d-1})$ is a Floer-like homology of $M_{d-1}$. It can be defined as a $Q$-cohomology with respect to the supercharge $Q$ preserved by the partial topological twist on $\R \times M_{d-1}$. For example, in $d=4$ this gives a physical interpretation~\cite{Witten:1988ze} of the original Floer theory~\cite{MR956166} and a physical interpretation~\cite{Gukov:2016gkn} of the Heegaard Floer homology~\cite{MR2113019} that will be relevant in the rest of this section. When the QFT in question has global symmetries, the space of states $\CH (M_{d-1})$ is graded and, in such cases, one can introduced additional variables in the twisted index \eqref{topindex} that keep track of these additional gradings.

Equivalently, one can view the twisted index of a $d$-dimensional QFT as a partition function of a $(d-1)$-dimensional TQFT, with all Kaluza-Klein modes included. This approach to twisted indices and thinking about the underlying Floer-like homology $\CH (M_{d-1})$ is useful in both Lagrangian and non-Lagrangian theories. Even in Lagrangian theories, where \eqref{topindex} can often be computed by localization techniques, it offers useful insights into the structure of twisted indices: see \cite{Gukov:2016gkn,Dedushenko:2018bpp,Gukov:2018iiq} for the underlying algebraic structures in $d=3$, $d=4$, and $d=5$, respectively.

\subsection{Twisted Hilbert space on $F_g$}

In $d=3$, the first computation of the twisted index \eqref{topindex} and its application appeared in~\cite{Gukov:2015sna}. It is directly relevant to us here since the corresponding $3d$ $\CN=2$ theory is $T[M_3]$ with $M_3=L(k,1)$:

\be
T[L(k,1),G]
\quad = \quad
\overset{\displaystyle{k}}{\bullet}
\quad = \quad
\boxed{\begin{array}{c}
		~\text{3d } \CN=2 \text{ super-Chern-Simons} \\
		~\text{with } G_k \text{ and an adjoint chiral }\Phi_0
\end{array}}
\label{TLens}
\ee
where the charges (weights) of $\Phi_0$ under the R-symmetry and the flavor symmetry $U(1)_{\beta}$ are the same as in \eqref{RFlavorPhi0}. For a genus-$g$ Riemann surface $F_g$ (where ``$F$'' stands for ``fiber,'' as used in more general twisted compactifications), the twisted index of this theory on $S^1 \times F_g$ can be expressed as a $U(1)_{\beta}$-equivariant integral over the moduli space of Higgs bundles on~$F_g$,
\be
\CI_{\text{top}} (S^1 \times F_g) \; = \;
\int\limits_{\CM_H (F_g,G)} \text{Td} (\CM_H, \beta) \wedge e^{k \omega_I - k \beta \mu_I}
\ee
This integral makes sense even in the special case $k=0$, which corresponds to $M_3 = S^2 \times S^1$ and for which the naive specialization of the theory \eqref{TLens} leads to our earlier candidate~\eqref{TS1S2} for the theory $T[S^2 \times S^1]$. According to \cite{Gukov:2015sna,Nekrasov:2014xaa} (see also \cite{Benini:2015noa,Closset:2016arn,Closset:2017zgf,Gukov:2017zao}), twisted index of this theory, {\it i.e.} partition function on $S^1 \times F_g$, is given by
\be
Z_{T[S^2 \times S^1]} (S^1 \times F_g) \; = \;
\frac{1}{2} \sum_{\text{vacua} : \;  d \tilde \CW = 0} \,
Z_{\text{1-loop}} \vert_{{\frak m} = 0} \, \left( \tilde \CW'' \right)^{g-1}
\label{S1S2S1Fg}
\ee
where $\tilde \CW''$ denotes $\frac{\partial^2 \tilde \CW}{(\partial \log z)^2}$ and
$\tilde \CW' = \frac{\partial \tilde \CW}{\partial \log z} = \frac{ \partial \log Z_{\text{1-loop}} }{ \partial {\frak m} }$.
For the gauge theory \eqref{TS1S2} we have
$$
Z_{\text{1-loop}} \; = \;
\underbrace{ \left( 2 - \frac{1}{z^2} - z^2 \right)^{1-g} }_{SU(2) \text{ gauge}}
\underbrace{ \left( \frac{z t^{1/2}}{1 - z^2 t} \right)^{2 {\frak m} + g - 1}
	\left( \frac{t^{1/2}}{1 - t} \right)^{g - 1}
	\left( \frac{z^{-1} t^{1/2}}{1 - z^{-2} t } \right)^{ - 2 {\frak m} + g - 1} }_{\text{adjoint chiral } \Phi_0}
$$
The Bethe ansatz equation
\be
1 \; = \; \exp \left( \frac{\partial \tilde \CW}{\partial \log z} \right) \; = \;
\left( \frac{z^2 - t}{1 - z^2 t} \right)^2
\ee
has a total of four solutions $z = \{ \pm 1 , \pm i \}$, two of which (namely, $z = \pm 1$)
correspond to the points on the maximal torus of $G = SU(2)$ fixed by the Weyl group and, therefore, need to be discarded.
The sum over the other two roots of the Bethe ansatz equation gives \eqref{S1S2S1Fg}:
\be
Z_{T[S^2 \times S^1]} (S^1 \times F_g) \; = \;
\left( t^{1/2} + t^{-1/2} \right)^{3-3g}
\label{S2S1onFg}
\ee
which, for $g=0$, agrees with the limit $q \to 1$ of the topologically twisted index \eqref{ItopS1S2S1S2} computed earlier. And, for $g=1$, it also agrees with the Witten index $\text{Tr} (-1)^F$ of 3d $\CN=2$ adjoint SQCD with $N_f = 1$ \cite{Intriligator:2013lca}. (Note, the R-charge assignment does not play a role in that calculation.)

{}From the viewpoint of 3d $\CN=2$ super symmetryalgebra (see {\it e.g.} \cite{Intriligator:2013lca}),
\be
\begin{array}{rcl}
	\{ Q_{\pm} , \bar Q_{\pm} \} & = & P_1 \pm i P_2 \\
	\{ Q_{\pm} , \bar Q_{\mp} \} & = & - P_0
\end{array}
\qquad\qquad
\begin{array}{l@{\;}|@{\;}cccc}
	& ~Q_+~ & ~\bar Q_+~ & ~Q_-~ & ~\bar Q_-~ \\\hline
	U(1)_{E} & -1 & -1 & +1 & +1 \\
	U(1)_{R} & -1 & +1 & -1 & +1
\end{array}
\label{3dSUSY}
\ee
the twisted index $\CI_{\text{top}} (S^1 \times F_g) \equiv Z_{T[M_3]} (S^1 \times F_g)$ counts the states of 3d $\CN=2$ theory on $F_g$ that are in cohomology of the supercharge $Q_A = Q_- + \bar Q_+$. Indeed, partial topological twist along $F_g$ is a 3d version of the standard A-model twist in two dimensions, that replaces the little group $U(1)_E$ by the diagonal subgroup of $U(1)_E \times U(1)_R$. Under the latter, $Q_-$ and $\bar Q_+$ have zero spin, as one can easily see from \eqref{3dSUSY}.

In other words, as advertised in \eqref{topindex}, $\CI_{\text{top}} (S^1 \times F_g)$ computes the graded trace over the ``Floer homology'' $\CH (F_g)$ in 3d $\CN=2$ theory. Using modern terminology, one might call $\CH (F_g)$ a categorification of the A-model. A chiral multiplet of R-charge $R$ contributes to this $Q_A$-cohomology a boson $\phi$ and a fermion $\bar \psi$, which after the twist transform as holomorphic sections of $K_{F_g}^{R/2} \otimes \CL ({\frak m})$ and $K_{F_g}^{1-R/2} \otimes \CL (-{\frak m})$, respectively \cite{Bershadsky:1995vm,Gukov:2017zao,Dedushenko:2017osi,Bullimore:2018yyb}:
\be
\CH_{\text{chiral}} (F_g) \quad = \quad
\underbrace{ \CT^+ \; \otimes \quad \ldots \quad \otimes \; \CT^+ }_{\dim H^0 (K_{F_g}^{R/2} \otimes \CL ({\frak m}))}
\quad \otimes \quad
\underbrace{ \CF \; \otimes \quad \ldots \quad \otimes \; \CF }_{\dim H^0 (K_{F_g}^{1-R/2} \otimes \CL (-{\frak m}))}
\label{HchiralonFg}
\ee
Adopting notations from Heegaard Floer theory, here we denote a copy of bosonic Fock space by
\be
\CT^+ \; = \;
1 \oplus \phi \oplus \phi^2 \oplus \ldots
\; \cong \; H^*_{S^1} (\text{pt}) \; = \; H^* (\mathbb{C}{\bf P}^{\infty})
\label{Ttower}
\ee
and for a fermionic Fock space use
\be
\CF \; = \;
1 \oplus \bar \psi
\; \cong \; H^* (\mathbb{C}{\bf P}^1)
\label{FFockspace}
\ee

The Kaluza-Klein spectrum \eqref{HchiralonFg} can be easily computed directly from the topological reduction \cite{Bershadsky:1995vm} and admits a simple interpretation in the effective 1d quantum mechanics. Indeed, as explained below \eqref{3dSUSY}, the partial topological twist along $F_g$ breaks half of the supersymmetries, preserving only $Q_-$ and $\bar Q_+$ that transform as scalars on $F_g$. Moreover, it is also easy to see directly from \eqref{3dSUSY} that these two supercharges generate an algebra of $\CN=2$ superconformal quantum mechanics
\be
\{ Q , \bar Q \} \; = \; 2H
\,, \qquad
Q^2 = 0
\,, \qquad
\bar Q^2 = 0
\label{SQMalgebra}
\ee
which in the literature is often called $\CN=2B$ since it {\it also} appears in a reduction of a 2d $\CN=(0,2)$ theory on a circle, see {\it e.g.} \cite{Gibbons:1997iy,Okazaki:2015pfa}. In particular, on a circle of circumference $\beta \to 0$, the two standard matter multiplets in a 2d $\CN=(0,2)$ theory --- namely, chiral and Fermi --- become two types of supermultiplets in $\CN=2B$ quantum mechanics, usually called by the same names. Correspondingly, the indices of these supermultiplets in 1d quantum mechanics follow directly from the elliptic genera of 2d multiplets \cite{Gadde:2013wq}:
\be \begin{array}{l@{\qquad}cl}
	\text{2d $(0,2)$ Fermi:} & \frac{\theta (x,q)}{\eta (q)}
	& \xrightarrow[~~~~]{~~\beta \to 0~~} \quad x^{\frac{1}{2}} - x^{-\frac{1}{2}} \; = \; \chi \left( \CF \right) \\
	\text{2d $(0,2)$ chiral:} & \frac{\eta (q)}{\theta (x,q)}
	& \xrightarrow[~~~~]{~~\beta \to 0~~} \quad \frac{1}{x^{1/2} - x^{-1/2}} \; = \; \chi \left( \CT^+ \right)
\end{array}
\ee
where, in the last equality, we draw attention to the fact that these indinces match $\Tr (-1)^F x^{\text {flavor}}$ evaluated on \eqref{Ttower} and \eqref{FFockspace}. For this reason, the Floer homology \eqref{HchiralonFg} of A-twisted 3d $\CN=2$ theory on $F_g$ can also be interpreted as a field content of the effective $\CN=2B$ superconformal quantum mechanics, in such a way that each $\CT^+$ factor corresponds to a 1d chiral multiplet and each copy of $\CF$ corresponds to a 1d Fermi multiplet.

This language becomes especially convenient in describing the contribution of a 3d $\CN=2$ vector multiplet to the $Q_A$-cohomology $\CH (F_g)$. The naive analogue of~\eqref{HchiralonFg} reads
\be
\underbrace{ \CT^+ \; \otimes \quad \ldots \quad \otimes \; \CT^+ }_{g \cdot \dim G}
\quad \otimes \quad
\underbrace{ \CF \; \otimes \quad \ldots \quad \otimes \; \CF }_{\dim G}
\label{vectoronFgnaive}
\ee
where the bosonic and fermionic Fock spaces are generated by the zero-modes of the gauge connection and gluinos, respectively. However, only gauge-invariant combinations like mesons, baryons and glueballs can be part of $\CH (F_g)$, and this requires implementing a BRST-like reduction on the entire space of states, not just its gauge sector \eqref{vectoronFgnaive}. A convenient way to account for this is to interpret this operation as gauging in the effective 1d quantum mechanics, where \eqref{vectoronFgnaive} is simply a statement that topological reduction of a 3d $\CN=2$ vector multiplet on $F_g$ results in a 1d vector multiplet and $g$ copies of 1d chiral multiplet in the adjoint representation of $G$.

It is useful to note that the superconformal quantum mechanics \eqref{SQMalgebra} can be conveniently described in superspace $\R^{(1 \vert 2)}$ parametrized by two odd (Grassmann) coordinates $\theta, \bar \theta$ and the time coordinate $x^0$. Specifically, it is easy to check that the two supercharges\footnote{The covariant superspace derivatives
\be
D = i \frac{\partial}{\partial \theta} - \bar \theta \frac{\partial}{\partial x^0}
\,, \qquad
\bar D = i \frac{\partial}{\partial \bar \theta} - \theta \frac{\partial}{\partial x^0}
\ee
obey $\{ D , \bar D \} = -2i \partial_{x^0}$.}
\be
Q = i \frac{\partial}{\partial \theta} + \bar \theta \frac{\partial}{\partial x^0}
\,, \qquad
\bar Q = i \frac{\partial}{\partial \bar \theta} + \theta \frac{\partial}{\partial x^0}
\ee
satisfy the algebra \eqref{SQMalgebra} and can be distinguished by the R-charge
\be
[R, Q] \; = \; -Q
\,, \qquad
[R, \bar Q] \; = \; \bar Q
\,, \qquad
[R, H] \; = \; 0
\ee
which, in dimensional reduction from 2d $\CN=(0,2)$ algebra, simply comes from the reduction of 2d R-charge.

Note, in general, the space of states \eqref{HchiralonFg} is much larger than its contribution to the index, $\chi \Big( \CH_{\text{chiral}} \Big) = \left( x^{\frac{1}{2}} - x^{-\frac{1}{2}} \right)^{(R-1)(1-g) - {\frak m}}$, which appeared several times in the calculation of \eqref{S2S1onFg}.

\subsection{Twisted Hilbert space on $D^2$}

Compared to the Floer homology of $T[M_3]$ on $F_g$, the space of (supersymmetric) states $\CH (D^2)$ on a 2-disk $D^2$ has new interesting features:

\begin{itemize}

\item While $\CH (F_g)$ depends only on 3d theory, $\CH (D^2)$ depends on 3d theory {\it together} with a choice of 2d boundary condition $\CB$ at the disk boundary.

\item In the case of $\CH (F_g)$ it is natural to ask how this space depends on the genus $g$. In the case of $\CH (D^2, \CB)$, the analogous question involves introducing $\Z$-grading associated with $U(1)_E$ rotation symmetry of the disk and asking about graded components of~$\CH (D^2, \CB)$. As we explain shortly, this is equivalent to studying ``half-twisted'' theory~\cite{Silverstein:1995re,Katz:2004nn} (a.k.a. ``holomorphic twist'' \cite{Johansen:1994aw} of the theory) along 2d part of the 3d space-time or, equivalently, the Omega-background along $D^2$ (which, sometimes, is denoted $D^2_q$ or $\R^2_q$).

\end{itemize}

As for the choice of boundary conditions, following \cite{Gadde:2013wq,Gadde:2013sca}, we take $\CB$ to be invariant under 2d $\CN=(0,2)$ supersymmetry on the boundary. Then, several nice things happen. Perhaps the most important feature is that the hemispheres around local boundary operators (illustrated in Figure~\ref{fig:halfplane}) become twisted by $U(1)_R$ symmetry precisely as in the topological index \eqref{topindex}.
In other words, one can say that a holomorphic twist of a 3d $\CN=2$ theory along the boundary induces a topological A-twist along the 2-disks $D^2$ ``orthogonal'' to the boundary, {\it cf.} Figure~\ref{fig:halfplane}, and vice versa. From the viewpoint of the topological twist along $D^2$, the half-index of \cite{Gadde:2013wq} is simply a version of the twisted index \eqref{topindex} with $F_g = D^2$ refined\footnote{equivariant with respect to $U(1)_E$ in \eqref{3dSUSY}} by the rotation symmetry of the disk.

Conversely, the twisted index \eqref{topindex} with $F_g = D^2$ can be viewed as a limit of the $S^1 \times_q D^2$ partition function \eqref{halfindex}:
\be
\CI_{\text{top}} (S^1 \times D^2)
\; = \;
\lim_{q \to 1} \, \Tr_{\CH_{D^2}} (-1)^F q^{R/2 + J_3}
\label{q1limit}
\ee
Let us choose the conventions such that the $\CN=(0,2)$ supersymmetry preserved on the boundary is generated by $Q_+$ and $\bar Q_+$. According to \eqref{3dSUSY}, these satisfy the algebra \eqref{SQMalgebra} with the Hamiltonian $H_+$ that generates translations along the boundary. The states that contribute to the elliptic genus $\Tr (-1)^F q^{H_-}$ of the 2d $\CN=(0,2)$ boundary theory have $H_+ = 0$.
In the radial quantization, illustrated in Figure~\ref{fig:halfplane}, the solid torus $S^1 \times D^2$ is foliated by concentric disks embedded in the half-space $\R^2 \times \R_+$. At every point on the disk $D^2$, the anticommutator of the supercharges $Q$ and $\bar Q$ generates a translation orthogonal to the disk, which rotates as one goes from points at the boundary to the apex (center) of the disk, where $Q = Q_A$ and $\bar Q = Q_A^{\dagger}$. Put another way, translations orthogonal to the boundary of the half-space $\R^2 \times \R_+$ are exact in $\bar Q_+$-cohomology, and so is the anti-holomorphic dependence on $\bar z$ along the $\R^2$ parametrized by $z$ and $\bar z$.

In the context of 3d-3d correspondence, both sides of \eqref{q1limit} have a simple meaning. The right-hand side is basically our friend $\hat Z (M_3)$ in the limit $q \to 1$. The left-hand side, on the other hand, is a twisted partition function of $T[M_3]$ on $S^1 \times F_g$ with a special choice of $F_g = D^2$. For general $F_g$, this twisted partition function is basically determined \cite{Gukov:2016gkn} by the Reidemeister-Turaev torsion $\tau_{M_3}$, a close cousin of the Seiberg-Witten invariants $\text{SW}(M_3)$. More specifically, the information about all genus-$g$ twisted partition functions is completely contained in $\text{MTC} [M_3]$, in fact, in the $S$ and $T$ matrices:
\be
\CI_{\text{top}} (S^1 \times F_g) \; = \; \sum_{\lambda} \left( S_{0 \lambda} \right)^{\chi(F_g)}
\ee
The case most directly related to SW invariants / Turaev torsion has $\chi (F_g) = -1$, whereas the one in \eqref{q1limit} obviously has $\chi (D^2) = +1$.
Therefore, one should expect that \eqref{q1limit} is roughly inverse of the $\tau_{M_3}$.
This expectation agrees with the perturbative expansion of $Z_a$, related to $\hat Z_b$ via \eqref{ZvsZhat}, where $\tau_{M_3}$ originates from one-loop term in (complex) Chern-Simons theory.\footnote{In 3d $\CN=2$ theory $T[M_3]$, the relation between $Z_a$ and $\hat Z_b = \sum_{a} \CS_{ab} Z_a$ can be understood in terms of exchanging the roles of A- and B-cycles of $T^2 = \partial (S^1 \times D^2)$, see \cite{Gukov:2016gkn} for details.}

Moreover, taking the limit $q \to 1$ directly in the integral formula \eqref{Zforplumbed} for the plumbed manifolds leads to
\be
\sum_{m_i \in \Z} \int \prod_{i \in \text{vertices}}
\frac{dx_i}{2\pi i x_i} x_i^{\sum_j Q^{ij} m_j}
(1 - x_i)^{2 - \text{deg} (i)} x_i^{b_i}
\label{torintsum}
\ee
where $b \in \text{coker} \, Q$ can be interpreted as the label of a Wilson line (that runs orthogonal to the boundary of the half-space $\R^2 \times \R_+$ shown in Figure~\ref{fig:halfplane}). Then, following the same manipulations as in \cite{Gukov:2016gkn}, one can evaluate this infinite sum and the integral. The result is a finite sum over Bethe vacua, solutions to $\prod_j x_j^{Q^{ij}} = 1$, which can be identified with characters $\sigma \in \hat H = \text{Hom} (H,\C^*)$ in the Pontryagin dual of $H = H_1 (M_3)$. This sum has exactly the same form as what one finds for the Turaev torsion of plumbed 3-manifolds \cite{Nicolaescu}, except the exponent $\text{deg} (i) - 2$ is replaced by $2 - \text{deg} (i)$ in \eqref{torintsum}.

Therefore, we can summarize this discussion by saying that the Fourier transform of $\tau_{M_3} (b)$, viewed as a function of $\hat H$, is the inverse of \eqref{q1limit}. Equivalently, since the {\it unfolded} $S$-matrix $\CS_{ab}$ is independent of $q$ and implements the Fourier transform with respect to $b \in H$ and $\sigma \in \hat H$, we can state this as a relation between $\tau_{M_3} (b)$ and $Z (M_3)$:
\begin{conj}
\begin{equation}
\frac{1}{\tau_{M_3} (b)(a)} \; = \; e^{2\pi i(b,Q^{-1}a)} Z_a \bigg|_{q \rightarrow 1}
\label{Zvstorsion}
\end{equation}
\end{conj}
\noindent
The prefactor on the right-hand side of this relation is introduced to make $H$-equivariance of the Turaev's refined torsion manifest, {\it cf.} the discussion in the end of section~\ref{sec:Zhat}.
Also, note that, since the set of labels $a$ and $b$ on $Z$ and $\hat Z$ is {\it folded} with respect
to the Weyl symmetry $\Z_2 : a \to - a$, one must either implicitly
undo this folding on the right-hand side of \eqref{Zvstorsion} or average over the $\Z_2$ orbit
on the left-hand side.

For example, we can check \eqref{Zvstorsion} for knot complements. In this case, $Z$ and $\hat Z$ coincide, and the limit \eqref{q1limit} gives
\be
\frac{x^{1/2} - x^{-1/2}}{\Delta_K (x)}
\ee
This is precisely $\frac{1}{\tau_{M_3} (b)}$ for $M_3 = S^3 \setminus K$, up to an overall factor $x^b$ that can be interpreted as a defect (Wilson line) labeled by $b$.

A similar case of 0-surgeries on knots suggests that one may want to add a continuous label $x \in \C^*$ to $Z (S^3_0 (K))$ and $\hat Z (S^3_0 (K))$, such that expressions discussed in this paper are obtained by integrating over $x$, as in \eqref{Zxintvsa}. For example, for $M_3 = S^2 \times S^1 = S^3_0 (\text{unknot})$, we have $\tau_{S^2 \times S^1} (x) = \frac{1}{(x^{1/2} - x^{-1/2})^2}$. This is precisely the inverse of the integrand in \eqref{ZhatS1S2}, in the unrefined limit $t \to 1$.
It would be interesting to study the relation \eqref{Zvstorsion} further, in particular by checking it for other 3-manifolds.

\subsection{Twisted Hilbert space on $D^2$ with impurity}

The intriguing relation between $\tau_{M_3}$ and $\hat Z (M_3)$ discussed so far is based on \eqref{q1limit}, which involves partition function (or space of states) on $D^2 \, \cong \, S^2 \setminus \text{pt}$ and $G=SU(2)$.
It is instructive to compare this with a much more familiar appearance of $\tau_{M_3}$ in a similar problem \cite{Gukov:2016gkn} that involves $G=U(1)$ and an extra impurity $\CS_+$ on $D^2 \, \cong \, S^2 \setminus \text{pt}$ or, equivalently, the space of states on $S^2 \setminus \{ p_1,p_2 \}$ or, better yet, $S^2 \setminus \{ p_1,p_2,p_3 \}$.

Let is consider our basic example \eqref{SoneStwo} of a genus-0 mapping torus, $M_3 = S^2 \times S^1$, which also is a 0-surgery on the unknot. Unlike its non-abelian counterpart \eqref{TS1S2}, the abelian version of the theory $T[M_3,G]$ with $G=U(1)$ in this case is completely clear. It is simply a 3d $\CN=2$ super-Maxwell theory and a free chiral multiplet $\Phi_0$. The impurity $\CS_+$ is a hypermultiplet charged under $G=U(1)$ localized at a point on $D^2$ or, in $\CN=2$ language, a pair of localized chiral multiplets with charges $+1$ and $-1$. Therefore, the contribution of $\CS_+$ to the twisted Hilbert space consists of two copies of \eqref{HchiralonFg}, and the total index is
\be
\underbrace{\sum_{h \in \Z}
\int \frac{dx}{2\pi i x} q^{-h}}_{T[M_3,U(1)]} x^{h} \underbrace{\frac{t}{(1-xt)(1-x^{-1}t)}}_{\text{impurity}~\CS_+}
\label{S1S2Alex}
\ee
Note that the integrand at $t=1$, {\it i.e.} the contribution of the impurity $\CS_+$ in this expression, can also be interpreted as a contribution of a single vertex in a plumbing graph with Euler number 0,
\be
\overset{\displaystyle{0}}{\circ}
\quad = \quad
\frac{x}{(1-x)^2}
\label{TTvertex}
\ee
Here, following the conventions of \cite{Gukov:2019mnk}, we denote by ``$\circ$'' an {\it unintegrated} vertex / unfilled torus boundary (as opposed to ``$\bullet$,'' which denotes Dehn filled torus boundary / {\it integrated} vertex). In these conventions, the complement of the trefoil knot in $S^3$ evaluates to
\be
\begin{array}{cccc}
& \overset{\displaystyle{-2}}{\bullet} & \\
& \vline & \\
\overset{\displaystyle{-3}}{\bullet}
\frac{\phantom{xxxx}}{\phantom{xxxx}}
& \underset{\displaystyle{-1}}{\bullet} &
\frac{\phantom{xxxx}}{\phantom{xxxx}}
\overset{\displaystyle{0}}{\circ}
\end{array}
\qquad = \qquad
\frac{x \cdot \Delta_{{\bf 3_1}} (x)}{(1-x)^2}
\label{TTtref}
\ee
As in \eqref{S1S2Alex}, performing a Dehn filling that yields 0-surgery means integrating over $x$ and summing over Spin$^c$ structures $h$. This sum and the integral have the combined effect of replacing $x$ by $q$:
\be
\sum_{h \in \Z} \int \frac{dx}{2\pi i x} q^{-h} x^{h}
\frac{\Delta_K (x)}{(1-x)(1-x^{-1})}
\; = \; \frac{\Delta_K (q)}{(1-q) (1-q^{-1})}
\ee
This is indeed the Turaev-Milnor torsion of $M_3 = S^3_0 (K)$. Note, since the Alexander polynomial does not distinguish mirror knots, $\Delta_K (x) = \Delta_{\bar K} (x)$, 3-manifolds $M_3 = S^3_0 (K)$ and $M_3 = S^3_0 (\bar K)$ have the same Turaev torsion, which in turn equals the Euler characteristic of the Heegaard Floer homology $HF^+$. As illustrated in Table~\ref{tab:HFsurg}, though, the Heegaard Floer homology groups themselves can be quite different. From the viewpoint of 4d TQFT, this can be understood as a consequence of $M_4 \cong M_4'$, where $M_4 = S^1 \times M_3 = S^1 \times (S^3 \setminus K)$ and $M_4' = S^1 \times M_3' = S^1 \times (S^3 \setminus \bar K)$, which holds even when $M_3 \not\cong M_3'$.

\begin{table}[htb]
	\centering
	\renewcommand{\arraystretch}{1.3}
	\begin{tabular}{|@{~~}c@{~~}|@{~~}c@{~~}|@{~~}c@{~~}|@{~~}c@{~~}|}
		\hline {\bf knot}~$K$  & ~~$HF^+ \left( S^3_{1/r} (K) \right)$ & ~~$HF^+ \left( S^3_{-1/r} (K) \right)$ & ~~$HF^+ \left( S^3_{0} (K) \right)$
		\\
		\hline
		\hline
		${\bf 3_1^r}$ & $\CT^+_{-2} \oplus \Z_{-2}^{r-1}$ & $\CT^+_{0} \oplus \Z_{-1}^{r}$ & $\CT^+_{-1/2} \oplus \CT^+_{-3/2}$ \\
		\hline
		${\bf 3_1^{\ell}}$ & $\CT^+_{0} \oplus \Z_{0}^{r}$ & $\CT^+_{2} \oplus \Z_{1}^{r-1}$ & $\CT^+_{3/2} \oplus \CT^+_{1/2}$ \\
		\hline
		${\bf 4_1}$ & $\CT^+_0 \oplus \Z_{-1}^r$ & $\CT^+_0 \oplus \Z_{0}^r$ & ~~$\CT^+_{1/2} \oplus \CT^+_{-1/2} \oplus \Z_{-1/2}$ \\
		\hline
	\end{tabular}
	\caption{Heegaard Floer homology for $1/r$-surgeries  and $0$-surgeries behaves ``discontinuously'' as $1/r \to 0$, {\it i.e.} $HF^+ \left( S^3_{0} (K) \right)$ can not be viewed as a limit of $HF^+ \left( S^3_{1/r} (K) \right)$ with~$r \to \infty$.}
	\label{tab:HFsurg}
\end{table}

As in the case of $q$-series invariants $\hat Z (M_3)$ computed by the half-index of $T[M_3]$, integrating over $x$ in \eqref{TTvertex}--\eqref{TTtref} or, equivalently, filling all the hollow vertices in the plumbing diagram corresponds to performing a surgery.
For example, the following family of small surgeries on the trefoil knot\footnote{Our orientation conventions agree with \cite{MR1957829,MR2461862,MR3591644,Gukov:2019mnk}. Namely, $S^3_{+1} ({\bf 3_1^r}) = - \Sigma (2,3,5)$.}
\be
S^3_{-1/r} ({\bf 3_1^{\ell}})
\; = \; - S^3_{1/r} ({\bf 3_1^r})
\; = \; \Sigma (2,3,6r-1)
\ee
can be represented by a negative-definite plumbing graph
$$
\begin{array}{ccccccccc}
& & \overset{\displaystyle{-2}}{\bullet} & & & & & & \\
& & \vline & & & & & & \\
\overset{\displaystyle{-2}}{\bullet}
\frac{\phantom{xxxx}}{\phantom{xxxx}}
& \overset{\displaystyle{-2}}{\bullet}
\frac{\phantom{xxxx}}{\phantom{xxxx}}
& \underset{\displaystyle{-2}}{\bullet} &
\frac{\phantom{xxxx}}{\phantom{xxxx}}
\overset{\displaystyle{-2}}{\bullet} &
\frac{\phantom{xxxx}}{\phantom{xxxx}}
\overset{\displaystyle{-2}}{\bullet} &
\frac{\phantom{xxxx}}{\phantom{xxxx}}
\overset{\displaystyle{-2}}{\bullet} &
\frac{\phantom{xxxx}}{\phantom{xxxx}}
\overset{\displaystyle{-2}}{\bullet} &
\frac{\phantom{xxxx}}{\phantom{xxxx}}
\overset{\displaystyle{-3}}{\bullet} &
\frac{\phantom{xxxx}}{\phantom{xxxx}}
\underbrace{\overset{\displaystyle{-2}}{\bullet}
	\frac{\phantom{xxxx}}{\phantom{xxxx}}
	~\cdots 
	\frac{\phantom{xxxx}}{\phantom{xxxx}}
	\overset{\displaystyle{-2}}{\bullet}}_{r-2~\text{times}}
\end{array}
$$
Its Heegaard Floer homology, summarized in Table~\ref{tab:HFsurg}, is the twisted Hilbert space of $T[M_3,U(1)]$ on a disk $D^2$ with impurity $\CS_+$ representing a charged hypermultiplet localized at a point on $D^2$. In particular, it has the ``correction term'' (a.k.a. $d$-invariant) $\Delta = 2$ and $HF_{\text{red}} \big( S^3_{-1/r} ({\bf 3_1^{\ell}}) \big) \cong \Z^{r-1}$ in degree~1, so that the (regularized) Euler characteristic in this case is $\lambda (M_3) + \frac{\Delta (M_3)}{2} = - r +1$, where we used the Casson's surgery formula to compute $\lambda \left( S^3_{-1/r} ({\bf 3_1^{\ell}}) \right) = -r$, see {\it e.g.}~\cite{MR2172483}.

Tracing the surgery exact sequence \cite{MR1957829,MR2113020}, one can see that $HF^+ \left( S^3_{0} (K) \right)$ contains extra terms which are in the kernel of maps to $HF^+ \big( S^3_{1/r} (K) \big)$ for any finite $r$. This is similar to how $\hat Z \big( S^3_{0} (K) \big)$ compares to $\hat Z \big( S^3_{1/r} (K) \big)$. Based on this behavior, it is natural to conjecture that homology $\CH (M_3)$ categorifying $\hat Z (M_3)$ enjoys a surgery exact sequence:
\begin{conj}
	If $K \subset Y$ is a knot in an integral homology 3-sphere $Y$, then we have a long exact sequence
	\be
	\cdots \; \longrightarrow \; \CH \left( Y \right) \; \longrightarrow \; \CH \left( Y_0 (K) \right) \; \longrightarrow \; \CH \left( Y_{+1} (K) \right) \; \longrightarrow \; \CH \left( Y \right) \; \longrightarrow \; \cdots
	\ee
\end{conj}


\section{Generalizations and future directions}
\label{sec:discuss}

The results presented here lead to a variety of natural generalizations and questions for further study:

\begin{itemize}

\item {\bf 3d $\CN=2$ gauge theories with non-linear matter:} In some of our examples, we found it useful to describe $T[M_3]$ as a gauge theory with chiral matter multiplets valued in the complex group manifold $G_{\C}$. This might be a promising avenue for constructing $T[M_3]$ in general as well as exploring 3d $\CN=2$ physics on its own right. It plays an important role for describing the full set of vacua \eqref{Mvacua} based on the {\it multiplicative}, not {\it additive}, version of the character variety \eqref{Mflat}. Pursuing this direction may make contact with interesting recent work \cite{Elliott:2018yqm}, which might also be relevant to higher-genus mapping tori and Heegaard splittings.

\item {\bf $\CN=2$ S and T walls:}
As discussed in section~\ref{sec:Mflat} (and illustrated in Figure~\ref{fig:BPSwall}), the S and T walls relevant to $T[M_3]$ for genus-1 mapping tori are $\frac{1}{4}$-BPS (not $\frac{1}{2}$-BPS), {\it i.e.} they preserve only 3d $\CN=2$ supersymmetry.\footnote{Recall one of the lessons from section~\ref{sec:Mflat}: ``when in doubt, think of $G=U(1)$.''}
It would be interesting to understand such $\CN=2$ walls and, possibly, make contact with interesting recent work \cite{Eckhard:2019jgg}, where a particular choice was used. Arrangements of walls should obey equivalence relations (Kirby moves) represented by either known or new dualities of 3d $\CN=2$ theories.

\item {\bf Heegaard boundary conditions:} Close cousins of walls and interfaces are 3d $\CN=2$ boundary conditions associated with handlebodies in a Heegaard decomposition of $M_3$. It would be interesting to understand such boundary conditions, especially in higher genus, and produce a way to compute $\hat Z (M_3)$ using this approach. For example, for $g=2$ one should be able to reproduce the answer for the Poincar\'e sphere.

\item {\bf Bethe vacua and Coulomb branch superpotential:} In Kaluza-Klein reduction on $M_3$, one can think of 3d $\CN=2$ theories with infinitely many fields that correspond to {\it all} $G_{\C}$ connections on $M_3$. Integrating out most of these fields, one finds a collection of SCFTs with finitely many fields, namely $T[M_3]$. On its ``Coulomb branch,'' however, one could study the effective twisted superpotential function $\tilde \CW_{\text{eff}} (x_i)$, obtained by integrating out matter fields. The critical points of this function, called Bethe vacua, are related \cite{Gukov:2016gkn} to simple objects in the category of line operators in 3d $\CN=2$ theory $T[M_3]$. It would be interesting to study this aspect of 3d-3d correspondence further and, in particular, to understand {\it almost abelian flat connections} from this perspective.

\item {\bf General 3-manifolds with $b_1=1$:}
Any 3-manifold with $b_1=1$ can be obtained by a 0-surgery on a null-homologous knot in a rational homology sphere. In particular, the Ohtsuki's perturbative series \cite{Ohtsuki} makes sense. We expect that it is possible to resum the perturbative series into a $q^{1/r}$-series for some $r$. This should be a natural starting point to study $\hat{Z}$ for more general 3-manifolds with $b_1=1$, such as higher genus mapping tori with $b_1=1$.

\item {\bf Borromean rings and $M_3 = T^3$:}
0-surgeries on double twist knots studied in section~\ref{sec:Zhat} are special cases of surgeries on the Borromean rings. This larger class of examples is the next natural family to consider; it includes a 3-torus $M_3 = T^3$ as a prominent member, given by a surgery on the Borromean rings woth all surgery coefficients equal to 0. This example is especially interesting since, currently, there is no robust proposal for what $\hat Z_b (T^3)$ should be.

\item {\bf DAHA and toroidal algebras:} It is well known \cite{witten1989quantum,reshetikhin1991invariants,Elitzur:1989nr} that Chern-Simons TQFT with compact gauge group $G$ has finitely many states on $\Sigma = T^2$ which are in one-to-one correspondence with integrable representations of the affine Kac-Moody algebra at level $k$. For example, in $SU(2)_k$ Chern-Simons theory, gluing along $T^2$ involves summing over $k+1$ states. More generally, the sum runs over elements of the weight lattice $\Lambda$, with level $k$ playing the role of a ``cut-off.'' This has to be compared with the infinite-dimensional space of states in 3d TQFT $\hat Z$ on a torus, $\Sigma = T^2$. In fact, these states can be labeled by elements of {\it two} copies of the lattice, modulo the Weyl group. The fact that there is no cut-off is not surprising since the TQFT $\hat Z$ is supposed to be the ``complex Chern-Simons theory,'' whose Hilbert space on $\Sigma$ is a quantization of a (non-compact!) phase space $\CM_{\text{flat}} (\Sigma, G_{\C})$. However, the fact that gluing along $\Sigma=T^2$ involves summing over elements of $\frac{\Lambda \times \Lambda^{\vee}}{\text{Weyl}}$ suggests that the role of the affine Kac-Moody algebra is replaced by a toroidal algebra of some sort, possibly by a double affine Hecke algebra (DAHA). Further evidence for the latter can be inferred from the Hilbert space of 6d fivebrane theory on $\R \times \Sigma \times S^1 \times_q D^2$ which, for $\Sigma = T^2$, turns out to be an infinite-dimensional representation of spherical DAHA (called the functional representation).

\item {\bf A category associated to $\Sigma$:}
One of the main motivations for studying the $\hat Z$-TQFT is that, based on its physical origin, it is expected to have a categorification, {\it i.e.} a 4d TQFT that associates graded vector spaces to 3-manifolds and a category $\CC_{\Sigma}$ to $\Sigma$. Describing this category, either algebraically or geometrically, would be a major step toward constructing 4d TQFT categorifying $\hat Z$. Note, for $\Sigma = T^2$, the Grothendieck group of this category is the infinite-dimensional space of states on a 2-torus discussed in the previous bullet point. In the language of 3d-3d correspondence, it asserts
\be
\text{Gr} (\CC_{\Sigma}) \; = \; \CH_{T[\Sigma \times S^1]} (D^2)
\ee
More generally, a homological invariant of a mapping torus \eqref{mappingtori} in this 4d TQFT should be given by a categorical trace, or $\Tr_{\varphi} \, \CC_{\Sigma} = \CH_{T[M_3]} (D^2)$.

\item {\bf 3d Modularity:}
To shed light on the algebraic structures mentioned in the previous two bullet points, it would be useful to know what kind of functions $\hat Z_a (M_3;q)$ are, either in general or for particular classes of 3-manifolds.
For 3-manifolds with $b_1=0$, an important step in this direction was recently made in \cite{Bringmann:2018ddv,BringmannMM}.
A natural question, then, is whether a similar analysis can be carried out in the case of 3-manifolds with $b_1 \ge 1$ considered here.

\end{itemize}

\acknowledgments{It is pleasure to thank Ian Agol, Francesco Benini, Miranda Cheng, Francesca Ferrari, Michael Freedman, Sarah Harrison, Jeremy Lovejoy, Ciprian Manolescu, Satoshi Nawata, Du Pei, Pavel Putrov, Larry Rolen, Nathan Seiberg, Cumrun Vafa, and Christian Zickert for help and suggestions.
The work of S.C. was supported by the US Department of Energy under grant DE-SC0010008. The work of S.G. is supported by the U.S. Department of Energy, Office of Science, Office of High Energy Physics, under Award No. DE-SC0011632, and by the National Science Foundation under Grant No. NSF DMS 1664240. The work of S.P. is supported by Kwanjeong Educational Foundation. N.S. gratefully acknowledges the support of the Dominic Orr Graduate Fellowship at Caltech.}

\bibliography{3d3d}
\bibliographystyle{JHEP}
\end{document}